\documentclass[11pt, letterpaper,  onecolumn]{IEEEtran}

\usepackage[mathscr]{eucal}
\usepackage[cmex10]{amsmath}
\usepackage{epsfig,epsf,psfrag}
\usepackage{amssymb,amsmath,amsthm,amsfonts,latexsym}
\usepackage{amsmath,graphicx,bm,xcolor,url}
\usepackage[caption=false]{subfig} 
\usepackage{fixltx2e}%ordering of single and double column floats
\usepackage{array}%array and tabular environments
\usepackage{verbatim}
\usepackage{bm}
\usepackage{algorithmic, cite}
\usepackage{algorithm}
\usepackage{verbatim}
\usepackage{textcomp}
\usepackage{mathrsfs}
\usepackage{epstopdf}
\catcode`~=11 \def\UrlSpecials{\do\~{\kern -.15em\lower .7ex\hbox{~}\kern .04em}} \catcode`~=13 

\allowdisplaybreaks[4]
 
\newcommand{\nn}{\nonumber}

\newcommand{\calA}{\mathcal{A}}
\newcommand{\calB}{\mathcal{B}}
\newcommand{\calC}{\mathcal{C}}

\newcommand{\calE}{\mathcal{E}}
\newcommand{\calF}{\mathcal{F}}
\newcommand{\calG}{\mathcal{G}}

\newcommand{\calM}{\mathcal{M}}
\newcommand{\calN}{\mathcal{N}}

\newcommand{\calQ}{\mathcal{Q}}

\newcommand{\calS}{\mathcal{S}}
\newcommand{\calT}{\mathcal{T}}
\newcommand{\calU}{\mathcal{U}}

\newcommand{\calX}{\mathcal{X}}
\newcommand{\calY}{\mathcal{Y}}

\newcommand{\ba}{\mathbf{a}}
\newcommand{\bA}{\mathbf{A}}

\newcommand{\bB}{\mathbf{B}}
\newcommand{\bc}{\mathbf{c}}

\newcommand{\bg}{\mathbf{g}}
\newcommand{\bG}{\mathbf{G}}
\newcommand{\bh}{\mathbf{h}}
\newcommand{\bH}{\mathbf{H}}
\newcommand{\bi}{\mathbf{i}}
\newcommand{\bI}{\mathbf{I}}

\newcommand{\bJ}{\mathbf{J}}

\newcommand{\bL}{\mathbf{L}}
\newcommand{\boldm}{\mathbf{m}}
\newcommand{\bM}{\mathbf{M}}

\newcommand{\bR}{\mathbf{R}}
\newcommand{\bs}{\mathbf{s}}
\newcommand{\bS}{\mathbf{S}}

\newcommand{\bT}{\mathbf{T}}
\newcommand{\bu}{\mathbf{u}}
\newcommand{\bU}{\mathbf{U}}
\newcommand{\bv}{\mathbf{v}}
\newcommand{\bV}{\mathbf{V}}

\newcommand{\bx}{\mathbf{x}}
\newcommand{\bX}{\mathbf{X}}
\newcommand{\by}{\mathbf{y}}

\newcommand{\bz}{\mathbf{z}}
\newcommand{\bZ}{\mathbf{Z}}

\newcommand{\rmQ}{\mathrm{Q}}

\newcommand{\bbE}{\mathbb{E}}
\newcommand{\bbF}{\mathbb{F}}

\newcommand{\bbN}{\mathbb{N}}

\newcommand{\bbR}{\mathbb{R}}

\newcommand{\frakC}{\mathfrak{C}}

\newcommand{\scB}{\mathscr{B}}
\newcommand{\scC}{\mathscr{C}}
\newcommand{\scD}{\mathscr{D}}

\newcommand{\scP}{\mathscr{P}}

\newcommand{\scR}{\mathscr{R}}
\newcommand{\scS}{\mathscr{S}}
\newcommand{\scT}{\mathscr{T}}

\newcommand{\scV}{\mathscr{V}}

\DeclareMathAlphabet{\mathbsf}{OT1}{cmss}{bx}{n}
\DeclareMathAlphabet{\mathssf}{OT1}{cmss}{m}{sl}% slanted sans serif

\newcommand{\rvE}{\mathsf{E}}

\newcommand{\rvP}{\mathsf{P}}

\DeclareSymbolFont{bsfletters}{OT1}{cmss}{bx}{n}  
\DeclareSymbolFont{ssfletters}{OT1}{cmss}{m}{n}
\DeclareMathSymbol{\bsfGamma}{0}{bsfletters}{'000}
\DeclareMathSymbol{\ssfGamma}{0}{ssfletters}{'000}
\DeclareMathSymbol{\bsfDelta}{0}{bsfletters}{'001}
\DeclareMathSymbol{\ssfDelta}{0}{ssfletters}{'001}
\DeclareMathSymbol{\bsfTheta}{0}{bsfletters}{'002}
\DeclareMathSymbol{\ssfTheta}{0}{ssfletters}{'002}
\DeclareMathSymbol{\bsfLambda}{0}{bsfletters}{'003}
\DeclareMathSymbol{\ssfLambda}{0}{ssfletters}{'003}
\DeclareMathSymbol{\bsfXi}{0}{bsfletters}{'004}
\DeclareMathSymbol{\ssfXi}{0}{ssfletters}{'004}
\DeclareMathSymbol{\bsfPi}{0}{bsfletters}{'005}
\DeclareMathSymbol{\ssfPi}{0}{ssfletters}{'005}
\DeclareMathSymbol{\bsfSigma}{0}{bsfletters}{'006}
\DeclareMathSymbol{\ssfSigma}{0}{ssfletters}{'006}
\DeclareMathSymbol{\bsfUpsilon}{0}{bsfletters}{'007}
\DeclareMathSymbol{\ssfUpsilon}{0}{ssfletters}{'007}
\DeclareMathSymbol{\bsfPhi}{0}{bsfletters}{'010}
\DeclareMathSymbol{\ssfPhi}{0}{ssfletters}{'010}
\DeclareMathSymbol{\bsfPsi}{0}{bsfletters}{'011}
\DeclareMathSymbol{\ssfPsi}{0}{ssfletters}{'011}
\DeclareMathSymbol{\bsfOmega}{0}{bsfletters}{'012}
\DeclareMathSymbol{\ssfOmega}{0}{ssfletters}{'012}

\newcommand{\hatH}{\hat{H}}

\newcommand{\hatbH}{\hat{\bH}}

\newcommand{\hatI}{\hat{I}}

\newcommand{\hatbI}{\hat{\bI}}

\newcommand{\hatbJ}{\hat{\bJ}}

\newcommand{\hatm}{\hat{m}}
\newcommand{\hatM}{\hat{M}}
\newcommand{\tilm}{\tilde{m}}

\newcommand{\tilQ}{\tilde{Q}}

\newcommand{\tilU}{\tilde{U}}

\newcommand{\tilbU}{\tilde{\bU}}

\newcommand{\hatx}{\hat{x}}
\newcommand{\hatX}{\hat{X}}
\newcommand{\tilx}{\tilde{x}}
\newcommand{\tilX}{\tilde{X}}

\newcommand{\tilY}{\tilde{Y}}

\newcommand{\tilbz}{\tilde{\bz}}
\newcommand{\tilbZ}{\tilde{\bZ}}

\newcommand{\bLambda}{\bm{\Lambda}}

\newcommand{\bDelta}{\bm{\Delta}}

\def\fndot{\, \cdot \,}

\newcommand{\ceil}[1]{\lceil{#1}\rceil}

\newcommand{\Pen}{P_{\mathrm{e}}^{(n)}}

\DeclareMathOperator*{\argmax}{arg\,max}

\DeclareMathOperator{\var}{\mathsf{Var}}

\DeclareMathOperator{\cov}{\mathsf{Cov}}

\DeclareMathOperator{\rank}{rank}

\newcommand{\Ber}{\mathrm{Bern}}

\newcommand{\bzero}{\mathbf{0}}
\newcommand{\bone}{\mathbf{1}}

\newtheorem{theorem}{Theorem} 
\newtheorem{lemma}[theorem]{Lemma}

\newtheorem{corollary}[theorem]{Corollary}
\newtheorem{definition}{Definition}

\newcommand{\qednew}{\nobreak \ifvmode \relax \else
      \ifdim\lastskip<1.5em \hskip-\lastskip
      \hskip1.5em plus0em minus0.5em \fi \nobreak
      \vrule height0.75em width0.5em depth0.25em\fi}

\usepackage{amsmath}
\usepackage{overpic}
\usepackage{cite}
\allowdisplaybreaks[1]
\title{On the Dispersions  of Three  Network Information Theory Problems}
\author{Vincent Y.~F.\ Tan$^*$,~\IEEEmembership{Member,~IEEE}  \hspace{.025in}  and \hspace{.025in}     Oliver Kosut$^{\dagger}$,~\IEEEmembership{Member,~IEEE}  \thanks{$^*$   Department of Electrical and Computer Engineering (ECE), National University of Singapore (NUS)   (Email:
  {vtan@nus.edu.sg})} \thanks{$^{\dagger}$ School of Electrical, Computer and Energy Engineering, Arizona State University (Email:  okosut@asu.edu).   } \thanks{This paper was presented in part at the Information Theory and Applications workshop in San Diego, CA in January 2012, the Conference on Information Sciences and Systems in Princeton, NJ in March 2012 and the International Symposium on Information Theory in Cambridge, MA in July 2012. }}

\begin{document}
\flushbottom
\maketitle

\begin{abstract}
We analyze the dispersions of distributed lossless source coding (the Slepian-Wolf problem), the multiple-access channel and the asymmetric broadcast channel. For the two-encoder Slepian-Wolf problem, we introduce a   quantity known as the entropy dispersion matrix, which is analogous to the   scalar dispersions that have gained interest recently. We prove a {\em global dispersion} result that can be expressed in terms of this entropy dispersion matrix and provides intuition on the approximate rate losses at a given blocklength and error probability. To gain better intuition about the rate at which the non-asymptotic rate region converges to the Slepian-Wolf boundary, we define  and characterize two   operational dispersions: the {\em local dispersion} and the {\em weighted sum-rate dispersion}. The former represents the rate of convergence to a point on the Slepian-Wolf   boundary while the latter represents the fastest rate for which a weighted sum of the two rates converges to its asymptotic fundamental limit. Interestingly, when we approach either of the two corner points, the local dispersion  is characterized not by a univariate Gaussian  but a bivariate one as well as a subset of off-diagonal elements of the aforementioned entropy dispersion matrix. Finally, we demonstrate the versatility of our achievability proof technique by providing inner bounds for the multiple-access channel and the asymmetric broadcast channel in terms of dispersion matrices. All our proofs are unified  a so-called vector rate redundancy theorem which is proved using the multidimensional   Berry-Ess\'{e}en theorem. 

\end{abstract}

\section{Introduction}
{\em Network information theory}~\cite{elgamal} aims to find the fundamental limits of  communication in networks with multiple senders and receivers. The primary goal is to characterize the {\em optimal rate region} or {\em capacity region}--that is, the set of rate  tuples for which there exists codes with reliable transmission. Such rate tuples are known as being {\em achievable}. While the characterization of capacity regions is a difficult problem in general, there have been positive results for  several   special classes of networks such as the multiple-access channel~\cite{ahl71, liao} and  the asymmetric~\cite{kor77}  or degraded broadcast channels~\cite{cover, bergmans}. A prominent example in multi-terminal lossless source coding in which the optimal rate region is   known is the so-called Slepian-Wolf problem~\cite{sw73} which involves {\em separately} encoding two (or more) correlated sources and subsequently  estimating them from their rate-limited representations.

The  capacity region  for a   channel model is an asymptotic notion. One is allowed to design codes that operate over arbitrarily long blocks (or channel uses) in order to drive either the maximal or average probabilities of error to zero. To illustrate this point, let us recap Shannon's point-to-point channel coding theorem~\cite{Shannon48}. He showed that up to $n C$ bits can be  reliably transmitted  over $n$ uses of a discrete memoryless channel (DMC) $W$ as $n$ becomes large. Here,  $C = \max_{p_X} I(p_X,W)$ is termed the  {\em capacity} of the channel $W$. However, this fundamental result for reliable communication over a noisy channel can  be optimistic in practice as there may be  system constraints on the delay.  One can thus ask a slightly different and   more challenging question: What is the maximal code size $M^*(n,\epsilon)$ as a function of  a fixed blocklength $n$  and target average error probability $\epsilon$? The second-order asymptotic behavior  of $\log M^*(n,\epsilon)$ was studied first by Strassen~\cite{strassen} and the analysis was extended recently by Hayashi~\cite{Hayashi09} and Polyanskiy, Poor and Verd\'{u}~\cite{PPV10}. They showed that for most channels and for all $\epsilon\in (0,1)$, 
%This problem has been studied rather extensively recently. Perhaps the most prominent work is that by  Polyanskiy, Poor and Verd\'{u}~\cite{PPV10} who  showed using Gaussian approximations  (and    the   Berry-Ess\`{e}en theorem~\cite[Ch.\ XVI.5]{feller})  that 
\begin{equation} \label{eqn:channel_disp}
\frac{1}{n}\log M^*(n,\epsilon)  =  C - \sqrt{\frac{V}{n}} \rmQ^{-1}(\epsilon) + O\left(\frac{\log n}{n}\right).
\end{equation}
The constant  $V$  coincides with an operational quantity known as the  {\em channel dispersion}~\cite{PPV10}, which is similar to the {\em second-order coding rate} in \cite{Hayashi09}. The  channel dispersion is the variance of the log-likelihood ratio of the channel $W$ and the capacity-achieving output distribution $p_{Y^*}$ assuming uniqueness of the capacity-achieving input distribution $p_{X^*} :=\argmax_{p}I(p,W)$.  The term $\sqrt{\frac{V}{n}} \, \rmQ^{-1}(\epsilon)$ is  approximately the rate penalty in at blocklength $n$. The first two terms in \eqref{eqn:channel_disp} are   known as the {\em Gaussian approximation} to $R^*(n,\epsilon)$. %  Both \cite{Hayashi09} and~\cite{PPV08}  noted that \eqref{eqn:channel_disp} holds verbatim for the additive white Gaussian noise (AWGN) channel~\cite{Shannon59}.

In this paper, using   Gaussian approximations, we ask similar  dispersion-type questions for three multi-user  problems: distributed lossless source coding, also known as    the Slepian-Wolf (SW) problem, the multiple-access channel (MAC) and the asymmetric broadcast channel (ABC). We show that the network analogue of the scalar dispersion quantity  $V$   is a positive-semidefinite  matrix $\bV$ that generally depends on the channel, input distributions or sources. We call this a {\em global dispersion} result. Furthermore, we also perform {\em local dispersion} analysis. Just as $V$ in \eqref{eqn:channel_disp} quantifies the rate of convergence of $R^*(n,\epsilon)$ to capacity,  we examine the rate of convergence to various points on the boundary of the asymptotic rate region for the SW problem.  Our results  are     of practical importance    due to the ubiquity of communication networks where numerous users simultaneously share a data compression system or utilize a common channel. Since there may be hard constraints on the permissible number of channel uses (i.e., the delay in decoding), it is useful to gain an intuition of the approximate backoff from the asymptotic fundamental limits in terms of a quantity that is analogous to $V$ in \eqref{eqn:channel_disp}. 

%For example, given a tolerable average error probability of $\epsilon$ and a blocklength   $n$, what is the set of achievable   rate pairs $(R_1,R_2)$ for two non-cooperating parties to communicate to a common destination? Our results in Section~\ref{sec:mac} for the MAC provide a partial  answer (achievability/inner bound) to this question.  As the method of types~\cite{Csi97} is used as a key proof technique in this paper, we focus  on discrete memoryless systems   but,  at various points in the paper, we will also comment on how our techniques can be generalized to systems with arbitrary alphabets. 

\subsection{Summary of Main Results} \label{sec:summary}
There are three main results in this paper:

\begin{enumerate}
\item For the SW problem, we   define  the   $(n,\epsilon)$-optimal rate region $\scR_{\mathrm{SW}}^*(n,\epsilon)$   to be the set of rate pairs $(R_1, R_2)$ for which there exists a length-$n$ code  such that the error probability in reconstructing the sources does not exceed $\epsilon$. We characterize   $\scR_{\mathrm{SW}}^*(n,\epsilon)$  up to an $O(\frac{\log n}{n})$ factor. More precisely, we show the following  {\em global dispersion} result (Theorem~\ref{thm:orr}) for the SW problem: $\scR_{\mathrm{SW}}^*(n,\epsilon)$ is the set of rate pairs $(R_1, R_2)$ satisfying
\begin{equation} \label{eqn:sw_intro}
\begin{bmatrix}R_1\\ R_2\\R_1+R_2 \end{bmatrix} \in \begin{bmatrix}H(X_1|X_2 )\\ H(X_2|X_1 ) \\ H(X_1,X_2) \end{bmatrix}  + \frac{\scS(\bV,\epsilon)}{\sqrt{n}} \pm  O\left(\frac{\log n}{n}\right)\bone,
\end{equation}
where the set $\scS(\bV,\epsilon) \subset\bbR^3$ is the multidimensional analogue of the cumulative distribution function for a zero-mean multivariate Gaussian with covariance matrix  $\bV$. See Fig.~\ref{fig:slices} for a schematic of $\scR_{\mathrm{SW}}^*(n,\epsilon)$.  This is pleasingly analogous to \eqref{eqn:channel_disp} in which the backoff from the asymptotic optimal rate region is of the order  $O(\frac{1}{\sqrt{n}})$. The constant is also specified as the dispersion matrix $\bV$.

\item However, while the global dispersion result for SW in~\eqref{eqn:sw_intro}    resembles the channel dispersion one in \eqref{eqn:channel_disp}, it differs in one key aspect. Namely, the rate at which $\scR_{\mathrm{SW}}^*(n,\epsilon)$  approaches certain boundary points of the asymptotic SW region is somewhat nebulous. To clarify this, we define two related operational dispersions which we also characterize exactly. First, we consider approaching various points on the boundary at a specified angle $\theta$ (See Fig.~\ref{fig:slices}). Interestingly, when we approach either of the corner points, the {\em local dispersion} (proved in Theorem~\ref{thm:local_disp}) is characterized not by a univariate Gaussian (via the $\rmQ^{-1}$ function as in \eqref{eqn:channel_disp}) but a bivariate Gaussian, a subset of the off-diagonal elements of the dispersion matrix $\bV$ and the angle of approach. This phenomenon is not observed when we approach non corner-points. Indeed, in this case, the local dispersion is simply characterized by  an element on the diagonal of $\bV$ and the angle of approach.  Second, suppose we want to minimize a linear combination of $R_1$ and $R_2$, say $\alpha R_1+\beta R_2$ for some $\alpha,\beta\ge 0$.  It can be seen from the polygonal shape of the SW region that for almost all values of $(\alpha,\beta)$, the resulting rate pairs will converge to one of the two corner points. We characterize the speed at which the $(\alpha,\beta)$-weighted sum-rate of the best SW code of length $n$ with error probability not exceeding $\epsilon$ converges to either $\alpha H(X_1|X_2)+\beta H(X_2)$ or $\alpha H(X_1)+\beta H(X_2|X_1)$.  We call the proportionality constant involved in this speed  the {\em $(\alpha,\beta)$-weighted sum-rate dispersion}  or simply the {\em weighted sum-rate dispersion} (proved in Theorem~\ref{thm:sum_rate_disp}).

\item Lastly, to  demonstrate the full utility of our achievability proof technique which is based on the method of types~\cite{Csi97}, we apply it to obtain second-order-type inner bounds for the $(n,\epsilon)$-capacity regions for the discrete memoryless MAC (Theorem~\ref{thm:mac}) and the discrete memoryless  ABC (Theorem~\ref{thm:bc}). These inner bounds are expressed like the global dispersion result in \eqref{eqn:sw_intro} but similar local  and weighted sum-rate dispersions can also be derived. 
\end{enumerate}
 
\subsection{Related Work} \label{sec:related}

The asymptotic expansions   for the fundamental limits of hypothesis testing, source and channel coding were first studied by Strassen~\cite{strassen}. Subsequently, dispersion or second-order coding analysis for channel coding for various point-to-point channel models were studied in~\cite{Hayashi09} and~\cite{PPV10}.   Such dispersion   analysis has promptly been extended to   lossy source coding~\cite{Kos11, ingber11} and joint source-channel coding~\cite{wang11}.   Dispersion analysis is  complementary  to that of traditional error exponent analysis \cite{gallagerIT, Csi97}. In the latter, we fix a rate tuple in the capacity region and ask how rapidly the error probability decays as an exponential function of the blocklength. In the former, the error probability and the blocklength are fixed (though for tractability, we often allow $n$ to also grow and we study the  asymptotics). The spotlight is  now shone  on achievable rates at the specified blocklength and error probability.

The problem of SW coding for a  fixed error probability  and blocklength was   discussed by Baron et al.~\cite{Baron04}, Sarvotham et al.~\cite{Sar05b} and He et al.~\cite{He09}. However, in these works, the authors considered a single source $X_1$ to be compressed and (non-coded) side information $X_2$ available only at the decoder. Thus, $X_2$ is neither coded nor estimated. They showed that a scalar dispersion quantity governs the second-order coding rate. Thus, for this problem, we cannot observe the peculiar corner point phenomenon  discussed in the second point in Section~\ref{sec:summary}. He et al.~\cite{He09}  also analyzed the variable-length   SW problem and showed that the   dispersion is, in general, smaller than  in the fixed-length setting.  Due to the duality between one-encoder  SW coding and  channel coding~\cite{Chen07,  Ahls79, Ahls79b}, this variable-length dispersion  shown to be similar to that for channel coding \cite{Hayashi09, PPV10}. However, it is again not clear how to obtain the dispersion matrix-type result in \eqref{eqn:sw_intro} or the local and weighted sum-rate dispersions by exploiting the duality between channel coding and the one-encoder SW coding problem~\cite{Chen07,  Ahls79}.  There is also duality between the two-encoder SW problem and the MAC as stated in \cite[Theorem 14.3]{Csi97} but it is not clear whether this duality can be exploited for deriving conclusive dispersion results for the MAC. Sarvotham et al.~\cite{Sar05} considered the SW problem with two sources to be compressed but limited their setting to the case the sources  are binary and  symmetric. They demonstrated a result analogous to Baron et al.~\cite{Baron04}. The three constraints on the individual rates $R_1$, $R_2$ and the sum rate $R_1+R_2$ are decoupled when the sources are binary and symmetric. Similar conclusions were made by Chang and Sahai~\cite{Chang07} from an error exponent perspective.  Our work generalizes their setting in that we consider {\em all} finite alphabet sources (not necessarily symmetric)  with {\em multiple} encoders. We discuss  further connections  in Sections~\ref{sec:sing}.  

%To the best of the authors' knowledge, the SW problem is the only network information theory problem in which second-order coding rates and/or finite blocklength behaviours  have been studied and published. Through personal communication with Prof.\ Pierre Moulin (UIUC)~\cite{moulinPersonal}, the authors also came to know that Prof.\ Moulin has been  working concurrently on finite blocklength analysis for the MAC.  

\subsection{Paper Organization}
This paper is organized as follows: In the following subsection, we introduce our notation. In Section~\ref{sec:sw}, we present our dispersion results for the problem of distributed lossless source coding (the SW problem).  The global dispersion result is stated first, followed by the local  and weighted sum-rate dispersion results. We then provide a thorough discussion of these results, comparing and contrasting them. Following that in Sections~\ref{sec:mac} and \ref{sec:abc}, we present the second-order inner bounds  for the MAC and ABC respectively.   We conclude our discussion and suggest avenues for further research in Section~\ref{sec:concl}. Most of the proofs are presented in Section~\ref{sec:proofs} where we start by presenting  a general result known as the vector rate redundancy theorem. We subsequently apply it in the achievability proofs for the SW problem, the MAC and the ABC.  The  proofs of the  local  and weighted sum-rate dispersion results are presented in the appendices, together with other auxiliary results.

\subsection{Notation} \label{sec:notation}
We adopt the following set of notation: Random variables  and the values they take on will be denoted by upper case  (e.g., $X$) and lower case (e.g., $x$) respectively. Random vectors will be denoted by upper case bold font  or with a superscript indicating its length (e.g., $\bX$ or $X^n=(X_1,\ldots, X_n)$). Their realizations will be denoted by  lower case bold font or with a superscript (e.g., $\bx$ or $x^n= (x_1,\ldots, x_n)$). Matrices will also be denoted by upper case bold font (e.g., $\bM$); this should hopefully cause no confusion with random vectors.   The notation $\bM^T$ denotes the transpose of $\bM$.  The notations $\bM\succ 0$ and $\bM\succeq 0$  mean that $\bM$ is (symmetric) positive-definite and positive-semidefinite respectively. In addition, $\lambda_{\min}(\bM)$, $\lambda_{\max}(\bM)$ and  $\|\bM\|_2$ denote, respectively, the minimum and maximum eigenvalue and the spectral norm of $\bM$. The $(i,j)$ element of $\bM$ is denoted as $[\bM]_{i,j}$.  For a vector $\bv \in \bbR^d$, $\|\mathbf{v}\|_q= (\sum_{t=1}^d |v_t|^q)^{1/q}$  is the $\ell_q$  norm for $q\in [1,\infty]$.  The notation $\bone$ denotes the vector of all ones. For two vectors $\bu,\bv\in\bbR^d$,   $\bu\le\bv$ means $u_t\le v_t$ for all $t=1,\ldots, d$. The notation $\bu\ge\bv$   is defined similarly. Sets   will   be denoted by calligraphic font (e.g., $\calX$).  Subsets of Euclidean space will be denoted by  script font (e.g., $\scR$). 

Types (empirical distributions) will be denoted by upper case (e.g., $P$) and distributions by lower case (e.g., $p$).  The set of   distributions supported on a finite set $\calX$ and   the set of $n$-types supported on $\calX$ will be denoted by $\scP(\calX)$ and  $\scP_n(\calX)$ respectively. The type of a sequence $x^n$  is denoted as $P_{x^n}$. The set of all sequences whose type is some $P$ is denoted as $\calT_P$, the {\em type class}. For two sequences $x^n\in \calX^n, y^n\in \calY^n$, the conditional type of $y^n$ given $x^n$ is the stochastic matrix $V:\calX\to\calY$  satisfying $P_{x^n}(a)V(b|a) = P_{x^n, y^n} (a,b)$ for all $(a,b)\in\calX\times\calY$. The set of $y^n$ with conditional type $V$ given $x^n$ is denoted by $\calT_V(x^n)$, the {\em $V$-shell of $x^n$}.  The family  of   stochastic matrices $V:\calX\to\calY$ for which the $V$-shell   of a sequence $x^n \in\calT_P$ is not empty is denoted as $\scV_n(\calY;P)$~\cite[Sec.~2.5]{Csi97}. %In other words, $V\in\scV_n(\calY;P)$ if and only if $nP(a)V(b|a)$ is an integer for all $(a,b)\in\calX\times\calY$. 

Entropy  and conditional entropy are denoted as $H(X)=H(p_{X})$ and $H(Y|X) = H( p_{Y|X}|p_{X})$ respectively. Mutual information is denoted as $I(X;Y)= I(p_X, p_{Y|X})$. We often times make the dependence on the distribution explicit.  Let $x^n , y^n$ be a pair of sequences for which the $y^n$ has conditional type $V$ given $x^n$ and let $\tilX$ and $\tilY$ be dummy random variables with joint distribution  $P_{x^n, y^n}$. Then,   the notations $\hatH (x^n)= H(P_{x^n})=H(\tilX)$ and $\hatH (y^n|x^n)=H(V|P_{x^n})=H(\tilY|\tilX)$ denote, respectively, the empirical marginal and conditional entropies respectively. Note that empirical information quantities will generally be denoted with hats. So for example, the empirical mutual information  of the random variables $\tilX, \tilY$ above will be denoted interchangeably as $\hatI(x^n \wedge y^n) = I(P_{x^n}, V) = I(\tilX ; \tilY)$.     Empirical conditional mutual information is defined similarly. 

The multivariate Gaussian probability density function with mean $\boldm$ and covariance   $\bLambda$ is denoted as $\calN(\bu; \boldm, \bLambda)$ or more simply as $\calN(\boldm, \bLambda)$. For a standard univariate  Gaussian $\calN(u;0,1)$, the cumulative distribution function and $\rmQ$-function are defined as $\Phi(z)  := \int_{-\infty}^z \calN(u;0,1)\, \mathrm{d}u$ and $\rmQ(z):=1-\Phi(z)$ respectively. The functional inverse of the $\rmQ$-function is denoted as $\rmQ^{-1}(\epsilon)$. The Bernoulli random variable $X\sim\Ber(q)$ if $\rvP(X=1)=q$ and $\rvP(X=0)=1-q$.  Logarithms are to the base 2.  We also use the discrete interval notation $[2^{nR}]:=\{1,\ldots, \ceil{2^{nR}}\}$.  Asymptotic notation  such as $o(\fndot), O(\fndot)$ and $\Theta(\fndot)$ is used throughout. See~\cite[Sec.~I.3]{Cor03} for definitions.% We say that $f(n)= o(g(n))$ if and only if $\lim_{n\to\infty}f(n)/g(n)= 0$, $f(n)=O(g(n))$ if and only if $\limsup_{n\to\infty} |f(n)/g(n)|<\infty$, and $f(n)=\Theta(g(n))$ if and only if $f(n)=O(g(n))$ and $g(n)=O(f(n))$.  

\section{Dispersion of   Distributed Lossless Source Coding} \label{sec:sw}

Distributed lossless source coding---also known as the Slepian-Wolf or SW problem---consists in {\em separately} encoding two (or more) correlated sources $(X_1^n, X_2^n)\sim\prod_{k=1}^n p_{X_1, X_2}(x_{1k}, x_{2k})$ into a pair of rate-limited messages $(M_1,M_2) \in [2^{nR_1}]\times[2^{nR_2}]$. Subsequently, given these compressed versions of the sources, a decoder seeks to reconstruct $(X_1^n, X_2^n)$. One of the most  remarkable results in information theory, proved by Slepian and Wolf in 1973~\cite{sw73}, states  that the set of achievable rate pairs $(R_1, R_2)$ is asymptotically equal to that when each of the encoders is also given knowledge of the other source, i.e., encoder 1 knows $X_2^n$ and vice versa. The optimal rate region $\scR_{\mathrm{SW}}^*$ is  given by  the polyhedron
\begin{subequations}\label{eqn:sw_reg}
\begin{align}
R_1& \ge H(X_1|X_2 ) \\*
R_2&\ge H(X_2|X_1 ) \\*
R_1+R_2 &\ge H(X_1,X_2 ).  
\end{align}  
\end{subequations}
We are also interested in the optimal \emph{weighted sum-rate}; that is, for constants $\alpha,\beta\ge 0$, the minimum value of $\alpha R_1+\beta R_2$ for achievable $(R_1,R_2)$. Of particular interest is the case $\alpha=\beta=1$, corresponding to the standard sum-rate, but other cases may be important as well, such as if transmitting from encoder 1 is more costly than transmitting from encoder  2. Because of the polygonal shape of the optimal region described in \eqref{eqn:sw_reg}, the optimal weighted sum-rate is always achieved at one of the two corner points, and the optimal rate is given by
\begin{align}
R_\text{sum}^*(\alpha,\beta):=\begin{cases} \alpha H(X_1|X_2)+\beta H(X_2) & \alpha\ge\beta\\ \alpha H(X_1)+\beta H(X_2|X_1) & \alpha<\beta.\end{cases}
\label{eqn:opt_sum_rate}
\end{align}
As with most other statements in   information theory \cite{Csi97}, the results in \eqref{eqn:sw_reg} and \eqref{eqn:opt_sum_rate} are first-order asymptotic. In this section, we analyze the second-order, or dispersion behavior of the SW problem. That is, we study how quickly achievable rates can approach the asymptotic fundamental limits given in \eqref{eqn:sw_reg} and \eqref{eqn:opt_sum_rate} as the blocklength grows.

We will focus on the two-sender case. %Generalizations will turn out to be straightforward. 
A SW code is characterized by four parameters; the  {\em blocklength} $n$, the {\em rates} of the first and second sources $(R_1, R_2)$ and the {\em  probability of error}   defined as 
\begin{equation}
\Pen:= \rvP( ( \hatX_1^n, \hatX_2^n)\ne (X_1^n , X_2^n)), \label{eqn:perr_sw}
\end{equation}
where $\hatX_1^n$ and  $\hatX_2^n$ are the reconstructed versions of $X_1^n$ and $X_2^n$ respectively.\footnote{A more challenging task would be to consider constituent error probabilities $\rvP(\hatX_1^n\ne X_1^n)$, $\rvP(\hatX_2^n\ne X_2^n)$ and $\Pen$ and place three  different upper bounds $\epsilon_1, \epsilon_2$ and $\epsilon_3$  on these probabilities. We choose to consider the single compound error probability in \eqref{eqn:perr_sw} for simplicity.} Each blocklength $n$ and probability of error $\epsilon \in (0,\frac{1}{2})$ results in some achievable region of rate pairs that will in general be smaller than the asymptotically optimal region $\scR_{\mathrm{SW}}^*$ (for $\epsilon \in (\frac{1}{2},1)$, the achievable region will, in general, be larger). Our first result in this section is a characterization of the $(n,\epsilon)$-optimal region up to a $O(\frac{\log n}{n})$ correction term. This is a tight second-order result in the sense that it the gap between inner and outer bounds is $O(\frac{\log n}{n})$, and thus it exactly specifies the constants on the $O(\frac{1}{\sqrt{n}})$ terms with which the $(n,\epsilon)$-region approaches the asymptotically optimal region. However, it is a global rather than a local result, and as such it is opaque to certain behaviors about how the achievable region approaches the optimal SW boundary. To complete the story, we also define two other dispersions operationally. We characterize these  dispersions exactly by leveraging the first, global result. The first is \emph{local dispersion}, meaning the speed of convergence to a specific point on the boundary of the asymptotically optimal region from a specific angle. The second type of dispersion considers the weighted sum-rate discussed above: in particular, how quickly the weighted sum-rate can approach the asymptotically optimal rate given in \eqref{eqn:opt_sum_rate}.

We start with definitions followed by  the statements of our results. We then discuss the implications of our results. The proof of the global result is provided in Section~\ref{sec:prf_sw}, and the proofs of the other dispersion results are in Appendix~\ref{app:slice}.

\subsection{Definitions} \label{sec:defs_sw}
Let $(\calX_1,\calX_2,p_{X_1, X_2}(x_1,x_2))$ be a discrete memoryless multiple source (DMMS). This means that $(X_1^n, X_2^n)\sim\prod_{k=1}^n p_{X_1,X_2}(x_{1k}, x_{2k})$, i.e., the source is independent and identically distributed (i.i.d.). We remind the reader that the alphabets $\calX_1, \calX_2$ are finite. We also assume throughout that $p_{X_1, X_2}(x_1, x_2)>0$ for every $(x_1, x_2)\in\calX_1\times\calX_2$ and that the sources are not independent. Finally, we assume that the error probability $0<\epsilon <1$.

\begin{definition} \label{def:sw_code}
An {\em $(n,2^{nR_1}, 2^{nR_2}, \epsilon)$-SW code} consists of two encoders $f_{j,n}:\calX_j^n\to \calM_{j} =[2^{nR_j}], j=1,2$,   and a decoder $\varphi_n:\calM_{1}\times \calM_{2}\to \calX_1^n \times \calX_2^n$ such that the  the error probability in~\eqref{eqn:perr_sw} with $(\hatX_1^n , \hatX_2^n) :=\varphi_n(  f_{1,n}(X_1^n), f_{2,n}(X_2^n))$  does not exceed $\epsilon$. %The {\em compression  rates} are  defined  in the usual way as 
%\begin{equation}
%R_j:=\frac{\log|\calM_{j}|}{n}. \label{eqn:def_rates}
%\end{equation}
\end{definition}

\begin{definition} \label{def:neps_ach}
A rate pair $(R_1 ,R_2 )$ is  {\em  $(n,\epsilon)$-achievable}  if there exists an $(n, 2^{nR_1 }, 2^{nR_2 }, \epsilon)$-SW code for the DMMS $p_{X_1, X_2}(x_1,x_2)$. The {\em  $(n,\epsilon)$-optimal rate region} $\scR_{\mathrm{SW}}^*(n,\epsilon) \subset\bbR^2$ is the set of all  $(n,\epsilon)$-achievable  rate pairs. 
\end{definition}

\begin{definition} \label{def:weighted_sum}
A weighted sum-rate $R_\textrm{sum}$ is $(n,\epsilon,\alpha,\beta)$-\emph{achievable} if there exists an $(n,\epsilon)$-achievable pair $(R_1,R_2)$ such that $\alpha R_1+\beta R_2\le R_\text{sum}$. Let $R_\text{sum}^*(n,\epsilon;\alpha,\beta)$ be the minimum $(n,\epsilon,\alpha,\beta)$-achievable sum-rate.
\end{definition}

\begin{figure}
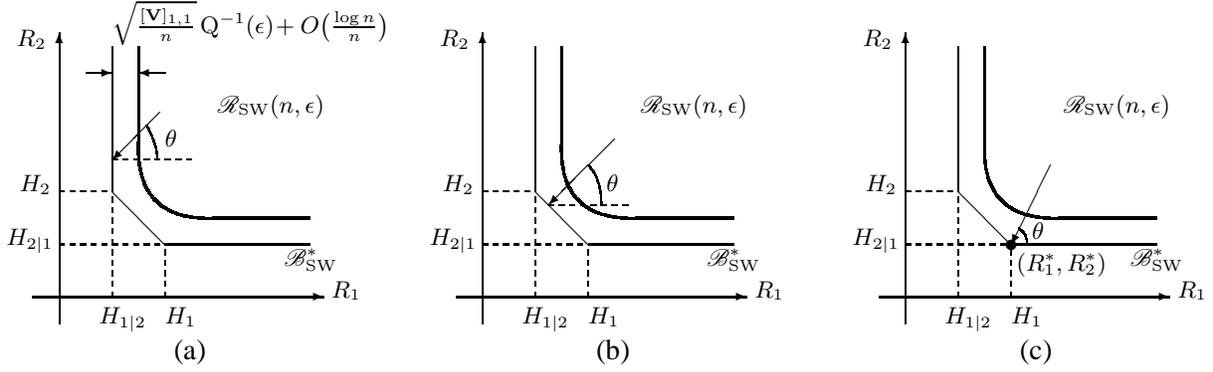

\centering
\begin{tabular}{ccc}
\begin{picture}(115, 115)
\input{common_pic}

\put(20, 95){\vector(1,0){10}}
\put(50, 95){\vector(-1,0){10}}
\put(30, 110){\footnotesize $\sqrt{ \frac{[ \bV]_{1,1}}{n}} \, \rmQ^{-1}(\epsilon)\!+O\big(\frac{\log n}{n}\big)$}%\!\tau_n$}

\put(48, 80){\vector(-1,-1){18}}
%\put(50, 80){\footnotesize $\vec{\bd}$}
\multiput(30,62)(4,0){8}{\line(1,0){2}}
\qbezier(43,75)(47,68)(47,62)
%\put(22, 56){\footnotesize $\delta$}
%\put(27, 60){\vector(0,1){10}}
%\put(27, 60){\vector(0,-1){10}}
\put(50,67){\footnotesize $\theta$}
\end{picture} &  \hspace{.35in}
\begin{picture}(115, 115)
\input{common_pic}

\put(60, 70){\vector(-1,-1){25}}
%\put(56, 56){\footnotesize $\vec{\bd}$}
%\put(10, 82){\line(1,-1){61}}
%\put(70, 14){\footnotesize $\scL_c$}
%\put(3, 81){\footnotesize $c$}
\multiput(35,45)(4,0){8}{\line(1,0){2}}
\qbezier(50,60)(55,55)(55,45)
\put(57,50){\footnotesize $\theta$}
\end{picture} &  \hspace{.35in}
\begin{picture}(115, 115)
\input{common_pic}
 
\put(65, 60){\vector(-1,-2){15}}
%\put(66,61 ){\footnotesize $\vec{\bd}$}
\qbezier(53,36)(56,34)(56,30)
\put(58,32){\footnotesize $\theta$}
%\put(56,41){\circle*{4}}
%\put(61,45){\footnotesize $(   R_1(n,\epsilon),   R_2(n,\epsilon)   )$}
\put(50,30){\circle*{4}}
\put(52,20){\footnotesize $(   R_1^*,    R_2^*   )$}
\end{picture} \\
(a) &  \hspace{.35in} (b)&  \hspace{.35in} (c) 
\end{tabular}
\caption{Schematic plots of the $(n,\epsilon)$-optimal rate region $\scR_{\mathrm{SW}} (n,\epsilon)$  for $\epsilon\le\frac{1}{2}$  and the asymptotic SW region in \eqref{eqn:sw_reg} whose boundary is indicated by $\scB_{\mathrm{SW}}^*$. We use the simplified notation $H_1:= H(X_1), H_2:= H(X_2), H_{1|2}:= H(X_1|X_2)$,  $H_{2|1}:= H(X_2|X_1)$ and $H_{1,2}=H(X_1, X_2)$. The directions of approach   are indicated by the arrows in the different subplots. In subplot (a), we approach the vertical boundary; the local dispersion $F(\theta,\epsilon;R_1^*,R_2^*)$ is given in~\eqref{eqn:subplota}. In subplot (b), we approach the sum-rate boundary; the local dispersion is given in~\eqref{eqn:subplotb}. In subplot (c), we approach the corner point $(H_1, H_{2|1})$; the local dispersion is given  implicitly in~\eqref{eqn:subplotc}.} \label{fig:slices}
\end{figure}

Our analysis in this paper will be focused not on providing direct bounds on $\scR_{\mathrm{SW}}^*(n,\epsilon)$ and $R_\text{sum}^*(n,\epsilon;\alpha,\beta)$ for finite $n$, but rather on the speed at which these approach $\scR_{\text{SW}}^*$ and $R_\text{sum}^*$ respectively, as $n\to\infty$. In particular, we are interested in characterizing the quanities defined in the following two definitions. These are both versions of operational dispersion. The first is \emph{local dispersion} (illustrated in Fig.~\ref{fig:slices}): the speed of convergence to a particular asymptotic rate pair from a given angle,
%\footnote{\label{fn:curve}One could also pose the following intriguing and arguably more challenging question. Given an arbitrary sequence $\{R_{1,n}\}_{n\ge 1}$ that converges   to some point  $R_{1,\infty}\in [H(X_1|X_2), H(X_1)]$, what is the best rate of convergence of $R_{2,n}$ to $R_{2,\infty}$ where the limit point $(R_{1,\infty}, R_{2,\infty})$ is on the boundary of the SW region? As shown in Fig.~\ref{fig:slices}, we only address the scenario in which the trajectory of $\{(R_{1,n}, R_{2,n})\}_{n\ge 1}$ is {\em linear}. Also see~\eqref{eqn:rate_approx} for how $\{(R_{1,n}, R_{2,n})\}_{n\ge 1}$ is parametrized. In general though, $\{(R_{1,n}, R_{2,n})\}_{n\ge 1}$ could lie on curved spaces resulting in a more involved analysis and requiring, for instance, techniques from moderate deviations~\cite{Tan12, altug10}. We believe that the linear perspective is sufficient to provide a good portrait of the second-order behavior and we  defer this more challenging question to future work.  } 
 and the second is \emph{weighted sum-rate dispersion}: the speed of convergence of the weighted sum-rate for a given weight pair.

\begin{definition} \label{def:local_disp}
Fix a rate pair $(R_1^*,R_2^*)$ on the boundary of the  asymptotic SW rate region $\scR_{\mathrm{SW}}^*$, and a probability of error $\epsilon>0$. The dispersion-angle pair $(F,\theta)$ is $(R_1^*,R_2^*,\epsilon)$-\emph{achievable} if there exists a sequence of $(n,2^{nR_{1,n}},2^{nR_{2,n}},\epsilon)$-SW codes such that 
\begin{align}
\limsup_{n\to\infty}\sqrt{n}\left( R_{1,n}-R_1^* \right) &\le \sqrt{F}(\cos\theta) \rmQ^{-1}(\epsilon)\\
\limsup_{n\to\infty}\sqrt{n}\left( R_{2,n}-R_2^* \right) &\le \sqrt{F}(\sin\theta ) \rmQ^{-1}(\epsilon).
\end{align}
The {\em local dispersion} $F(\theta,\epsilon;R_1^*,R_2^*)$ is the infimum of all $F$ such that $(F,\theta)$ is $(R_1^*,R_2^*,\epsilon)$-\emph{achievable}.
\end{definition}

\begin{definition} \label{def:weighted_sum_disp}
The {\em weighted sum-rate dispersion} for the weight pair $(\alpha,\beta)$ and probability of error $\epsilon$ is given by
\begin{equation}
G(\epsilon;\alpha,\beta):=\liminf_{n\to\infty}\  n\left(\frac{R_\text{sum}^*(n,\epsilon;\alpha,\beta)-R_\text{sum}^*(\alpha,\beta)}{\rmQ^{-1}(\epsilon)}\right)^2
\end{equation}
where $R_\text{sum}^*(\alpha,\beta)$ is defined in \eqref{eqn:opt_sum_rate}.
\end{definition}

Observe from Definition~\ref{def:local_disp} that for any $\epsilon>0$, angle $\theta$, and asymptotic rate pair $(R_1^*,R_2^*)$, there exist codes with rates $\{(R_{1,n},R_{2,n})\}_{n\in\bbN}$ and probability of error $\epsilon$ satisfying the approximate relationships
\begin{subequations}\label{eqn:rate_approx}
\begin{align}
R_{1,n}&\approx R_1^*+\sqrt{\frac{F(\theta,\epsilon;R_1^*,R_2^*)}{n}}(\cos\theta)\rmQ^{-1}(\epsilon) \\
R_{2,n}&\approx R_2^*+\sqrt{\frac{F(\theta,\epsilon;R_1^*,R_2^*)}{n}}(\sin\theta)\rmQ^{-1}(\epsilon).
\end{align}
\end{subequations}
The only interesting values of $\theta$ are those for which the rates $\{(R_{1,n},R_{2,n})\}_{n\in\bbN}$ approach the asymptotic rate pair $(R_1^*, R_2^*)$ from the interior (resp.\ exterior) of the asymptotic SW rate region $\scR_{\mathrm{SW}}^*$ when $\epsilon\le\frac{1}{2}$ (resp.\ when  $\epsilon  > \frac{1}{2}$). For example, when approaching a point on the vertical boundary [see Fig.~\ref{fig:slices}(a)], the local dispersion is only interesting if $-\frac{\pi}{2}<\theta<\frac{\pi}{2}$.% assuming $\epsilon\le\frac{1}{2}$ here. 

From Definition~\ref{def:weighted_sum_disp}, for any $\epsilon>0$ and weight pair $(\alpha,\beta)$, there exists codes with rates  $\{(R_{1,n},R_{2,n})\}_{n\in\bbN}$ and probability of error $\epsilon$ satisfying
\begin{equation}
\alpha R_{1,n}+\beta R_{2,n}\approx R_\text{sum}^*(\alpha,\beta)+\sqrt{\frac{G(\epsilon;\alpha,\beta)}{n}}\rmQ^{-1}(\epsilon). \label{eqn:weighted_approx}
\end{equation}
Below, Theorem~\ref{thm:local_disp} exactly characterizes $F(\theta,\epsilon;R^*_1,R^*_2)$, and Theorem~\ref{thm:sum_rate_disp} exactly characterizes $G(\epsilon;\alpha,\beta)$.

We now define quantities that will allow us to state our results. For a positive-semidefinite symmetric matrix $\bV \in\bbR^{d\times d}$, let the random vector $\bZ\sim \calN(\bzero, \bV)$. Note that $\calN(\bzero, \bV)$  is a degenerate Gaussian if $\bV$ is singular. If $\rank(\bV)=r<d$,  all the probability mass of $p(\bu) = \calN(\bu;\bzero,\bV)$ lies in a subspace of dimension $r$ in $\bbR^d$. Define the  set 
\begin{equation}
\scS (\bV,\epsilon):= \{ \bz\in \bbR^3: \rvP (\bZ\le \bz)  \ge 1-\epsilon\}. \label{eqn:SVset}
\end{equation}
Note that $\scS(\bV,\epsilon)\subset\bbR^3$  is well-defined even if $\bV$ is singular. Furthermore,  $\scS (\bV,\epsilon')\subset\scS (\bV,\epsilon)$ if $\epsilon' \le \epsilon$.   This set is  analogous to the (inverse) cumulative distribution function of a zero-mean Gaussian with covariance matrix $\bV$.   
If $\epsilon\le \frac{1}{2}$, $\scS (\bV,\epsilon)$  is a convex, unbounded set in the positive orthant in $\bbR^3$. The boundary of $\scS (\bV,\epsilon)$ is   smooth if $\bV$ is positive-definite. We shall see that this set scaled by $\frac{1}{\sqrt{n}}$, namely $\frac{1}{\sqrt{n}}\scS(\bV,\epsilon)$, plays an important role in specification of bounds on  the $(n,\epsilon)$-optimal rate region. This set is diagrammed in two dimensions (for ease of visualization) in  Fig.~\ref{fig:SV}.  We note that the boundaries are indeed curved due to the fact that $\bV\succ 0$. Note that as $n$  increases to infinity or $\epsilon$ increases towards $\frac{1}{2}$, the boundaries are translated closer to the horizontal and vertical axes. If $\epsilon>\frac{1}{2}$, the region strictly includes the positive orthant. Also observe that as the condition number\footnote{Recall that the {\em condition number} of $\bV$ is the ratio of its maximum to minimum eigenvalues, i.e., $\mathrm{cond}(\bV)=\lambda_{\max}(\bV)/\lambda_{\min}(\bV)$. } $\bV$ increases, i.e., $\bV$   tends towards being singular,  the corners of the curves become ``sharper'' (or ``less rounded'').  Indeed, in the limiting case when $\bV$  has rank one, the support of $p(\bu)=\calN(\bu;\bzero,\bV)$ belongs to a subspace of dimension one. In this case,  the set  $\scS (\bV,\epsilon)$ is an axis-aligned, unbounded rectangle (a cuboid in higher dimensions). See further discussions in Section~\ref{sec:sing}.

\begin{figure}
\centering
\begin{overpic}[width = .8\columnwidth]{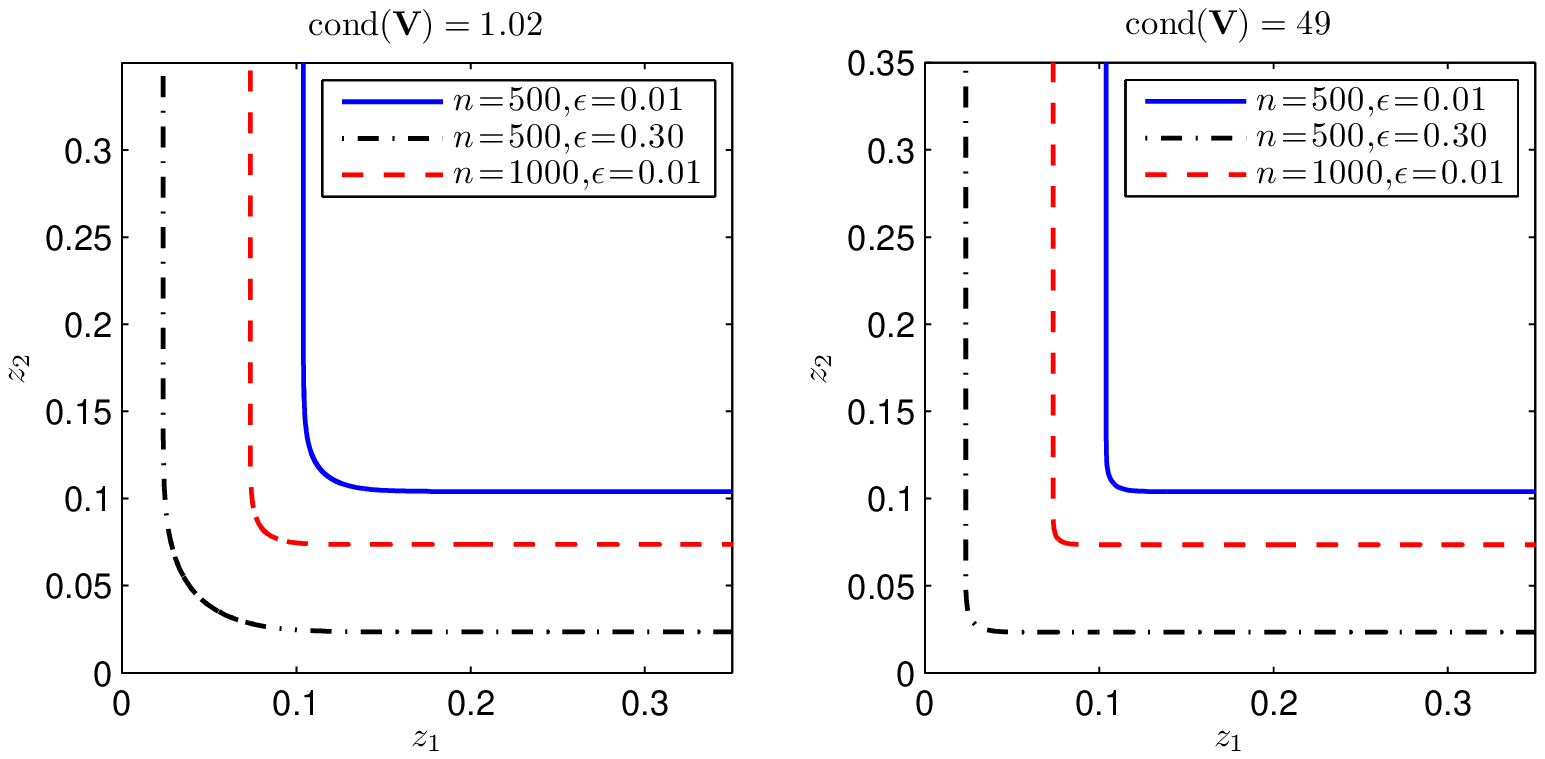}
\put(25,27){\small $\frac{1}{\sqrt{n}}\scS(\bV,\epsilon)$}
\put(75,27){\small $\frac{1}{\sqrt{n}}\scS(\bV,\epsilon)$}
\end{overpic}
\caption{The boundaries of the region $\frac{1}{\sqrt{n}}\scS(\bV,\epsilon)$ for different values  $n$, $\epsilon$ and $\bV$. On the left plot, $\bV=[ 1 \,\,\, 0.01 ; 0.01 \,\,\, 1]$ (small condition number) and on the right,  $\bV=[ 1 \,\,\, 0.96 ; 0.96 \,\,\, 1]$ (large condition number).     The regions $\frac{1}{\sqrt{n}}\scS(\bV,\epsilon)$ lie  to the top right corner of the boundaries. $\scS(\bV,\epsilon)$  defined in \eqref{eqn:SVset} is a subset of $\bbR^3$ but in the figures, we only illustrate the projection of the set in two dimensions.} 
\label{fig:SV}
\end{figure}

\begin{definition}
The {\em  entropy density vector} is defined as 
\begin{equation}
\bh( X_1, X_2 ) := \begin{bmatrix} -\log p_{X_1|X_2}(X_1| X_2)\\ -\log p_{X_2|X_1}(X_2| X_1) \\ -\log p_{X_1,X_2}(X_1 , X_2) \end{bmatrix} . \label{eqn:ent_dens}
\end{equation}
\end{definition}
The mean of the entropy density vector is  the  vector of entropies, i.e., 
\begin{equation}\label{eqn:entropy_vec}
\rvE [\bh( X_1, X_2 ) ] = \bH( p_{X_1, X_2}) := \begin{bmatrix} H(X_1|X_2) \\ H(X_2|X_1) \\ H(X_1, X_2) \end{bmatrix}. 
\end{equation}
We denote the entries of $ \bH( p_{X_1, X_2})$ as $H_t( p_{X_1, X_2})$ for $1\le t\le 3$. Also, let 
\begin{equation}
\kappa:= \max_{1\le t\le 3}   \big\|\nabla_{p_{X_1,X_2}}^2 H_t(p_{X_1,X_2}) \big\|_2 \label{eqn:kappa_sw}
\end{equation}
 be  the maximum of the spectral norms of the Hessians of $p_{X_1,X_2}\mapsto H_t(p_{X_1,X_2})$, viewed as functions of the vectorized version of $p_{X_1,X_2} \in \bbR^{|\calX_1| |\calX_2|}$. 
\begin{definition} 
 The {\em entropy dispersion matrix}   $\bV( p_{X_1, X_2})$ is the covariance matrix of the random vector $\bh( X_1, X_2)$ i.e., 
 \begin{equation}
\bV(p_{X_1, X_2}) = \cov( \bh( X_1, X_2) ). \label{eqn:swv}
\end{equation}
\end{definition}
We abbreviate the deterministic quantities  $\bH( p_{X_1, X_2}) \in \bbR^3$ and $\bV( p_{X_1, X_2}) \succeq 0$ as $\bH$ and $\bV$ respectively.   Observe that $\bV$ is a matrix analogue of scalar dispersion.
% quantities that have gained attention in recent years~\cite{PPV10, Kos11, ingber11, wang11}. 
We will find it convenient, in this and following sections, to define the non-negative {\em rate vector} $\bR \in \bbR^3$  as
\begin{equation} \label{eqn:rate_vec}
 \bR:=\begin{bmatrix} R_1 \\ R_2  \\ R_1+R_2 \end{bmatrix} .
\end{equation}
\begin{definition} \label{def:inner_sw}
Define the region $\scR_{\mathrm{in}}(n,\epsilon) \subset\bbR^2$ to be  the set of rate pairs $(R_1,R_2)$ that satisfy
\begin{equation}
\bR \in \bH  + \frac{1}{\sqrt{n}}  \scS (\bV  ,  \epsilon) +   \frac{\nu\log n}{n}  \bone ,\label{eqn:ach_sw}
\end{equation}
where $\nu: = |\calX_1||\calX_2|+ \kappa+3/2$. Also define the region $\scR_{\mathrm{out}}(n,\epsilon)\subset\bbR^2$ to be  the set of rate pairs $(R_1,R_2)$ that satisfy
\begin{equation}
\bR \in \bH  + \frac{1}{\sqrt{n}}  \scS (\bV  ,  \epsilon) -  \frac{ \log  n}{n} \bone . \label{eqn:con_sw}
\end{equation}
\end{definition}
%An illustration of the regions, disregarding the $O(\frac{ \log  n}{n})$ factors, is provided in Fig.~\ref{fig:slices}. Also see Fig.~\ref{fig:varyn}  for how the regions vary with $n$ and $\epsilon$. % and~\ref{fig:varyn}. %Henceforth, $\epsilon\in (0,1)$. 

\begin{definition} 
Define  the bivariate generalization of the $\rmQ$-function as
\begin{equation}
\Psi(\rho; x',y'):= \frac{1}{2\pi\sqrt{1-\rho^2}} \int_{x'}^{\infty}\int_{y'}^{\infty} \exp \left\{ -\frac{x^2-2\rho xy+y^2}{2(1-\rho^2)} \right\}\, \mathrm{d}y  \, \mathrm{d}x.   \label{eqn:Psi}
\end{equation}
\end{definition}

\subsection{Main Results and Interpretation}
\begin{theorem}[Global Dispersion   for Slepian-Wolf] \label{thm:orr}% If $\bV\succ 0$ and $\xi:=\rvE [ \| \bh(X_1, X_2) \|_2^3 ]<\infty$, 
Let $\epsilon\in (0,1)$. The $(n,\epsilon)$-optimal rate region $\scR_{\mathrm{SW}}^*(n,\epsilon)$ satisfies
\begin{equation}
 \scR_{\mathrm{in}}(n,\epsilon) \subset \scR_{\mathrm{SW}}^*(n,\epsilon) \subset \scR_{\mathrm{out}}(n,\epsilon)   \label{eqn:rr_sw}
\end{equation}
for all $n$ sufficiently large. Furthermore, the inner bound is universally attainable, i.e., the coding scheme does not depend on the knowledge of the source statistics.
\end{theorem}
The proof is provided in Section~\ref{sec:prf_sw}. 
We now state our results on the local and sum-rate dispersion. These are proved in Appendix~\ref{app:slice} and hold for all $\epsilon \in (0,1)$.

\begin{theorem}[Local Dispersion for Slepian-Wolf] \label{thm:local_disp}
Let $\theta \in [0,2\pi]$. Depending on $(R_1^*,R_2^*)$, there are five cases:
\begin{enumerate}
\item $R_1^*=H(X_1|X_2)$ and $R_2^*>H(X_2)$ (vertical boundary). Then if $-\frac{\pi}{2}<\theta <\frac{\pi}{2}$, 
\begin{equation}
F(\theta,\epsilon;R_1^*,R_2^*)= \frac{[\bV]_{1,1}}{ \cos^2 \theta     }. \label{eqn:subplota}
\end{equation}
\item $R_2^*=H(X_2|X_1)$ and   $R_1^*>H(X_1)$ (horizontal boundary). Then if $0< \theta < \pi$, 
\begin{equation}
F(\theta,\epsilon;R_1^*,R_2^*)= \frac{[\bV]_{2,2}}{ \sin^2 \theta     }. \label{eqn:hori}
\end{equation}
\item $R_1^*+R_2^*=H(X_1, X_2)$, $R_1^*>H(X_1|X_2)$ and  $R_2^*>H(X_2|X_1)$  (sum-rate boundary). Then if $ -\frac{\pi}{4}<\theta <   \frac{3\pi}{4}$, 
\begin{equation}
F(\theta,\epsilon;R_1^*,R_2^*)= \frac{[\bV]_{3,3}}{  (\cos \theta+\sin \theta)^2 }. \label{eqn:subplotb}
\end{equation}
\item $R_1^*= H(X_1|X_2) $ and $R_2^* = H(X_2)$. Then  if $-\frac{\pi}{4} < \theta< \frac{\pi}{2}$, $F(\theta,\epsilon;R_1^*,R_2^*)$ is the solution to 
\begin{equation}
\Psi \left(\rho_{1,3};-\sqrt{\frac{F}{[\bV]_{1,1}}} \, (\cos\theta)\, \rmQ^{-1}(\epsilon), -\sqrt{\frac{F}{[\bV]_{3,3}}} \, (\cos\theta+\sin\theta)\, \rmQ^{-1}(\epsilon) \right) = 1-\epsilon, \label{eqn:corner_pt}
\end{equation}
where $\rho_{1,3}:=[\bV]_{1,3}/\sqrt{[\bV]_{1,1}[\bV]_{3,3}}$ is the correlation coefficient of  the random variables $-\log p_{X_1|X_2}(X_1|X_2)$ and  $-\log p_{X_1,X_2}(X_1, X_2)$. 
\item $R_1^*= H(X_1 ) $ and $R_2^* = H(X_2|X_1)$. Then   if $0<\theta< \frac{3\pi}{4}$, $F(\theta,\epsilon;R_1^*,R_2^*)$ is the solution to 
\begin{equation}
\Psi \left(\rho_{2,3};-\sqrt{\frac{F}{[\bV]_{2,2}}} \, (\sin\theta)\, \rmQ^{-1}(\epsilon), -\sqrt{\frac{F}{[\bV]_{3,3}}} \, (\cos\theta+\sin\theta)\, \rmQ^{-1}(\epsilon) \right) = 1-\epsilon \label{eqn:subplotc}
\end{equation}
where $\rho_{2,3}$ is defined analogously to $\rho_{1,3}$.
\end{enumerate}
\end{theorem}

\begin{theorem}[Weighted Sum-Rate Dispersion for Slepian-Wolf]\label{thm:sum_rate_disp}
If $\alpha\ge\beta$, then
\begin{equation}
G(\epsilon;\alpha,\beta)=\min_{w_1,w_2} (\alpha w_1+\beta w_2)^2
\label{eqn:G_sol}
\end{equation}
where the minimum is taken over all $w_1,w_2$ satisfying
\begin{equation}\label{eqn:G_constraint1}
\Psi\left(\rho_{1,3};-\frac{w_1 \rmQ^{-1}(\epsilon)}{\sqrt{[\bV]_{1,1}}},-\frac{(w_1+w_2)\rmQ^{-1}(\epsilon)}{\sqrt{[\bV]_{3,3}}}\right)=1-\epsilon.
\end{equation}
If $\alpha<\beta$, then $G(\epsilon;\alpha,\beta)$ is also given by \eqref{eqn:G_sol} but with the minimization subject to
\begin{equation}\label{eqn:G_constraint2}
\Psi\left(\rho_{2,3};-\frac{w_2 \rmQ^{-1}(\epsilon)}{\sqrt{[\bV]_{2,2}}},-\frac{(w_1+w_2)\rmQ^{-1}(\epsilon)}{\sqrt{[\bV]_{3,3}}}\right)=1-\epsilon.
\end{equation}
The correlation coefficients $\rho_{1,3}$ and $\rho_{2,3}$ are as defined in Theorem~\ref{thm:local_disp}.
\end{theorem}

\subsubsection{Discussion of Theorem~\ref{thm:orr}} \label{sec:sw_dis}
The direct part  of Theorem~\ref{thm:orr} is proved using the usual random binning argument~\cite{sw73, cover75} together with a multidimensional  Berry-Ess\`{e}en theorem \cite{Ben03}. The latter allows us to prove an important  vector rate redundancy theorem (Theorem~\ref{thm:vrrt}). This    theorem is a recurring proof technique---it is also used to prove the direct parts of the analogous results for the multiple-access and broadcast channels.  The decoder is a modification of a minimum empirical entropy \cite{Csi97} decoding rule. More precisely, we require  the three empirical entropies $\hatH ( X_1^n|X_2^n )$, $\hatH ( X_2^n|X_1^n )$ and $\hatH ( X_1^n,X_2^n  )$   to be {\em jointly}  smaller than some perturbed rate vector $\bR-\delta_n\bone$, where $\bR$ is in the inner bound and $\delta_n = O(\frac{\log n}{n})$. By Taylor's theorem, it can   be  seen that the empirical entropy vector behaves  like a multivariate Gaussian with mean $\bH$ and covariance $\frac{\bV}{n}$, explaining the presence of these terms in~\eqref{eqn:ach_sw} and~\eqref{eqn:con_sw}.      The converse is proved by leveraging on  an   information spectrum theorem for the SW problem by  Miyake and Kanaya~\cite{miyake}. Also see~\cite[Lemma~7.2.2]{Han10}.  Theorem \ref{thm:orr} extends naturally to the case where there are more than two senders.

\subsubsection{Comparison with Polygonal Region} \label{sec:comparison_poly}

We now focus on interpreting the results of Theorem~\ref{thm:local_disp}, which provide a different perspective on the rate region. In particular, we compare the rate region with that of another source coding problem. Consider the $(n,\epsilon)$-region for lossless source coding with side information  at encoders and decoder (SI-ED), also known as {\em cooperative} source coding.  Specifically, first consider the problem of source coding $X_1$ with $X_2$ available as (full non-coded) side information at the encoder and the decoder. Second, we swap the roles of $X_1$ and $X_2$. Third, we consider a single-user source coding problem for the pair  $(X_1,X_2)$. Up to   $O(\frac{\log n}{n})$ terms, this region  $\scR^*_{\mathrm{SI-ED}}(n,\epsilon) \subset\bbR^2$ is the set of rate pairs $(R_1, R_2)$ satisfying the three scalar constraints
\begin{subequations} \label{eqn:sied} 
\begin{align}
R_1&\ge H(X_1|X_2)+\sqrt{\frac{[\bV]_{1,1}}{n}}\, \rmQ^{-1}(\epsilon) \label{eqn:sieda} \\*
R_2&\ge H(X_2|X_1)+\sqrt{\frac{[\bV]_{2,2}}{n}}\, \rmQ^{-1}(\epsilon)  \label{eqn:siedb} \\*
R_1+R_2&\ge H(X_1,X_2)+\sqrt{\frac{[\bV]_{3,3}}{n}}\, \rmQ^{-1}(\epsilon)   . \label{eqn:siedc} 
\end{align}
\end{subequations}
The three {\em decoupled}   constraints in \eqref{eqn:sied}, which describe a {\em piecewise linear} region, represent three single-user simplifications of the problem and therefore three outer bounds to $\scR^*_{\mathrm{SW}}(n,\epsilon)$. The first two inequalities characterizing the region in \eqref{eqn:sied} can be derived in a straightforward manner using a side information (conditional) version of Strassen's original result~\cite{strassen} for hypothesis testing. The last inequality is simply one of Strassen's original results on source coding. Also see Problem 1.1.8 in  Csisz\'{a}r and   K\"{o}rner~\cite{Csi97} and Theorem 1 in Kontoyiannis~\cite{Kot97}.

It may appear that the piecewise linear region in \eqref{eqn:sied} is very different from that described by Theorem~\ref{thm:orr}. In fact, Theorem~\ref{thm:local_disp} asserts that these two regions differ \emph{only} at the two corner points. For example, consider a rate pair approaching a point $(R_1^*,R_2^*)$ on the vertical boundary of the asymptotic region (i.e. $R_1^*=H(X_1|X_2)$ and $R_2^*>H(X_2)$) as in Fig.~\ref{fig:slices}(a). For the region in \eqref{eqn:sied}, the only relevant constraint in the neighborhood of $(R_1^*,R_2^*)$ is the constraint on $R_1$ in~\eqref{eqn:sieda}. For the region for the full SW problem,  substituting \eqref{eqn:subplota} into \eqref{eqn:rate_approx}, we find that the best rates approaching $(R_1^*,R_2^*)$ from direction $\theta$ are given by the approximate relations
\begin{subequations} \label{eqn:approx_rates} 
\begin{align}
R_1&\approx H(X_1|X_2)+\sqrt{\frac{[\bV]_{1,1} }{n}}\, \rmQ^{-1}(\epsilon)  \label{eqn:rate1}\\
R_2&\approx R_2^*+\sqrt{\frac{[\bV]_{1,1} }{n}}\,(\tan\theta)\, \rmQ^{-1}(\epsilon). \label{eqn:rate2}
\end{align}
\end{subequations}
%By varying $\theta$, $R_2$ can be made as small as $R_2^*$, 
Notice that if $\theta\notin\{\pm\frac{\pi}{2}\}$ and $n$ is sufficiently large, $R_2$ can be made arbitrarily close to $R_2^*$, whereas the constraint on $R_1$ in \eqref{eqn:rate1} is identical to that in \eqref{eqn:sieda}. That is, up to $O(\frac{\log n}{n})$ terms, achievable rates are the same as if $X_2$ were known perfectly at the decoder. This makes intuitive sense, since $R_2^*>H(X_2)$, so we are in the large deviations regime  for the second source $X_2$ and the error probability for reconstructing $X_2$ vanishes much more quickly than that for $X_1$. In fact, it vanishes  exponentially fast  and the exponent is the almost-lossless source coding error exponent~\cite[Ch.~2]{Csi97}. Similarly, the SI-ED and SW regions do not differ when approaching the horizontal boundary or the sum-rate boundary as in Fig.~\ref{fig:slices}(b).

However, when approaching either of the corner points, as in Fig.~\ref{fig:slices}(c), the situation is more complicated. In particular, the \emph{scalar} perspective on dispersion illustrated in the region in \eqref{eqn:sied} is insufficient, because characterizing the dispersions at the corner points require \emph{off-diagonal terms} of the $\bV$ matrix, as stated in \eqref{eqn:corner_pt} and \eqref{eqn:subplotc}. Intuitively, this is because there are several forces at play---for the $(H(X_1),H(X_2|X_1))$ point, the contribution from the marginal dispersion $[\bV]_{2,2}$, the contribution from the  sum rate dispersion $[\bV]_{3,3}$ and also the  correlation coefficient $\rho_{2,3}$. These interact to give an local dispersion that can only be expressed implicitly as in \eqref{eqn:corner_pt}--\eqref{eqn:subplotc}. Note that now the dispersion depends   on the angle of approach and the correlation  coefficient of  $-\log p_{X_2|X_1}(X_2|X_1)$ and $-\log p_{X_1, X_2}(X_1, X_2)$ namely $\rho_{2,3}$. Hence, the off-diagonal elements of the dispersion matrix $[\bV]_{1,3}$ and $[\bV]_{2,3}$ are required to characterize the dispersion. However, the element $[\bV]_{1,2}$ never appears, because there is no point at the intersection of the vertical and horizontal boundaries of the optimal rate region (for dependent sources), so they are never simultaneously at play. Hence, even though $[\bV]_{1,2}$ is an element of the dispersion matrix, it has no impact on the local dispersion behavior.

It may at first appear that because this non-scalar dispersion behavior occurs at only two points, it is merely a curiosity, but Theorem~\ref{thm:sum_rate_disp} asserts that this is not the case. In particular, because the corner points are the extreme points of the optimal rate region, the behavior in their vicinity is vital to the behavior of the optimal weighted sum-rate. This is evident in the statement of Theorem~\ref{thm:sum_rate_disp}, that in order to characterize the weighted sum-rate dispersion requires off-diagonal terms of the $\bV$ matrix. However, the dispersion for certain pairs $(\alpha,\beta)$ reduces to that for the scalar case. In particular, if $\beta=0$, it is not hard to show (see Appendix~\ref{app:slice}) using Theorem~\ref{thm:sum_rate_disp} that
\begin{equation}\label{eqn:Gsimp1}
G(\epsilon;\alpha,0)=\alpha^2 [\bV]_{1,1}.
\end{equation}
Similarly, if $\alpha=0$, then
\begin{equation}\label{eqn:Gsimp2}
G(\epsilon;0,\beta)=\beta^2 [\bV]_{2,2}.
\end{equation}
Finally, if $\alpha=\beta$, then
\begin{equation}\label{eqn:Gsimp3}
G(\epsilon;\alpha,\alpha)=\alpha^2 [\bV]_{3,3}.
\end{equation}
Because these special cases are the only ones for which the weighted sum-rate over the asymptotic rate region is not uniquely minimized at a corner point, these results agree with the assessment that the SW region does not differ from the SI-ED region away from the corner points.

\begin{figure}
\centering
\includegraphics[width = .73\columnwidth]{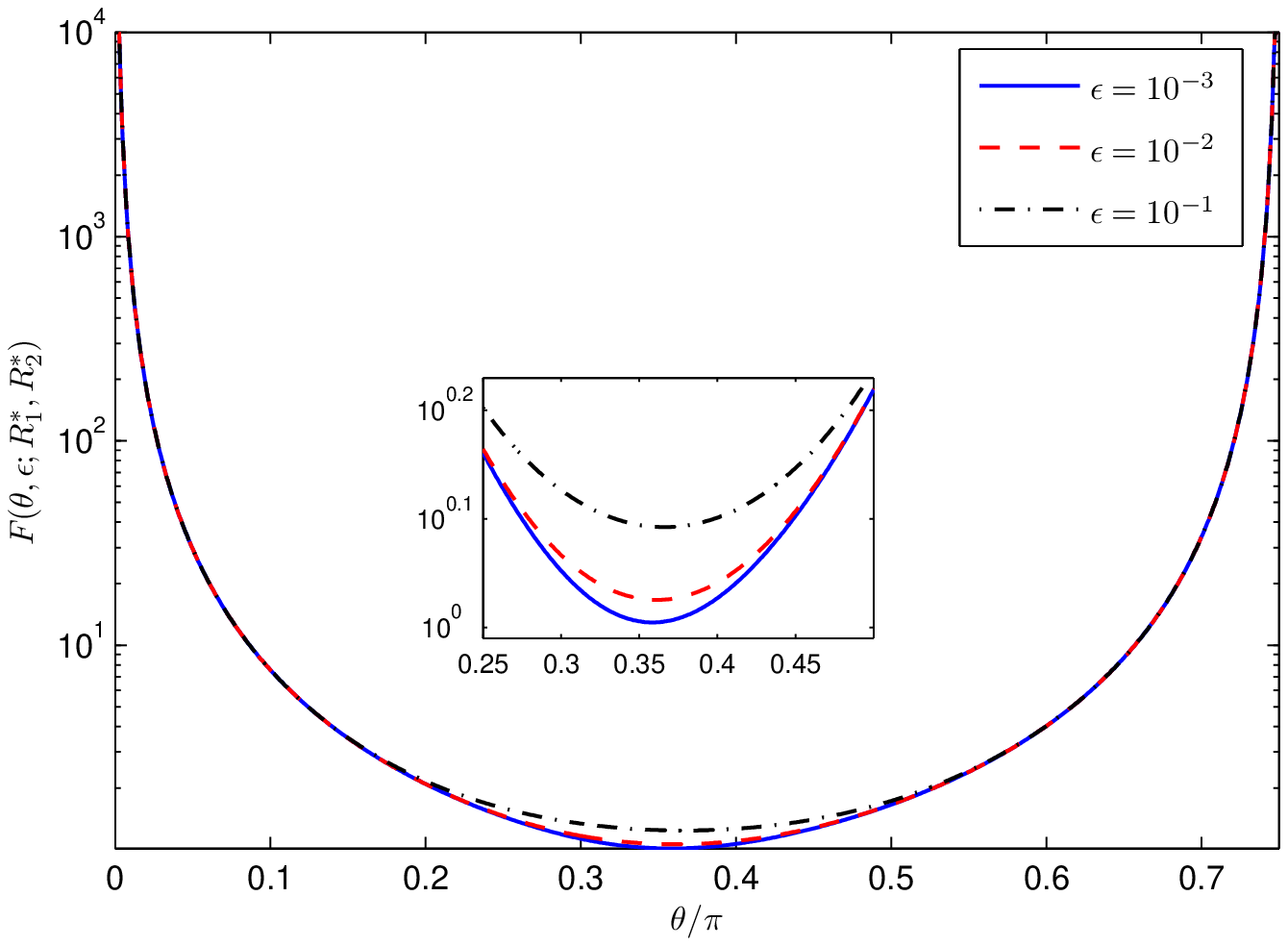}
\caption{Plots of  $F(\theta,\epsilon;R^*_1,R^*_2 )$  against $\theta \in (0, \frac{3\pi}{4})$ for different $\epsilon$'s. This plot  shows the local dispersion as we approach the corner point $(R^*_1,R^*_2)=(H(X_1) , H(X_2|X_1))$ from various angles.  See Fig.~\ref{fig:slices}(c) and the expression for $F(\theta,\epsilon,R^*_1,R^*_2)$ in \eqref{eqn:subplotc}. } 
\label{fig:angle}
\end{figure}

\subsubsection{Comments on Local Dispersion at Corner Points}\label{sec:local_disp_comments}

Interestingly, at corner points  the local dispersions  given by \eqref{eqn:corner_pt}--\eqref{eqn:subplotc}  depends  on $\epsilon$, unlike the corresponding local dispersions  for non-corner points in \eqref{eqn:subplota}--\eqref{eqn:subplotb}.  This is illustrated numerically in Fig.~\ref{fig:angle} for the source $p_{X_1, X_2}=[0.7 \, \, \,  0.1; 0.1\,\, \,  0.1]$ and the   corner point $(R^*_1,R^*_2)=(H(X_1) , H(X_2|X_1))$. Also, it can be seen that as $\theta \downarrow 0$,   $F(\theta,\epsilon;H(X_1),H(X_2|X_1))$ increases without bound. This
agrees with intuition because when $\theta$ is small, we are approaching the corner point   almost
parallel to the horizontal boundary of $\scR_{\mathrm{SW}}(n,\epsilon)$. When $\theta\uparrow\frac{3\pi}{4}$, similarly, we are almost parallel to
the sum rate boundary. 
On the other hand, when $\theta$ is moderate (say $\theta\approx\frac{3\pi}{8}$), the rate pair is further into the interior of $\scR_{\mathrm{SW}}(n,\epsilon)$,
hence the local dispersion is smaller. The constant $\frac{3\pi}{8}$ is in fact not arbitrary because the angle between
the horizontal boundary and the sum rate boundary of $\scR_{\mathrm{SW}}(n,\epsilon)$,  is exactly $\frac{3\pi}{4}$. Hence $\frac{3\pi}{8}$ is the half-angle,
which means that the rate pair is, in a sense, furthest away from either boundary. However, the smallest local dispersion does not
occur at exactly $\theta=\frac{3\pi}{8}$ because of some asymmetry between the entropy densities $-\log p_{X_1|X_2}(X_1|X_2)$ and
$-\log p_{X_1,X_2}(X_1,X_2)$.

The case of approaching a corner point parallel to a boundary line deserves further discussion. One may ask, for example, what trajectories of $(R_{1,n},R_{2,n})$ are achievable that approach the point $(H(X_1),H(X_2|X_1))$ parallel to the horizontal boundary. All we learn  from  Fig.~\ref{fig:angle} and from the characterization of the local dispersion is that $R_{1,n}$ must approach $H(X_1)$ with a perturbation term larger than $O(\frac{1}{\sqrt{n}})$. This can also be observed  for the case where we approach a non-corner point along the horizontal boundary of the SW region. See~\eqref{eqn:hori} with $\theta=0$ in which case $F(\theta,\epsilon; R_1^*, H(X_2|X_1))=\infty$ for any $R_1^*>H(X_1)$. Answering these types of questions would seem to require techniques from moderate deviations \cite{Tan12,sason12, altug10}, and as such it is beyond the scope of this paper. However, we believe that our characterization of the constants on all $O(\frac{1}{\sqrt{n}})$ terms provides a good---if in this sense incomplete---portrait of the second-order behavior, and we defer this more challenging question to future work.

\subsubsection{Singular Entropy Dispersion Matrices} \label{sec:sing}

What are the implications of the $(n,\epsilon)$-SW region (Theorem~\ref{thm:orr}) for singular $\bV$'s?  Note that Theorem~\ref{thm:orr} holds regardless of whether $\bV$ is singular or positive-definite (but not for the trivial case where $\bV=\bzero$ so  we assume throughout that $\rank(\bV)\ge 1$). Sources for which $\bV$ is singular include those which are (i) independent, i.e., $I(X_1; X_2)=0$, (ii) either $X_1$ or $X_2$ is uniform over their alphabets. It is easy to see why $I(X_1;X_2)=0$ results in a singular $\bV$ --- this is because the third entry in the entropy density vector is a linear combination of the first two. Thus $\bV$ loses rank. Case (ii) was analyzed by   Sarvotham et al.~\cite{Sar05} where $X_1,X_2 \in\bbF_2$,   $X_1\sim\Ber(\frac{1}{2})$, $X_2=X_1\oplus N$ with $N\sim\Ber(\zeta), \zeta\in (0, \frac{1}{2})$. The pair of random variables $(X_1, X_2)$ is the so-called {\em discrete symmetric  binary source} (DSBS) with crossover probability $\zeta$. For the DSBS,  Theorem 1 in~\cite{Sar05} asserts that the $(n,\epsilon)$-optimal rate region is (up to terms in $o(\frac{1}{\sqrt{n}})$)  
\begin{equation}
\bR\ge\bH + \sqrt{\frac{V_\zeta}{n}}\, \rmQ^{-1}(\epsilon)\bone , \label{eqn:sar_eq}
\end{equation}
where $V_\zeta$ is a scalar entropy dispersion [to be specified precisely in~\eqref{eqn:Vq}]. Thus, the three inequalities are {\em decoupled}. In contrast, in   Theorem~\ref{thm:orr}, we showed that  the $(n,\epsilon)$-optimal rate region for general DMMSes is such that the constraints on $R_1$, $R_2$ and $R_1+R_2$ are {\em coupled} through the set $\scS(\bV,\epsilon)$. %which depends on the entropy dispersion matrix $\bV(p_{X_1,X_2})$. 

Let us relate \eqref{eqn:sar_eq} to our Theorem~\ref{thm:orr}.  For the DSBS, it can be verified that $\rank(\bV)=1$ and  that $\bV$ is a scalar multiple of the all ones matrix, i.e.,  $\bV=V_\zeta \bone_{3\times 3}$ where $V_\zeta = \var(-\log p_{X_1|X_2}(X_1|X_2))=\var(-\log p_{X_2|X_1}(X_2|X_1))= \var(-\log p_{X_1,X_2}(X_1,X_2))$. Intuitively, this is because there is only one {\em degree of freedom} in a DSBS with crossover probability $\zeta$.  The parameter  $V_\zeta$ is exactly the scalar dispersion in~\eqref{eqn:sar_eq}. In fact, it can be calculated in closed-form for the DSBS with crossover probability $\zeta$ as 
\begin{equation}
V_\zeta = \zeta(1-\zeta) \left[ \log \left( \frac{1-\zeta}{\zeta} \right)\right]^2.  \label{eqn:Vq}
\end{equation}
For this source, since $\rank(\bV)=1$,  all the probability mass of the degenerate Gaussian $\calN(\bzero, \bV)$ lies in a subspace of dimension one.  Therefore,  it is easy to see that  $\scS(\bV, \epsilon) $ defined in \eqref{eqn:SVset} degenerates to the axis-aligned cuboid
\begin{equation}
\scS(\bV, \epsilon) = \big\{\bz\in \bbR^3: z_t\ge \sqrt{V_\zeta} \, \rmQ^{-1}(\epsilon),\forall\, 1\le t\le 3  \big\}. \label{eqn:SVset_rank1}
\end{equation}
The quantity $\sqrt{\frac{V_\zeta}{n}} \, \rmQ^{-1}(\epsilon)$ is approximately the    rate redundancy  \cite{Baron04, Sar05b,Sar05, He09} for fixed-length SW coding  for a DSBS. In this case,  the inner and outer bounds of the $(n,\epsilon)$-optimal rate region   degenerate  to \eqref{eqn:sar_eq}.   Thus the  fixed-length  results in~\cite{Baron04,Sar05b, Sar05, He09} are special cases of our general result. 
 This argument for singular dispersion matrices can be formalized and we do so in the latter half of the proof of Theorem~\ref{thm:vrrt}.

In fact, when the dispersion matrix $\bV$ is singular, the conclusions resulting from Theorems~\ref{thm:local_disp} and~\ref{thm:sum_rate_disp} become considerably simpler. Let us illustrate this on Theorem~\ref{thm:local_disp} with the DSBS with crossover probability $\zeta$ defined above. Clearly, for Theorem~\ref{thm:local_disp}, the conclusions in~\eqref{eqn:subplota}--\eqref{eqn:subplotb} stay the same. However, \eqref{eqn:corner_pt} and \eqref{eqn:subplotc} simplify to the following closed-form expressions:
\begin{align}
F(\theta,\epsilon; H(X_1|X_2), H(X_2) )  &= \frac{V_\zeta}{\cos^2 \theta}, \qquad\mbox{and} \label{eqn:singular1}\\*
F(\theta,\epsilon; H(X_1 ), H(X_2|X_1) )  &= \frac{V_\zeta}{\sin^2 \theta} .\label{eqn:singular2}
\end{align}
This is because $\rho_{1,3}=\rho_{2,3}=1$ and all elements of $\bV$ are identically equal to $V_\zeta$ so the two-dimensional analogue of the $\rmQ$ function, namely the $\Psi$ function  defined in~\eqref{eqn:Psi}, degenerates to $\rmQ$ function evaluated at the second argument. For example, the $\Psi$ function in~\eqref{eqn:corner_pt} becomes
\begin{equation}
\lim_{\rho\uparrow 1}\Psi \left(\rho;-\sqrt{\frac{F}{ V_\zeta }} \, (\cos\theta)\, \rmQ^{-1}(\epsilon), -\sqrt{\frac{F}{V_\zeta}} \, (\cos\theta+\sin\theta)\, \rmQ^{-1}(\epsilon) \right) =\rmQ\left( -\sqrt{\frac{F}{V_\zeta}} \, (\cos\theta)\right),
\end{equation}
which when equated to $1-\epsilon$ yields~\eqref{eqn:singular1}.

\section{Dispersion of  the Multiple-Access Channel} \label{sec:mac}
The multiple-access channel or MAC is the channel coding dual to the Slepian-Wolf problem described in Section~\ref{sec:sw} \cite[Sec.\ 3.2]{Csi97}. The MAC model has found numerous applications, especially in wireless communications where multiple parities would like to communicate to a single base station reliably.  For a MAC, there are two (or more)  independent messages $M_1 \in [2^{nR_1}]$ and $M_2 \in [2^{nR_2}]$. The two messages, which are uniformly distributed over their respective message sets, are {\em separately} encoded into sequence codewords $X_1^n \in \calX_1^n$ and $X_2^n \in \calX_2^n$ respectively. These codewords are  the inputs to a discrete memoryless multiple-access channel (DM-MAC) $W:\calX_1\times\calX_2 \to\calY$. The decoder receives $Y^n$ from the output of the DM-MAC and provides estimates of the messages $\hatM_1$ and $\hatM_2$ or declares that a decoding error has occurred. It is usually desired to send both messages {\em reliably}, that is, to ensure that the average probability of error 
\begin{equation}
\Pen := \rvP  ( \{\hatM_1 \ne M_1 \}\cup \{\hatM_2\ne M_2 \}  ) \label{eqn:err_mac}
\end{equation}
tends to zero as $n\to\infty$. The set of achievable rates, or the {\em capacity region} $\scC_{\mathrm{MAC}}^*$,  is given by 
\begin{align} 
R_1&\le I(X_1;Y|X_2, Q) \nn \\*
R_2&\le I(X_2;Y|X_1, Q) \nn \\*
R_1+R_2&\le I(X_1, X_2;Y| Q)  \label{eqn:mac_reg}
\end{align}
for some $p_Q$, $p_{X_1|Q}$ and  $p_{X_2|Q}$ and $|\calQ|\le 2$. This  asymptotic result was proved independently by Ahlswede~\cite{ahl71} and Liao~\cite{liao} and can be written in an alternative form which involves taking the convex hull instead of the introduction of the auxiliary time-sharing variable $Q$. See \cite{elgamal} for further discussions. A somewhat surprising result in the theory of MACs, which differs from point-to-point channel coding, is that the capacity region for average probability of error is strictly larger than that for maximal probability of error~\cite{dueck}.  We emphasize that we focus on the average  probability of error defined in \eqref{eqn:err_mac} throughout. Note that as with the SW case, we can consider  $\rvP( \hatM_1 \ne M_1)$, $\rvP( \hatM_2\ne M_2 )$ and $\Pen$ separately and place upper bounds on each of these constituent error probabilities but, for simplicity, we consider only $\Pen$  in \eqref{eqn:err_mac}.

In this section, we prove an inner bound  to $\scC_{\mathrm{MAC}}^*(n,\epsilon)$, the $(n,\epsilon)$-capacity region, when $n$ is large. This inner bound illustrates the second-order behavior of the   capacity region in~\eqref{eqn:mac_reg}. We propose a coding scheme for a block of length $n$ that satisfies  $\Pen\le \epsilon$ for $n$ sufficiently large.  Our encoding scheme is  the coded time-sharing procedure  by Han and Kobayashi~\cite{Han81}. The decoding scheme is similar to MMI decoding~\cite{Goppa}. However, the  error probability analysis is rather different.  The result we present here is a {\em global dispersion} one (in the sense of Theorem~\ref{thm:orr}) but one can define similar notions of local  (Theorem~\ref{thm:local_disp}) and sum-rate dispersion  (Theorem~\ref{thm:sum_rate_disp}) for this and the ABC problem in the following section. In a similar manner to SW, our results lead naturally to inner bounds. We do not pursue this for the MAC and ABC problems as the analysis turns out to be similar to Theorems~\ref{thm:local_disp} and~\ref{thm:sum_rate_disp}.

%As with the   SW region, note that the capacity region in \eqref{eqn:mac_reg} is a polyhedron  (in fact a pentagon) for  a given set of time-sharing and input distributions $p_Q$, $p_{X_1|Q}$ and  $p_{X_2|Q}$. While we  currently do not have a full characterization of the $(n,\epsilon)$-capacity region and proving a converse appears to be difficult, we remark that, somewhat surprisingly, the boundary of the inner bound for  fixed   $p_Q$, $p_{X_1|Q}$ and  $p_{X_2|Q}$ is  {\em curved} if the so-called information dispersion matrix, to be defined in \eqref{eqn:macv}, is non-singular.  This is likened to the $(n,\epsilon)$-SW region (Theorem~\ref{thm:orr}) and the information dispersion matrix  is analogous to the entropy dispersion matrix defined in \eqref{eqn:swv} for the SW problem.   As in the SW case,  the main tool  that we use is a multidimensional  version of the Berry-Ess\`{e}en theorem~\cite{Ben03} and the vector rate redundancy theorem (Theorem~\ref{thm:vrrt}). We focus on the two-sender case. Generalizations to more than two senders are straightforward. We begin with definitions then state the  main result. We also provide some hints on how a possible converse theorem could be proved after the main achievability result is stated.

\subsection{Definitions} \label{sec:defmac}
Let $(\calX_1,\calX_2,W, \calY)$ be a DM-MAC, i.e., for any input  codeword sequences $x_1^n \in\calX_1^n$ and $x_2^n \in\calX_2^n$, 
\begin{equation}
W^n( y^n |x_1^n, x_2^n) =\prod_{k=1}^n W(y_k | x_{1k}, x_{2k} ).
\end{equation}
% We assume all alphabets $\calX_1, \calX_2, \calY$ are finite. 

\begin{definition}
An {\em $(n, 2^{nR_1}, 2^{nR_2}, \epsilon)$-code for the DM-MAC}  $(\calX_1,\calX_2,W, \calY)$  consists of two encoders $f_{j,n}: \calM_{j} =[2^{nR_j}] \to \calX_j^n, j=1,2$,   and a decoder $\varphi_n:\calY^n\to\calM_{1}\times \calM_{2}$ such that the  average  error probability defined in \eqref{eqn:err_mac} does not exceed $\epsilon$. Note that the outputs of the encoders are $f_{j,n}(M_j),j=1,2$ and the output of the decoder are the estimates $(\hatM_1,\hatM_2) =\varphi_n(Y^n)$.   The {\em coding rates} are  defined in the usual way.% as in~\eqref{eqn:def_rates}.
\end{definition}
\begin{definition} \label{def:neps_ach_mac}
A rate pair $(R_1 ,R_2 )$ is  {\em  $(n,\epsilon)$-achievable}  if there exists an   $(n, 2^{nR_1}, 2^{nR_2}, \epsilon)$-code for the DM-MAC   $(\calX_1,\calX_2,W, \calY)$.   The {\em  $(n,\epsilon)$-capacity region} $\scC_{\mathrm{MAC}}^*(n,\epsilon) \subset\bbR^2$ is the set of all  $(n,\epsilon)$-achievable  rate pairs. 
\end{definition}
%Again traditionally, capacity regions are defined with an additional closure operation~\cite{elgamal}. However, we refrain from doing so for the $(n,\epsilon)$-capacity region. Also, i
In contrast to the asymptotic setting, it is not obvious that $\scC_{\mathrm{MAC}}^*(n,\epsilon)$ is convex. The usual Time Sharing argument~\cite[Lemma 3.2.2]{Csi97} --- that the juxtaposition of two good multiple-access codes leads to a good but longer code --- does not hold because the blocklength is constrained to be a fixed integer $n$ so juxtaposition is not allowed.   Fix  a triple of  distributions $p_Q(q)$, $p_{X_1|Q}(x_1|q)$ and  $p_{X_2|Q}(x_2|q)$.  Given the channel $W$, these distributions induce the  following output conditional distributions
\begin{align}
p_{Y|X_2,Q}(y|x_2,q) &:= \sum_{x_1 }  p_{X_1|Q}(x_1|q) W(y|x_1,x_2) \label{eqn:output1} \\*
%p_{Y|X_1,Q}(y|x_1,q) &:= \sum_{x_2 } p_{X_2|Q}(x_2|q)W(y|x_1,x_2) \label{eqn:output2} \\*
p_{Y| Q}(y| q) &:= \sum_{x_1,x_2}  p_{X_1|Q}(x_1|q)p_{X_2|Q}(x_2|q)W(y|x_1,x_2). \label{eqn:output3}
\end{align}
The output conditional distribution $p_{Y|X_1,Q}$ is defined similarly to $p_{Y|X_2,Q}$ with $1$ replaced by $2$ and vice versa.

\begin{definition}
The {\em  information density vector} is defined as 
\begin{equation}
\bi( Q,X_1, X_2, Y ) := \begin{bmatrix} \log   [ W(Y|X_1, X_2 ) / p_{Y|X_2,Q}(Y|X_2,Q) ] \\ \log [  {W(Y|X_1, X_2)}/{p_{Y|X_1,Q}(Y|X_1,Q) }] \\ \log [ {W(Y|X_1, X_2)}/{p_{Y|Q}(Y| Q) }]  \end{bmatrix} . \label{eqn:info_dens}
\end{equation}
where the distributions $p_{Y|X_2,Q}, p_{Y|X_1,Q}$ and $p_{Y|Q}$ are defined in \eqref{eqn:output1}--\eqref{eqn:output3}. The random variables $(Q, X_1,X_2, Y)$ have joint distribution $p_{Q}p_{X_1|Q} p_{X_2|Q} W$. 
\end{definition}
%Note that in  the definition of $\bi( Q,X_1, X_2, Y )$,  we have not made the dependence on the  input distributions $p_Q(q)$, $p_{X_1|Q}(x_1|q)$ and  $p_{X_2|Q}(x_2|q)$ explicit but  information density vector   $\bi( Q,X_1, X_2, Y )$ indeed depends on the these input distributions  through the output distributions in \eqref{eqn:output1}--\eqref{eqn:output3}.  
Observe that the expectation of the information density vector with respect to $p_Q p_{X_1|Q} p_{X_2|Q} W$ is  the vector of mutual information quantities in~\eqref{eqn:mac_reg}, i.e., 
\begin{align}
\rvE [\bi( Q,X_1, X_2, Y ) ]& = \bI (  p_Q, p_{X_1|Q}, p_{X_2|Q}, W)   := \begin{bmatrix} I(X_1;Y|X_2,Q)\\I(X_2;Y|X_1,Q)\\ I(X_1, X_2;Y|Q) \end{bmatrix}. 
\end{align}
Let $I_t( p_Q, p_{X_1|Q}, p_{X_2|Q}, W) $ be the $t$-th entry of $ \bI (  p_Q, p_{X_1|Q}, p_{X_2|Q}, W)$.  As in \eqref{eqn:kappa_sw}, let  
\begin{equation}
\kappa:= \max_{1\le t\le 3} \left\|\nabla_p^2 I_t(p) \right\|_2 \label{eqn:kappa_mac}
\end{equation}
where $p:=p_{Q, X_1, X_2, Y} =p_Qp_{X_1|Q} p_{X_2|Q}W $. 
\begin{definition} 
 The {\em information  dispersion matrix}   $\bV(  p_Q, p_{X_1|Q}, p_{X_2|Q}, W )$ is the covariance matrix  of the random vector $\bi( Q,X_1, X_2, Y )$ i.e., 
 \begin{equation}
\bV(  p_Q, p_{X_1|Q}, p_{X_2|Q}, W ) = \cov( \bi( Q,X_1, X_2, Y )  ). \label{eqn:macv}
\end{equation}
\end{definition}
If there is no risk of confusion, we abbreviate the deterministic vector  $ \bI (  p_Q, p_{X_1|Q}, p_{X_2|Q}, W) \in \bbR^3$ and  the deterministic matrix $\bV(  p_Q, p_{X_1|Q}, p_{X_2|Q}, W  ) \succeq 0$ as $\bI$ and $\bV$ respectively.  We assume throughout that the channel and the input distributions are such that $\rank(\bV)\ge 1$, i.e., $\bV$ is not the all-zeros  matrix. Recall the definition of the rate vector $\bR = [R_1, R_2, R_1+R_2]^T$ in \eqref{eqn:rate_vec}. 

\begin{definition} \label{def:inner_bd_mac}
Given triple of input distributions  $ (p_Q, p_{X_1|Q}, p_{X_2|Q})$, define the region $\scR(n,\epsilon; p_Q,p_{X_1|Q}, p_{X_2|Q}  ) \subset\bbR^2$ to be  the set of rate pairs $(R_1,R_2)$ that satisfy
\begin{equation}
\bR \in \bI  - \frac{1}{\sqrt{n}}  \scS (\bV  ,  \epsilon) -  \frac{\nu\log n}{n} \bone ,\label{eqn:ach_mac}
\end{equation} 
where $\nu:= |\calQ||\calX_1| |\calX_2| |\calY|+ \kappa+\frac{3}{2}$. In~\eqref{eqn:ach_mac}, $\bI :=  \bI (  p_Q, p_{X_1|Q}, p_{X_2|Q}, W)$ and $\bV := \bV(  p_Q, p_{X_1|Q}, p_{X_2|Q}, W  )$ and   the set $\scS (\bV  ,  \epsilon) \subset\bbR^3$ is defined in  \eqref{eqn:SVset}.
\end{definition}
%Plots of   $\scR(n,\epsilon; p_Q,p_{X_1|Q}, p_{X_2|Q}  )$ are shown in Fig.~\ref{fig:mac} for different blocklengths $n$ and error probabilities $\epsilon$. 
\subsection{Main Result and Interpretation} \label{sec:mac_interp}
 
\begin{theorem}[Global Dispersion   for DM-MAC] \label{thm:mac} Let $\epsilon\in (0,1)$. The $(n,\epsilon)$-capacity  region $\scC_{\mathrm{MAC}}^*(n,\epsilon)$  for the DM-MAC satisfies
\begin{equation}
\bigcup_{p_Q, p_{X_1|Q}, p_{X_2|Q}}  \scR(n,\epsilon; p_Q,p_{X_1|Q}, p_{X_2|Q}  )  \subset \scC_{\mathrm{MAC}}^*(n,\epsilon)  \label{eqn:rr_mac}
\end{equation}
for all $n$ sufficiently large. Furthermore,   to preserve $\bI$ and $\bV$, the union over $p_Q$ can be restricted to those discrete distributions with support $\calQ$ whose  cardinality $|\calQ|\le 9$.   The inner bound is also universally attainable.
\end{theorem} 
This theorem is proved   in Section~\ref{sec:prf_mac}.  The bounds on cardinality can be proved using the support lemma~\cite[Theorem 3.4]{Csi97}. See Section~\ref{sec:card}. 
%The inner bound is illustrated for various $n$'s and $\epsilon$'s in Fig.~\ref{fig:mac}.   It is  relatively   straightforward  to extend the result to the case where there is a cost  constraint on the codewords, i.e., 
%\begin{equation}
%\frac{1}{n}\sum_{k=1}^n \Lambda_j(x_{jk}(m_j) )\le \Gamma_j, 
%\end{equation}
%for $j=1,2$ and  all $(m_1 , m_2) \in \calM_1\times\calM_2$. We omit the statement and proof. In Section~\ref{sec:universal_mac}, we comment on how the coding scheme  can be modified to deal with channels with arbitrary input and output alphabets at the cost of universality in decoding.  %In addition, the analogue of Proposition~\ref{prop:slice} on the  dispersion along specific slices of the $(n,\epsilon)$-inner bound   holds for a fixed $(p_Q, p_{X_1|Q}, p_{X_2|Q})$. We   omit  the statement since it is similar to the SW case. 
From~\eqref{eqn:ach_mac}, we see that the  inner bound to the $(n,\epsilon)$-capacity region $\scC_{\mathrm{MAC}}^*(n,\epsilon)$ approaches the usual MAC region~\eqref{eqn:mac_reg} at a rate of $O(\frac{1}{\sqrt{n}})$ for fixed input distributions. Unsurprisingly, this rate is a consequence of the multidimensional central limit theorem. The {\em redundancy set} $\frac{1}{\sqrt{n}}\scS(\bV,\epsilon)$  in~\eqref{eqn:ach_mac} is approximately the loss in rate to the three mutual information quantities in \eqref{eqn:mac_reg} one must incur when operating at blocklength $n$ and with average error probability $\epsilon$. 

For the proof of Theorem~\ref{thm:mac}, we use the coded time-sharing scheme introduced by Han and Kobayashi in their seminal work on interference channels \cite{Han81}. The decoding step, however, is novel and is a modification of the maximum mutual information (MMI) decoding rule~\cite{Goppa, Csi97}. This MMI-decoding step allows us to define a new notion of typicality for empirical mutual information quantities.   Interestingly, the error event that contributes to the $\epsilon$ probability of error is the one in which the transmitted pair of  codewords $x_1^n(m_1), x_2^n(m_2)$ is not  jointly typical (in a refined sense of typicality) with the output of the channel $y^n$ (and a time-sharing sequence $q^n$).   The probabilities of the other  error events --- that there exists another codeword jointly typical with the output --- can be shown to be vanishingly small relative to $\epsilon$. Intuitively, this is because we are operating close to  the boundaries of the rate region for given input distributions, i.e., at very high rates. The sphere-packing argument~\cite{Csi97, gallagerIT, Forney} implies that the dominant (typical) error events at high rates are of the form where a large number of incorrect codewords are jointly typical with the transmitted one, i.e., what Forney calls Type I error~\cite{Forney}. Thus, the probability of error is dominated by an atypically large noise  event and expurgation does not improve the exponents. 

\subsection{Difficulties In The Converse}
A converse (outer bound to $\scC_{\mathrm{MAC}}^*(n,\epsilon)$) has unfortunately remained elusive.  To the best of the authors' knowledge, there are three  strong converse proof techniques  for the average probability of error  of the DM-MAC. The first is by Han~\cite[Lemma 7.10.2]{Han10} \cite[Lemma~4]{Han98}  and is based on information spectrum ideas. Applying it is difficult because the specified input distributions   $p_{X_1^n}, p_{X_2^n}$ are the Fano-distributions on the codewords.\footnote{Given DM-MAC codebooks $\calC_j := \{x_j^n(m_j): m_j\in \calM_j\},j=1,2$,  the {\em Fano-distribution}  $p_{X_j^n}$ is   the uniform distribution over~$\calC_j$.} Since $p_{X_j^n}$ does not decompose into independent factors, the Berry-Ess\`{e}en theorem is not directly applicable.  The second is by  Dueck~\cite{dueck81} who used the blowing-up lemma~\cite[Sec.\ 1.5]{Csi97}. The third and   most promising  technique   is by  Ahlswede~\cite{Ahl82} who built  on Dueck's work~\cite{dueck81}. Ahlswede first  applies Augustin's strong converse for DMCs~\cite{augustin} to the so-called Fano$^*$-distribution\footnote{The {\em Fano$^*$-distribution} is  $p_{X_j^n} = \prod_{k=1}^n p_{X_{jk}}$ with $p_{X_{jk}} (a) := |\calM_j|^{-1}| \{ m_j  : x_{jk}(m_j) = a\}|$ for all $a\in\calX_j$. }  which factorizes. Then, he obtains a region that resembles the capacity region for the DM-MAC. Finally, he utilizes a  {\em wringing} technique to remove (or {\em wring out}) the dependence between $X_1$ and $X_2$. Unfortunately, it appears   that the use  of both the blowing-up lemma and the  wringing technique  results in estimates  of an outer bound that are too loose to match the $O(\frac{1}{\sqrt{n}})$ dispersion term in the inner bound in Theorem~\ref{thm:mac}. Another major obstacle to proving a global dispersion-style converse is the need to introduce the time-sharing variable $Q$ or the convex hull operation judiciously. Hence, we believe that genuinely new  strong converse techniques for the DM-MAC   (and other multi-user problems) have to be developed to prove a tight outer bound that matches (or approximately matches) our inner bound in Theorem~\ref{thm:mac}.

 % The {\em Fano$^*$-output distribution}  is simply as $q_{Y^n}(y^n) := \sum_{x_1^n, x_2^n} W^n(y^n|x_1^n, x_2^n) p_{X_1^n} (x_1^n) p_{X_2^n}(x_2^n) $.

\section{Dispersion of  the Asymmetric Broadcast Channel} \label{sec:abc}
We now turn our attention to the broadcast channel \cite{cover}, which is another fundamental problem in network information theory. Despite more than 40 years of research, the capacity region has resisted attempts at proof.   One special instance in which the capacity is known is the so-called {\em asymmetric broadcast channel} or ABC \cite{ kor77}. The ABC   is also known as the broadcast channel with {\em degraded message sets}. % In fact, the analysis in this section applies to the general broadcast channel. We focus on the ABC for concreteness.  

In the ABC problem, there are two independent messages $M_1 \in [2^{nR_1}]$ and $M_2 \in [2^{nR_2}]$ at the sender. These two messages, which are uniformly distributed over their respective message sets, are encoded into a codeword $X^n \in \calX^n$. These codewords are then the inputs to a discrete memoryless asymmetric broadcast channel (DM-ABC)  $W: \calX \to\calY_1\times\calY_2$.  Decoder 1 receives $Y_1^n$ and estimates {\em both} messages $M_1$ and $M_2$, while decoder 2 receives $Y_2^n$ and estimates {\em only} $M_2$. Let the estimates of the messages at decoder 1 be denoted at $(\hatM_1,\hatM_2)$ and let the estimate of message 2 at decoder 2 be denoted as $\check{M}_2$.    The average error probability is  defined as 
\begin{align}
\Pen  :=   \rvP  ( \{\hatM_1 \ne M_1 \}\cup \{\hatM_2\ne M_2 \}\cup\{\check{M}_2\ne M_2\}  ), \label{eqn:bc1}  
\end{align}
Note that the error error event above corresponds to receiver 1 not decoding either message correctly {\em or} receiver 2 not decoding her intended message $M_2$ correctly. An alternative formulation, which turns out to be more challenging, would be to define average probabilities of error for receiver  1 and receiver  2 and to  put {\em different} upper bounds on these error probabilities. 

Returning to our setup,  it usually is desired to drive $\Pen$, defined in \eqref{eqn:bc1}, to zero as the blocklength $n\to\infty$. The set of achievable rate pairs $(R_1, R_2)$ first derived by K\"{o}rner and  Marton~\cite{kor77} is then given by the region
\begin{align}
R_1& \le I(X;Y_1|U) \nn\\*
R_2& \le I(U;Y_2) \nn\\*
R_1+R_2& \le I(X;Y_1) \label{eqn:abc_reg}
\end{align}
for some $p_{U, X}(u, x)$ where $|\calU|\le |\calX|+1$, i.e.,  $U-X-(Y_1, Y_2)$ form a Markov chain in that order. The proof for the direct part uses  the  superposition coding technique~\cite{cover}. The auxiliary variable $U$ basically plays the role of the cloud center while the input random variable $X$ plays the role of a satellite codeword centered at the cloud center $U$.  A weak converse can be proved using the Csisz\'{a}r-sum-identity~\cite{elgamal}. For a strong converse, see~\cite[Sec 3.3]{Csi97} or the  original work by K\"{o}rner and  Marton~\cite{kor77}.

%In this section, we again depart from the traditional asymptotic setting. More specifically, we fix a (finite) blocklength $n$ and  a tolerable upper bound on the average error probability in~\eqref{eqn:bc1}, say $\epsilon$.  We attempt to characterize    the  so-called $(n,\epsilon)$-capacity region $\scC_{\mathrm{ABC}}^*(n,\epsilon)$, i.e.,  the set of $(n,\epsilon)$-achievable rate pairs $(R_1, R_2)$ for the ABC $W$. As with the DM-MAC, a rate pair $(R_1, R_2)$  is said to be $(n,\epsilon)$-achievable for the ABC $W$ if there exists a code, i.e., an encoder and two decoders operating on blocks of symbols of length $n$  for which $\Pen \le \epsilon$.  Precise definitions will be provided in Section~\ref{sec:defbc}.

We show in this section that the tools we have developed for SW coding and the MAC, such as  the vector rate redundancy theorem,  are versatile enough for us to provide  an inner bound to $\scC_{\mathrm{ABC}}^*(n,\epsilon)$, the $(n,\epsilon)$-capacity region, when $n$ is large. We again provide a global dispersion result that is analogous to Theorems~\ref{thm:orr} and \ref{thm:mac}.  Our coding scheme is based on superposition coding~\cite{cover} but the analysis is somewhat different and uses a variant of MMI-decoding. Like the DM-MAC, all three inequalities that characterize the capacity region  in~\eqref{eqn:abc_reg} are ``coupled'' through an   information dispersion matrix for a given input distribution $p_{U,X}$. Thus, the main result  in this section is  conceptually very similar to that for the DM-MAC. And as with the DM-MAC, we do not yet have an outer bound for this problem but we note that strong converses for this problem are available~\cite{kor77,bouch00}. We start with relevant definitions.

\subsection{Definitions} \label{sec:defbc}
Let $(\calX,W, \calY_1,\calY_2)$ be a 2-receiver DM-ABC. That is given an input codeword sequence $x^n \in\calX^n$, 
\begin{equation}
W^n (y_1^n, y_2^n |x^n ) = \prod_{k=1}^n W(y_{1k}, y_{2k} |x_k).
\end{equation}
%The sets $\calX,\calY_1,\calY_2$ are finite. Also, 
We will use the notations $W_1$ and $W_2$ to denote the $\calY_1$- and $\calY_2$-marginal of $W$ respectively, i.e., $W_1(y_1|x) := \sum_{y_1} W(y_1, y_2|x)$ and similarly for $W_2(y_2|x)$.

\begin{definition}
An {\em $(n, 2^{nR_1}, 2^{nR_2}, \epsilon)$-code for the DM-ABC}  $(\calX, W, \calY_1, \calY_2)$  consists of one encoder $f_{n}: \calM_1\times\calM_2 =[2^{nR_1}]\times  [2^{nR_2}] \to \calX^n$,   and two decoders $\varphi_{1,n}:\calY_1^n\to\calM_{1}\times\calM_2$ and $\varphi_{2,n}:\calY_2^n\to\calM_2$ such that the  average  error probability defined in \eqref{eqn:bc1} does not exceed $\epsilon$. Note that the output  of the encoder is $f_{n}(M_1, M_2)$ and the output of the decoders are the estimates $(\hatM_1,\hatM_2) =\varphi_{1,n}(Y_1^n)$ and $\check{M}_2=\varphi_{2,n}(Y_2^n)$.  % The {\em coding rates} are  defined in the usual way as in~\eqref{eqn:def_rates}.
\end{definition}

\begin{definition} \label{def:neps_ach_bc}
A rate pair $(R_1 ,R_2 )$ is  {\em  $(n,\epsilon)$-achievable}  if there exists an   $(n, 2^{nR_1}, 2^{nR_2}, \epsilon)$-code for the DM-ABC  $(\calX, W, \calY_1,\calY_2)$.   The {\em  $(n,\epsilon)$-capacity region} $\scC_{\mathrm{ABC}}^*(n,\epsilon) \subset\bbR^2$ is the set of all  $(n,\epsilon)$-achievable  rate pairs. 
\end{definition}

Fix  an input distribution $p_{U,X} \in \scP(\calU\times\calX)$ where  the auxiliary random variable $U$ takes values on some finite set $\calU$.  Given the channel $W$ and input distribution $p_{U,X}$, the  following   distributions are defined as:
\begin{align}
p_{Y_j|U}(y_j|u) &= \sum_{x} W_j(y_j|x) p_{X|U}(x|u) ,  \label{eqn:output1_bc} \\*
p_{Y_j} (y_j)&=\sum_x W_{j}(y_j|x)p_X(x), \qquad j = 1,2. \label{eqn:output3_bc}
\end{align}

\begin{definition}
The {\em information density vector}  for the ABC is defined as 
\begin{equation}
\bi( U,X,  Y_1, Y_2 ) := \begin{bmatrix} \log   [ W_1(Y_1|X) / p_{Y_1|U}(Y_1|U) ] \\ \log [{ p_{Y_2|U}(Y_2|U) }/{ p_{Y_2}(Y_2)}]  \\ \log [ W_1(Y_1|X_1) / p_{Y_1}(Y_1) ]    \end{bmatrix} . \label{eqn:info_dens_bc}
\end{equation}
where the distributions $p_{Y_1|U}, p_{Y_1}, p_{Y_2|U}, p_{Y_2} $ are defined in \eqref{eqn:output1_bc} and \eqref{eqn:output3_bc} respectively.  The random variables $(U, X, Y_1, Y_2)$ have joint distribution $p_{U,X} W$. 
\end{definition}

%Note that in  the definition of $\bi( U,X,  Y_1, Y_2 )$,  we have not made the dependence on the  input distribution  $p_{U, X}$ explicit but  information density vector   $\bi( U,X,  Y_1, Y_2 )$ indeed depends on the the  input distribution     through the output distributions in \eqref{eqn:output1_bc} and \eqref{eqn:output3_bc}.  
Observe that the expectation of the information density vector with respect to $p_{U,X}W$ is  the vector of mutual information quantities, i.e., 
\begin{equation}
\rvE [\bi( U,X,  Y_1, Y_2  ) ] = \bI (  p_{U,X}, W ) := \begin{bmatrix} I(X;Y_1|U )\\I(U;Y_2 )\\ I(X;Y_1) \end{bmatrix}. 
\end{equation}

\begin{definition}
 The {\em information  dispersion matrix}   $\bV( p_{U,X}, W )$ is the covariance matrix of the random vector $\bi( U,X,  Y_1, Y_2  )$ i.e., 
 \begin{equation}
\bV(  p_{U,X}, W ) = \cov( \bi( U,X,  Y_1, Y_2 )  ). \label{eqn:bcv}
\end{equation}
\end{definition}
As with the SW and MAC cases, we usually abbreviate $ \bI (  p_{U,X}, W )$ and $\bV( p_{U,X}, W )$ as $\bI$ and $\bV$ respectively.  We will again use the definition of the rate vector $\bR = [R_1, R_2, R_1+R_2]^T$ in \eqref{eqn:rate_vec}. 

\begin{definition} \label{def:reg_bc}
Given an input distribution $p_{U,X}$, define the region $\scR(n,\epsilon; p_{U,X}  ) \subset\bbR^2$ to be  the set of rate pairs $(R_1,R_2)$ that satisfy
\begin{equation}
\bR \in \bI  - \frac{1}{\sqrt{n}}  \scS (\bV  ,  \epsilon) -\frac{ \nu\log n}{n}\bone ,\label{eqn:ach_bc}
\end{equation} 
where $\nu:=|\calU ||\calX| \max\{ |\calY_1|, |\calY_2| \}+\kappa +\frac{3}{2}$ and $\kappa$ is defined similarly that  for the DM-MAC problem (see~\eqref{eqn:kappa_mac}). Here $\bI :=  \bI (   p_{U,X}, W )$ and $\bV := \bV(  p_{U,X}, W  )$ and   the set $\scS (\bV  ,  \epsilon) \subset\bbR^3$ is defined in  \eqref{eqn:SVset}.
\end{definition}
\subsection{Main Result and Interpretation}

\begin{theorem}[Global Dispersion   for the DM-ABC] \label{thm:bc} Let $\epsilon\in (0,1)$. The $(n,\epsilon)$-capacity  region $\scC_{\mathrm{ABC}}^*(n,\epsilon)$  for the DM-ABC satisfies
\begin{equation}
\bigcup_{p_{U }, p_{X|U} }  \scR(n,\epsilon; p_{U,X}  )  \subset \scC_{\mathrm{ABC}}^*(n,\epsilon)  \label{eqn:rr_bc}
\end{equation}
 for all $n$ sufficiently large. Furthermore, to preserve $\bI $ and $\bV$, the union over $p_{U}$ can be restricted to those discrete distributions with support $\calU$ whose  cardinality $|\calU|\le |\calX|+ 6$.  The inner bound is also universally attainable.
\end{theorem} 
The proof of this result can be found in Section~\ref{sec:prf_abc}.% Again the focus is on the case where $\bV(  p_{U,X}, W  )$ is positive-definite. 

Conceptually, this result is very similar to that for the SW problem (Theorem~\ref{thm:orr}) and the DM-MAC (Theorem~\ref{thm:mac}). The reason for its inclusion in this paper is to demonstrate  that the proof techniques we have developed here are general and widely applicable to many network information theory problems, including problems whose capacity regions involve auxiliary random variables. One can also derive local dispersions and sum-rate dispersions. %Indeed, it is  easy  to apply our vector rate redundancy theorem to prove an achievable global dispersion for the $(n,\epsilon)$-capacity region of the discrete memoryless interference channel (DM-IC)~\cite{Han81, Chong08}. However, for the DM-IC, there is an additional Fourier-Motzkin step to eliminate the common and private message rates. %This step can indeed be  done for deriving the  finite blocklength inner bound.  

For the ABC, one can easily improve on the global dispersion result presented in Theorem~\ref{thm:bc} by using constant   composition codes, i.e., first generate the $U$ codewords (cloud centers) uniformly at random from some type class $\calT_{P_U}$ then generate the $X$ codewords (satellites) uniformly at random from some $P_{X|U}$-shell $\calT_{P_{X|U}}(u^n)$ centered at a cloud center $u^n \in\calT_{P_U}$. Then instead of the unconditional information dispersion matrix $\bV(p_{U,X}, W)$ in \eqref{eqn:bcv}, we see that the following conditional  information dispersion matrix is also achievable:
\begin{equation}
\bV'(P_{U,X}, W) = \bbE_{X,U}\left[ \cov \big( \bi(U,X, Y_1, Y_2) | X,U \big) \right]. 
\end{equation}
Note that  $\bV'(P_{U,X}, W)\preceq \bV (P_{U,X}, W)$ so the dispersion   is not increased using such constant   composition codes. We do not pursue this extension in detail here but note that instead of of i.i.d.\ version of the multi-dimensional Berry-Ess\'een theorem  (Corollary~\ref{cor:be}) we need a version that deals with independent but not necessarily identically distributed random vectors, e.g., the one provided by G\"{o}etze~\cite{Got91}.

%Loosely speaking, for the DM-ABC and  for a fixed $p_{U,X}$, the inner bound approaches to the capacity region in \eqref{eqn:abc_reg} at a rate of $O(\frac{1}{\sqrt{n}} )$ as prescribed by the multidimensional central limit theorem. The redundancy set is, as per the DM-MAC, $\frac{1}{\sqrt{n}} \scS(\bV,\epsilon)$.  It would be interesting to extend the above result to derive a finite blocklength version of Marton's inner bound~\cite{Marton79, elgamal81}, which is the best (largest) inner bound for the broadcast channel. This problem, however, appears to be rather  challenging because of the need to generalize the mutual covering lemma \cite{elgamal}.

\section{Discussion and Open Problems} \label{sec:concl}
To summarize,   we characterized the $(n,\epsilon)$-optimal rate region for the SW problem up to the $O(\frac{\log n}{n})$ term. We showed that this global dispersion result can be stated in terms of a new object which we call the dispersion matrix. We also provided similar inner bounds for the DM-MAC and DM-ABC problems. We unified our achievability proofs through an important theorem known as the vector rate redundancy theorem. We believe this general result would be useful in other network information theory problems. 

To gain better insight to the dispersion of network problems, we focused on the SW problem and considered the rate of convergence of the non-asymptotic rate region $\scR_{\mathrm{SW}}^*(n,\epsilon)$ to the boundary of the SW region. We defined and exactly characterized two operational dispersions, namely the local  and weighted sum-rate dispersions. One of the most interesting and novel results presented here is the following: When we approach a corner point, the scalar dispersions that have been prevalent in the recent literature~\cite{Hayashi09,PPV10, wang11, Kos11, ingber11} do not suffice. Rather, to characterize the local  and weighted sum-rate dispersions, we need to use the bivariate Gaussian as well as some off-diagonal elements of the dispersion matrix.

Clearly, it would be desirable to derive dispersion-type outer bounds for the $(n,\epsilon)$-capacity region of the DM-MAC and DM-ABC. We  have discussed the difficulties to obtaining such outer bounds. For the DM-MAC, it appears that generalizations of Polyanskiy et al.'s meta (or minimax) strong converse \cite[Theorem 26]{PPV10} or Augustin's strong converse~\cite{augustin} to multi-terminal settings are required.  For the MAC, it was mentioned in Section \ref{sec:mac_interp} that a sharpening of Ahlswede's wringing technique~\cite{Ahl82} seems necessary for a converse proof. For the ABC, appears that strengthening of  the information spectrum technique in~\cite{bouch00} or the entropy and image size characterizations technique~\cite[Ch.~15]{Csi97} are required for a dispersion-type outer bound.

%It is also of interest to extend these finite blocklength results to channel and lossy source coding problem with side information. These include the Gel'fand-Pinsker problem (channel coding with non-causal state information at the encoder) and the Wyner-Ziv problem (lossy source coding with side information at the decoder). Again, the authors believe that versatile strong converses,  such as that in~\cite{tyagi} and~\cite{Kel07}, have to be developed and strengthened for meaningful finite blocklength results.   Preliminary work for channels with random state which is known   at the receiver was presented by Ingber and Feder \cite{ingber10}.  
%
%Finally,  for the relay channel, the most well-known achievability schemes are decode-forward and compress-forward. These coding procedures rely on block-Markov coding \cite{elgamal}. Essentially, one codes over $b$ (correlated) blocks each of length $n$, achieving rate of approximately $\frac{b-1}{b}R$ for some rate $R$. Given a fixed super-blocklength $N$, how can we resolve the tradeoff   between the number of blocks $b$ and the sub-blocklength $n$ to maximize the overall rate subject to an error probability of $\epsilon$?
\section{Proofs of Global Dispersions} \label{sec:proofs}
In this section, we provide the proofs  for the global dispersion results in the previous sections (Theorems~\ref{thm:orr}, \ref{thm:mac} and~\ref{thm:bc}).  We start in Section~\ref{sec:vrrt}  by stating and proving a preliminary but important result known as the vector rate redundancy theorem. This result is a generalization of the (scalar) rate redundancy theorem in \cite{ingber11, wang11}. We then  prove Theorems~\ref{thm:orr}, \ref{thm:mac} and~\ref{thm:bc} in Sections \ref{sec:prf_sw}, \ref{sec:prf_mac}, and \ref{sec:prf_abc} respectively. %The main ingredient  common to each of the proofs is the following important  theorem:
\subsection{A Preliminary Result} \label{sec:vrrt}
\begin{theorem}[Vector Rate Redundancy Theorem] \label{thm:vrrt}
Let $\bg : \scP(\calX)\to\bbR^d$ be   twice continuously differentiable. Let  
\begin{equation}
g'_t(x) := \frac{\partial g_t(q_X)}{\partial q_X(x)} \bigg|_{q_X=p_X },
\end{equation}
for $t=1,\ldots, d$  be the component-wise derivatives of $\bg$. Denote  the vector of derivatives (the gradient vector) as $\bg'(x) = [g'_1(x),\ldots, g_t'(x)]^T$. Let $\bV\in\bbR^{d\times d}$ be  the covariance matrix   of the random vector $\bg'(X)$, i.e.,
\begin{equation}\label{eqn:vdef}
\bV = \cov_X [\bg'(X) ] = \rvE [ (\bg'(X)   - \rvE[ \bg'(X) ] )(\bg'(X)   - \rvE[ \bg'(X) ] )^T].
\end{equation}  
 Assume that $\rank(\bV)\ge 1$ and $\xi:= \rvE [ \|\bg'(X) - \rvE[\bg'(X)] \|_2^3]<\infty$.  Furthermore, let $X^n = (X_1, \ldots, X_n)$ be an i.i.d.\ random vector with $X_k \sim  p_X(x)$. Define 
\begin{equation}
 \kappa :=  \max_{1\le t\le d}\left\| \nabla_{p_X}^2 g_t(p_X)\right\|_2, \label{eqn:def_beta} 
\end{equation} 
where $\nabla_{p_X}^2 g_t(p_X)$ denotes the Hessian matrix of $p_X\mapsto g_t(p_X)$.   Define the sequence  
 \begin{equation} \label{eqn:seq_an}  
 b_n =  \frac{(\kappa+1)\log n}{n}.
 \end{equation} 
  Then, for any vector $\bz\in\bbR^d$, we have
\begin{equation}
\rvP \left(\bg(P_{X^n}) \ge \bg(p_X) + \frac{\bz}{\sqrt{n}}-b_n \bone\right) \ge\rvP(\bZ\ge\bz)  +O\left(\frac{\log n}{\sqrt{n}}\right), \label{eqn:rrt}
\end{equation}
where $P_{X^n} \in\scP_n(\calX)$ is the (random)  type of the sequence $X^n$ and $\bZ\sim\calN(\bzero,\bV)$.
\end{theorem}
%\begin{remark}
%Close inspection of the   proof reveals that the sequence $a_n$ in \eqref{eqn:seq_an}   can be chosen to be ${c}/{n}$ with a sufficiently large constant $c>0$ (depending on the function $\bg$). However, in the sequel we will only need the weaker form of the result with $a_n$ defined as in~\eqref{eqn:seq_an}. \end{remark}

%\begin{remark} The result in \eqref{eqn:rrt} also holds for the case where $\bV$ is singular. See Section~\ref{sec:sing} for the justification.\end{remark}

Before we prove Theorem~\ref{thm:vrrt}, let us state Bentkus' version of the multidimensional Berry-Ess\'{e}en theorem. % For a more general version of this theorem  for random vectors that are not necessarily identically distributed, see~\cite{Ben05}. 
\begin{theorem}[Bentkus~\cite{Ben03}] \label{thm:bentkus}
Let $\bU_1,\ldots, \bU_n$ be normalized i.i.d.\ random vectors  in $\bbR^d$ with zero mean  and identity covariance matrix, i.e., $\rvE [\bU_1] =\bzero$ and  $\cov [ \bU_1]=\bI$. Let $\bS_n := \frac{1}{\sqrt{n}} (\bU_1+ \ldots+ \bU_n)$ and $\xi = \rvE [ \| \bU_1 \|_2^3]$. Let $\bZ\sim\calN(\bzero,\bI)$ be a standard Gaussian random vector  in $\bbR^d$.  Then,  for all $n\in\bbN$,
\begin{equation}
\sup_{\scC \in \frakC_d} | \rvP( \bS_n \in \scC )-\rvP(\bZ \in \scC) |\le\frac{400 d^{1/4}\xi}{\sqrt{n}} \label{eqn:bent}
\end{equation}
where $\frakC_d$ is the family of all convex, Borel measurable subsets of $\bbR^d$.
\end{theorem}
Bentkus   remarks in~\cite{Ben03} that the constant $400$  in Theorem~\ref{thm:bentkus} can be ``considerably improved especially for large $d$''. For simplicity, we will simply use~\eqref{eqn:bent}.   Because we will frequently encounter random vectors with non-identity covariance matrices and ``whitening''  is not applicable, it is necessary to modify Theorem~\ref{thm:bentkus} as follows:
\begin{corollary} \label{cor:be}
Assume the  same setup as in Theorem~\ref{thm:bentkus} with the exception that $\cov[\bU_1] = \bV\succ 0$ and $\bZ\sim\calN(\bzero,\bV)$. Then \eqref{eqn:bent} becomes
\begin{equation}
\sup_{\scC \in  \frakC_d} | \rvP( \bS_n \in \scC )-\rvP(\bZ \in \scC) |\le\frac{400 d^{1/4}\xi}{\lambda_{\min}(\bV)^{3/2}\sqrt{n}} .  \label{eqn:bent2}
\end{equation}
\end{corollary}
The proof of the corollary is by simple linear algebra and is presented in Appendix~\ref{prf:be}. We are now ready to prove the important vector rate redundancy theorem. 
\begin{proof}
First we assume that $\lambda_{\min}(\bV)>0$. In the latter part of the proof, we relax this assumption.  By Taylor's theorem applied component-wise, we can rewrite $\bg(P_{X^n})$   as 
\begin{equation}
\bg(P_{X^n}) = \bg(p_X) + \sum_{x\in\calX} \bg'(x)  [P_{X^n}(x)-p_X(x)]+ \bDelta.\label{eqn:taylor_rrt}
\end{equation}
Recall  that  $\bg$ is  twice  continuously differentiable and the probability simplex $\scP(\calX)$ is compact. As such, we can conclude that  each entry of the second-order residual term in \eqref{eqn:taylor_rrt} can be bounded above as 
\begin{align}
|\Delta_t| \le  \frac{1}{2}\left\| \nabla_{p_X}^2 g_t(p_X) \right\|_2 \|P_{X^n}-p_X\|_2^2. \label{eqn:bound_delta}
\end{align}
Using the definition of $\kappa$ in \eqref{eqn:def_beta} yields,
\begin{equation}
\|\bDelta\|_\infty \le  \frac{\kappa}{2}\|P_{X^n}-p_X\|_2^2 . \label{eqn:bDelta}
\end{equation}
We now evaluate the probability that $\|\bDelta\|_\infty$ exceeds $c_n>0$: 
\begin{align}
\rvP ( \|\bDelta\|_{\infty} \ge c_n  )  & \le  \rvP \left( \frac{\kappa}{2}\| P_{X^n} -p_{X}  \|_2^2 \ge c_n \right) \label{eqn:use_delta_bd}   \\
& \le \rvP \left(  \| P_{X^n} -p_{X}  \|_1^2 \ge \frac{ 2c_n}{\kappa} \right)  \label{eqn:l1l2} \\
%& \le  \rvP (  (2\ln 2) D ( P_{X^n} || p_{X})   \ge t_n/K)  \label{eqn:pink} \\
&\le  2^{|\calX|} 2^{-n  c_n / \kappa  },   \label{eqn:sanov} 
\end{align}
where~\eqref{eqn:use_delta_bd} uses the bound on $\|\bDelta\|_\infty$ in~\eqref{eqn:bDelta},  \eqref{eqn:l1l2} follows because the $\ell_2$-norm dominates the $\ell_1$-norm for finite-dimensional vectors, and finally \eqref{eqn:sanov} follows from a sharpened bound on the $\ell_1$-deviation of the type from the generating distribution by Weissman et al.~\cite{Weiss03}. Setting 
\begin{equation}
\label{eqn:cn}
c_n:=\frac{ { \kappa}\log n}{n}
\end{equation}
    establishes that 
\begin{equation}
\rvP(  \|\bDelta\|_{\infty} \ge c_n ) \le \frac{ 2^{|\calX|}}{ n }. \label{eqn:res_sanov}
\end{equation}
For convenience, let us denote the  left-hand-side (LHS) of \eqref{eqn:rrt} as  $q_n$. Then,  using~\eqref{eqn:taylor_rrt}, 
\begin{align}
q_n &= \rvP\left(\sum_{x\in\calX} \bg'(x) [P_{X^n}(x)-p_X(x)] +\bDelta\ge\frac{\bz}{\sqrt{n}}-b_n\bone   \right). \label{eqn:apply_tay}
\end{align}
Now, we note the following fact which is proved in Appendix~\ref{app:simple_pr}.
\begin{lemma} \label{lem:simple_pr}
Let $\bG$ and $\bDelta$ be random vectors in $\bbR^d$. Let $\bv$ be a vector in $\bbR^d$. Then for any $\phi\ge 0$, 
\begin{equation}
\rvP(\bG+\bDelta\ge\bv)\ge\rvP(\bG\ge \bv+  \phi  \bone ) - \rvP( \|\bDelta\|_{\infty}\ge  \phi )   . \label{eqn:simple_prob1}
\end{equation}
\end{lemma}
Using the identifications $\bG\leftarrow\sum_x \bg'(x) [ P_{X^n}(x) -p_X(x) ]$, $\phi\leftarrow c_n$, $\bDelta\leftarrow \bDelta$ and  $\bv \leftarrow  \frac{\bz}{\sqrt{n}}-b_n\bone$, we can lower bound the right hand side of  \eqref{eqn:apply_tay} as follows,
\begin{align}
q_n&\ge \rvP\left(\sum_{x\in\calX} \bg'(x) [P_{X^n}(x)-p_X(x)]\ge\frac{\bz}{\sqrt{n}}-b_n\bone  +c_n\bone \right) -\rvP(\|\bDelta\|_{\infty} \ge c_n ) \label{eqn:prob_ineq}\\
&\ge \rvP\left(\sum_{x\in\calX} \bg'(x) [P_{X^n}(x)-p_X(x)]\ge\frac{\bz}{\sqrt{n}}-b_n \bone +c_n\bone \right) - \frac{ 2^{|\calX|}}{ n }.\label{eqn:use_app}
\end{align}
In the last inequality, we used the result in~\eqref{eqn:res_sanov} for  the chosen $c_n$. Because the type $ P_{X^n}$ puts a probability mass of $\frac{1}{n}$ on each sample $X_k$,
\begin{equation}
\sum_{x\in\calX} \bg'(x) P_{X^n}(x) = \frac{1}{n}\sum_{k=1}^n \bg'(X_k). \label{eqn:type_place}
\end{equation}
By definition of the expectation, we also have 
\begin{equation}
\sum_{x\in\calX} \bg'(x) p_X(x)=\rvE [\bg'(X)] .\label{eqn:expect_def}
\end{equation}
The substitution of~\eqref{eqn:type_place} and \eqref{eqn:expect_def} in~\eqref{eqn:use_app} yields
\begin{align}
q_n &\ge \rvP\left( \frac{1}{n}\sum_{k=1}^n  ( \bg'(X_k)-  \rvE [ \bg'(X) ] )\ge \frac{\bz}{\sqrt{n}}-b_n\bone+c_n\bone\right) - \frac{ 2^{|\calX|}}{ n } \\*
& = \rvP\left( \frac{1}{\sqrt{n}}\sum_{k=1}^n  ( \bg'(X_k)-  \rvE [ \bg'(X) ] )\ge  {\bz}-\sqrt{n}(b_n-c_n)\bone\right) - \frac{ 2^{|\calX|}}{ n } . \label{eqn:iid_rv}
\end{align}
Now note that the random vectors $\{ \bg'(X_k)-  \rvE [ \bg'(X) ]\}_{k=1}^n$ are i.i.d.\ and have zero-mean and covariance $\bV$ defined in~\eqref{eqn:vdef}. In addition, the set $\{\bg \in \bbR^d:\bg\ge\bz'\}$ is convex so it belongs to $\frakC_d$. Using the multidimensional Berry-Ess\'{e}en theorem in~\eqref{eqn:bent2} to further lower bound~\eqref{eqn:iid_rv} yields
\begin{align}
q_n\ge \rvP \left( \bZ\ge \bz-\sqrt{n}(b_n-c_n) \bone\right)  - \frac{400d^{1/4} \xi}{\lambda_{\min}(\bV) \sqrt{n}}- \frac{ 2^{|\calX|}}{ n }  \label{eqn:apply_berry}
\end{align}
where  the third moment $\xi = \rvE [\|  \bg'(X )-  \rvE [ \bg'(X) ]  \|_2^3]<\infty$ by assumption. In addition,   we assumed that $\lambda_{\min}(\bV)>0$  so the second term  is finite.   Now, note that the sequence $\sqrt{n} (b_n-c_n) = \frac{\log n}{\sqrt{n}}$  from the definition of $b_n$ in \eqref{eqn:seq_an}  and $c_n$ in \eqref{eqn:cn}. Since $\delta\mapsto\rvP(\bZ\ge\bz-\delta\bone)$ is continuously differentiable and monotonically increasing, we have \begin{equation}\label{eqn:taylorapp}\rvP(\bZ\ge\bz-\delta\bone)=\rvP(\bZ\ge\bz)+O(\delta)\end{equation} by Taylor's approximation theorem. Applying \eqref{eqn:taylorapp} to \eqref{eqn:apply_berry} with $\delta=\sqrt{n}(b_n-c_n)=\frac{\log n}{\sqrt{n}}$ yields the lower bound
\begin{align}
q_n  &\ge \rvP \left( \bZ\ge \bz \right)  + O\left(\frac{\log n}{\sqrt{n}} \right)- \frac{400d^{1/4} \xi}{\lambda_{\min}(\bV) \sqrt{n}}- \frac{ 2^{|\calX|}}{ n } , \label{eqn:end_vrrt}% \\
%&= 1-\epsilon  + O\left(\frac{\log n}{\sqrt{n}} \right),\label{eqn:end_vrrt2}
\end{align}
whence the desired result  follows for the case $\bV\succ 0$. 

Now we consider the case where $\bV$ is singular but  recall that we assume $\rank(\bV)\ge 1$.  The only step in which we have to modify  in the proof for the case $\bV\succ 0$ is in the application of the multidimensional  Berry-Ess\'{e}en theorem in~\eqref{eqn:apply_berry}. This is because we would be dividing by $\lambda_{\min}(\bV)=0$. To fix this, we reduce the problem to the non-singular case. Assume that $\rank(\bV)=r<d$ and define the zero-mean i.i.d.\ random vectors $\bA_k:=\bg'(X_k) - \rvE [\bg'(X)]$. There exists a $d\times r$ matrix $\bT$ such that $\bA_k = \bT\bB_k$ where $\bB_k\in\bbR^r$ are i.i.d.\ random vectors  with positive-definite covariance matrix $\tilde{\bV}$. The matrix $\bT$ can be taken to be composed of the $r$ eigenvectors corresponding to the non-zero eigenvalues of $\bV$.   We can now replace the $\bA_k$ vectors in~\eqref{eqn:iid_rv} with $\bT\bB_k$ and apply the  multidimensional Berry-Ess\'{e}en theorem~\cite{Ben03} to $\bB_1,\ldots, \bB_n$. The theorem clearly applies since the set $\{\mathbf{b} \in \bbR^r:\bT \mathbf{b}\ge\bz'\}$ is   convex. This gives the same conclusion as in~\eqref{eqn:end_vrrt} with $r$ in place of $d$ and $\lambda_{\min}(\tilde{\bV })$ in place of $\lambda_{\min}(\bV) $.
\end{proof}

\subsection{Proof  of the Global Dispersion for the SW Problem (Theorem~\ref{thm:orr})} \label{sec:prf_sw}

We now present the  proof of Theorem~\ref{thm:orr} on the $(n,\epsilon)$-optimal rate region for distributed lossless source coding. We present the achievability proof in Section~\ref{sec:ach} and the converse proof in Section~\ref{sec:conv}.   We will see that the achievability procedure (coding scheme) is universal. In Section~\ref{sec:universal}, we discuss the implications of choosing not to use a universal decoding rule but a rule that is akin to maximum-a posteriori decoding~\cite{Gal76}. %Finally in Section~\ref{sec:sing}, we discuss implications when the entropy dispersion matrix $\bV$ is singular. 

\subsubsection{Achievability}\label{sec:ach} 
\textcolor{white}{..}  
\begin{proof}
Let $(R_1, R_2)$ be a rate pair  in the inner bound $\scR_{\mathrm{in}}(n,\epsilon)$ defined in \eqref{eqn:ach_sw}. \\% where the $O(\frac{\log n}{n})$ term is chosen specifically as $\frac{\alpha\log n}{n},\alpha>0$.\\
{\em Codebook Generation}: For $j=1,2$, randomly and independently assign an index $f_{1,n}(x_j^n) \in [2^{nR_j}]$ to each sequence $x_j^n\in\calX_j^n$ according to a uniform probability mass function. The sequences of the same index form a {\em bin}, i.e., $\calB_j(m_j) := \{x_j^n\in\calX_j^n : f_{1,n}(x_j^n)  = m_j\}$. Note that $\calB_j(m_j),  m_j  \in [2^{nR_j}]$ are random subsets of $\calX_j^n$. The bin assignments are revealed to all parties. In particular, the decoder knows the bin rates $R_j$. \\
{\em Encoding}: Given $x_j^n\in\calX_j^n$, encoder $j$  transmits the bin index  $f_{j,n}(x_j^n)$. Hence, for length-$n$ sequence, the rates of   $m_1$ and $m_2$ are $R_1$ and $R_2$ respectively.  \\
{\em Decoding}: The decoder, upon receipt of the bin indices $(m_1, m_2)$ finds the unique sequence pair $(\hatx_1^n, \hatx_2^n) \in \calB_1(m_1)\times \calB_2(m_2)$ such that the empirical entropy vector 
\begin{equation}
\hat{\bH}(\hatx_1^n, \hatx_2^n) :=  \begin{bmatrix}\hatH(\hatx_1^n|\hatx_2^n) \\ \hatH(\hatx_2^n|\hatx_1^n)  \\ \hatH(\hatx_1^n,\hatx_2^n)  \end{bmatrix}  \le\bR-\delta_n\bone , \label{eqn:decoding_rule_sw}
\end{equation}
where the thresholding sequence $\delta_n$ is defined as 
\begin{equation}
\delta_n:= \left(|\calX_1||\calX_2|+\frac{1}{2} \right)\frac{\log (n+1)}{n}.\label{eqn:delta_n}
\end{equation}   Define   $\scT(\bR, \delta_n):= \{ \bz\in\bbR^3:\bz\le \bR -\delta_n\bone \}$ to be the {\em typical empirical entropy set}. Then, \eqref{eqn:decoding_rule_sw} is equivalent to $\hat{\bH}(\hatx_1^n, \hatx_2^n)\in\scT(\bR,\delta_n)$.  If there is more than one pair or no such pair in $\calB_1(m_1)\times \calB_2(m_2)$, declare a decoding error. Note that our decoding scheme is {\em universal} \cite{Csi97}, i.e., the decoder does not depend on knowledge of the true distribution $p_{X_1, X_2}$.  It does depend on the rate pair which is known to the decoder since the codebook (bin assignments) is known to all parties. 
\\ {\em Analysis of error probability}: Let the sequences sent by the two users be $(X_1^n, X_2^n)$ and let their corresponding  bin indices be $(M_1, M_2)$. We bound the probability of error averaged over the random code construction. Clearly, the ensemble probability of error   is bounded above by the sum of the probabilities of the following four events:
\begin{align}
\calE_1 &:= \{ \hatbH (X_1^n, X_2^n) \notin   \scT(\bR,\delta_n)  \}   \\
\calE_2 &:= \{\exists \, \tilx_1^n\in \calB_1(M_1) \setminus\{X_1^n\} :    \hatbH (\tilx_1^n, X_2^n) \in   \scT(\bR,\delta_n)  \}  \\
\calE_3 &:= \{\exists \,  \tilx_2^n\in \calB_2(M_2) \setminus\{X_2^n\} :    \hatbH (X_1^n, \tilx_2^n)   \in   \scT(\bR,\delta_n)  \} \\
\calE_4 &:= \{\exists  \, \tilx_1^n\in \calB_1(M_1) \setminus\{X_1^n\}, \tilx_2^n\in \calB_2(M_2) \setminus\{X_2^n\} :  \nn\\*
&\qquad  \qquad\qquad  \hatbH (\tilx_1^n, \tilx_2^n)  \in   \scT(\bR,\delta_n)   \} 
\end{align}
We bound the probabilities of these events in turn. Consider
\begin{align}
\rvP(\calE_1)  &= 1-\rvP(  \hatbH (P_{X_1^n, X_2^n}) \in \scT(\bR,\delta_n) )  \label{eqn:depend_type} \\
&= 1-\rvP(  \hatbH (P_{X_1^n, X_2^n}) \le \bR - \delta_n\bone)  \label{eqn:depend_type2}\\
&= 1-\rvP \left(  \hatbH (P_{X_1^n, X_2^n}) \le \bH(p_{X_1, X_2})+\frac{\tilbz}{\sqrt{n}} + (a_n - \delta_n )\bone\right)  \label{eqn:depend_type3}
\end{align}
where  we made the dependence of the empirical entropy vector on the type explicit in \eqref{eqn:depend_type}.   In \eqref{eqn:depend_type2}, we invoked the definition of $\scT(\bR,\delta_n)$. In \eqref{eqn:depend_type3}, we used the fact that $\bR = \bH(p_{X_1, X_2}) + \frac{\tilbz}{\sqrt{n}} +a_n$ for some  vector $\tilbz\in\bbR^3$ that satisfies $\rvP(\bZ\le\tilbz)\ge 1-\epsilon$ where $\bZ\sim\calN(\bzero,\bV)$ and $a_n = \frac{\nu\log n}{n}$ for $\nu = |\calX_1| |\calX_2| +\frac{3}{2} +   \kappa$, where $\kappa$ was defined in \eqref{eqn:kappa_sw}. Note that $\kappa<\infty$ because we assumed that $p_{X_1 , X_2}(x_1, x_2)>0$ for all $(x_1, x_2)$. 

We now   bound the probability in~\eqref{eqn:depend_type3} using the vector rate redundancy theorem with the following identifications:  random variable  $X\leftarrow (X_1, X_2)$, smooth function $\bg(p_{X_1, X_2}) \leftarrow-\bH(p_{X_1, X_2})$, evaluation vector $\bz \leftarrow -\tilbz$ and sequence $b_n\leftarrow a_n-\delta_n$. The function $p_{X_1, X_2}\mapsto-\bH(p_{X_1, X_2})$ is twice continuously differentiable because $p_{X_1 , X_2}(x_1, x_2)>0$ for all $(x_1, x_2)$.  Note that setting the coefficient of $a_n$, namely $\nu$,  to be $|\calX_1||\calX_2|+1/2+( \kappa+1)$ results in $b_n = ( \kappa+1)\frac{\log n}{n}$  as required by Theorem~\ref{thm:vrrt}.  This has been  ensured with the choice of $\nu$ in Definition~\ref{def:inner_sw}. Also, the third moment is uniformly bound as stated    in Appendix~\ref{app:finite_third} .

With the above identifications   and the realization that the matrix  $\bV$ in the vector rate redundancy theorem   equals $\cov( \bh(X_1, X_2))$ (by direct differentiation of entropy functionals),  
\begin{align}
\rvP( \calE_1^c)&\ge \rvP(\bZ\ge-\tilbz)+O\left(\frac{\log n}{\sqrt{n}} \right)  \\*
&=\rvP(\bZ\le\tilbz)+O\left(\frac{\log n}{\sqrt{n}} \right)  \label{eqn:use_zm} \\*
&\ge 1-\epsilon+O\left(\frac{\log n}{\sqrt{n}} \right) ,  \label{eqn:pe1c}
\end{align}
where in \eqref{eqn:use_zm} we used the fact that $ \rvP(\bZ\ge-\tilbz)=\rvP(\bZ\le\tilbz)$ because $\bZ$ has zero mean.  Consequently,
\begin{equation}
\rvP( \calE_1 ) \le \epsilon-O\left(\frac{\log n}{\sqrt{n}} \right). \label{eqn:pe1}
\end{equation}
For the second event,  by symmetry and uniformity, $\rvP(\calE_2)=\rvP(\calE_2|X_1^n\in \calB_1(1))$.  For ease of notation, let $p:= p_{X_1^n, X_2^n}$. Now consider the chain of inequalities:
\begin{align}
&\!\! \rvP(\calE_2|X_1^n\in \calB_1(1))  \nn \\*
&=  \sum_{x_1^n, x_2^n} p(x_1^n, x_2^n)   \rvP\Big[\exists \, \tilx_1^n\in \calB_1(1)\setminus\{X_1^n\}:  \nn\\*
& \qquad\hatbH ( \tilx_1^n, x_2^n ) \in   \scT(\bR,\delta_n)    \Big| (X_1^n, X_2^n)=(x_1^n,x_2^n), X_1^n\in\calB_1(1)\Big]  \label{eqn:events_indep}    \\
%&\leb \sum_{x_1^n, x_2^n} p(x_1^n, x_2^n)   \sum_{\tilx_1^n \ne x_1^n: \hatbH ( \tilx_1^n, x_2^n )  \in\calT(R_1 , R_2,\delta_n)}\rvP \left(\tilx_1^n \in \calB_1(1)  \right)    \\
&\le \sum_{x_1^n, x_2^n} p(x_1^n, x_2^n)   \sum_{\tilx_1^n \ne x_1^n: \hatbH( \tilx_1^n,x_2^n) \in\scT(\bR,\delta_n)}\rvP \left(\tilx_1^n \in \calB_1(1)  \right)    \label{eqn:ub_sw}     \\
&\le \sum_{x_1^n, x_2^n} p(x_1^n, x_2^n)   \sum_{\tilx_1^n \ne x_1^n: \hatH( \tilx_1^n|x_2^n) \le R_1-\delta_n  }\rvP \left(\tilx_1^n \in \calB_1(1)  \right) \label{eqn:inclu}       \\
&= \sum_{x_1^n, x_2^n} p(x_1^n, x_2^n)  \sum_{\tilx_1^n \ne x_1^n: \hatH( \tilx_1^n|x_2^n)  \le R_1-\delta_n  } \frac{1}{\ceil{2^{nR_1}}} \label{eqn:unif}     \\
%&\le\sum_{x_2^n} p(x_2^n)   \sum_{\tilx_1^n  : \hatH( \tilx_1^n|x_2^n)  \le R_1-\delta_n  } 2^{-nR_1} \label{eqn:unif2} \\    
&\le  \sum_{Q \in \scP_n( \calX_2)}  \sum_{  x_2^n  \in \calT_{Q} } p(  x_2^n)  \sum_{ \substack{V \in \scV_n (\calX_1;Q) : \\ H(V| P_{x_2^n})\le R_1-\delta_n}}\sum_{\tilx_1^n\in    \calT_V(x_2^n)} 2^{-nR_1} \label{eqn:parti}   \\
&\le  \sum_{Q \in \scP_n( \calX_2)}  \sum_{  x_2^n  \in \calT_{Q} } p(  x_2^n)   \sum_{ \substack{V \in \scV_n(\calX_1;Q) : \\ H(V| P_{x_2^n})\le R_1-\delta_n}}    2^{nH(V |P_{x_2^n} )} 2^{-nR_1} \label{eqn:card_vshell}   \\
&\le   \sum_{    x_2^n   } p(  x_2^n)    (n+1)^{|\calX_1| |\calX_2|} 2^{n (R_1-\delta_n)} 2^{-nR_1} \label{eqn:pe2}\\
&=   \  (n+1)^{|\calX_1| |\calX_2|} 2^{n (R_1-\delta_n)} 2^{-nR_1} \label{eqn:pe22}
\end{align}
where   \eqref{eqn:events_indep} follows from the definition of $\calE_2$, \eqref{eqn:ub_sw} follows  from the union bound and because  for $\tilx_1^n\ne x_1^n$, the events $\{x_1^n\in\calB_1(1)\}$, $\{\tilx_1^n\in\calB_1(1)\}$ and $\{(X_1^n, X_2^n)=(x_1^n, x_2^n)\}$  are mutually independent,  and \eqref{eqn:inclu}  follows from  the inclusion $\{\tilx_1^n: \hatbH ( \tilx_1^n, x_2^n )  \in   \scT(\bR,\delta_n)\}\subset \{\tilx_1^n:  \hatH( \tilx_1^n|x_2^n)  \le R_1-\delta_n \}$.  Equality~\eqref{eqn:unif}  follows from the uniformity in the random binning. In \eqref{eqn:parti}, we  first dropped the constraint $\tilx_1^n\ne x_1^n$ and marginalized over $x_1^n$. Then, we partitioned the sum over $x_2^n$ into disjoint  type classes indexed by      $Q  \in \scP_n( \calX_2)$ and we partitioned the sum over   $\tilx_1^n \in\calX_1^n$ into sums over stochastic matrices  $V\in\scV_n  (\calX_1;Q)$ (for notation see Section~\ref{sec:notation}). In \eqref{eqn:card_vshell},  we  upper bounded  the cardinality of the $V$-shell as $|\calT_V( x_2^n)|\le 2^{n H(V|P_{ x_2^n})}$~\cite[Lem.~1.2.5]{Csi97}. In \eqref{eqn:pe2}, we used the  Type Counting Lemma~\cite[Eq.~(2.5.1)]{Csi97}. By the choice of  $\delta_n$ in~\eqref{eqn:delta_n},  inequality~\eqref{eqn:pe22} reduces to
\begin{equation}
\rvP(\calE_2)\le \frac{1}{\sqrt{n+1}}.
\end{equation} Similarly $\rvP(\calE_3)\le \frac{1}{\sqrt{n+1}}$ and $\rvP(\calE_4)\le \frac{1}{\sqrt{n+1}}$. 

Together with~\eqref{eqn:pe1}, we conclude that the  error probability defined in \eqref{eqn:perr_sw} averaged over the random binning  is upper bounded as 
\begin{equation}
\rvP(\calE)\le \sum_{i=1}^4\rvP(\calE_i)\le \epsilon ,
\end{equation}
for all $n$ sufficiently large. Hence, there is a deterministic code whose error probability in~\eqref{eqn:perr_sw} is no greater than $\epsilon$ if the rate pair  $(R_1,  R_2)$ belongs to $\scR_{\mathrm{in}}(n,\epsilon)$.  \end{proof}
 \subsubsection{Converse} \label{sec:conv}
 \textcolor{white}{..} 
 \begin{proof} 
To prove the outer bound,  we use  Lemma 7.2.2.\ in Han~\cite{Han10}  (which was originally proved by Miyake and Kanaya~\cite{miyake}) which asserts that every $(n, 2^{nR_1}, 2^{nR_2}, \epsilon)$-SW code must satisfy 
\begin{align}
\epsilon&\ge \rvP\Bigg[   \frac{1}{n}\log  \frac{1}{p_{X_1^n|X_2^n}(X_1^n|X_2^n)}\ge R_1 +\gamma  \nn\\*
& \quad\mbox{ or } \frac{1}{n}\log\frac{1}{ p_{X_2^n|X_1^n}(X_2^n|X_1^n)}\ge R_2 +\gamma  \nn\\*
&\quad\mbox{ or } \frac{1}{n}\log \frac{1}{p_{X_1^n,X_2^n}(X_1^n,X_2^n)}\ge R_1+R_2 +\gamma  \Bigg] - 3( 2^{-n\gamma} ) \label{eqn:hansw} \\
& =  1-\rvP\left[  \frac{1}{n} \bh(X_1^n, X_2^n) \le \bR + \gamma \bone  \right] - 3( 2^{-n\gamma}  ), \label{eqn:han_sw}
\end{align}
for any $\gamma>0$.  This result is typically used  for proving strong converses for general (non-stationary, non-ergodic) sources but as we will see it is also very useful for proving a dispersion-type converse. Recall that $\bh(X_1^n, X_2^n)$ is the entropy density vector in~\eqref{eqn:ent_dens} evaluated at $(X_1^n,X_2^n)$. By the memorylessness of the source, it can be written as a sum of i.i.d.\ random vectors $\{\bh(X_{1k}, X_{2k})\}_{k=1}^{n}$. 

We assume that $\bV\succ 0$. The case where $\bV$ is singular can be handled in exactly the same way as we did in the proof of the vector rate redundancy theorem. See discussion after \eqref{eqn:end_vrrt}. Fix  $\gamma:=\frac{\log n}{2n}$ and define  $\tilbz:=\sqrt{n} ( \bR-\bH + \frac{\log n}{n}\bone )$.   Now consider  the probability in \eqref{eqn:han_sw}, denoted as $s_n$:
\begin{align}
s_n &=\rvP\left[ \frac{1}{{n}} \sum_{k=1}^n    \bh(X_{1k},  X_{2k}) \le \bH+\frac{\tilbz}{\sqrt{n}}-\frac{\log n}{n}\bone+ \gamma\bone\right]\\*
& =    \rvP\left[  \frac{1}{\sqrt{n}} \sum_{k=1}^n   ( \bh(X_{1k},  X_{2k})  -  \bH   )  \le  \tilbz- \frac{\log n}{2\sqrt{n}}     \bone  \right]   \label{eqn:def_tilbz} 
\end{align}
We are now ready to use the  multidimensional Berry-Ess\'{e}en theorem. We can easily verify that  the third moment $\xi_{\mathrm{SW}} = \rvE [ \|\bh(X_1, X_2) - \bH(p_{X_1, X_2}) \|_2^3]$ is uniformly bounded. See  Appendix~\ref{app:finite_third}. As such, using~\eqref{eqn:bent2} we can upper bound $s_n$ as follows: 
\begin{align}
s_n & \le    \rvP \left[ \bZ\le \tilbz-\frac{\log n}{2\sqrt{n}}  \bone \right] +  \frac{400 (3^{1/4})\xi_{\mathrm{SW}}}{\lambda_{\min}(\bV)^{3/2}\sqrt{n}}    \label{eqn:be_conv}   \\  
%& =   \rvP \left[ \bZ\le \tilbz-\frac{\log n}{2\sqrt{n}}    \bone \right] + O\left(\frac{1}{\sqrt{n}}  \right) \label{eqn:be_conv2}   \\  
& =   \rvP \left( \bZ\le \tilbz \right) - O\left(\frac{\log n}{\sqrt{n}}  \right) .  \label{eqn:taylor_sw}
\end{align}
The last step follows  by Taylor's approximation theorem. See~\eqref{eqn:taylorapp}. On account of~\eqref{eqn:han_sw} and~\eqref{eqn:taylor_sw}, 
\begin{equation}
\epsilon \ge 1-\rvP(\bZ\le\tilbz) +O\left(\frac{\log n}{\sqrt{n}} \right)-  \frac{3}{\sqrt{n} } % 1-`\left[ 1-\epsilon  - O\left(\frac{\log n}{\sqrt{n}} \right)  \right] - \frac{3}{\sqrt{n} } > \epsilon  %+ O\left(\frac{\log n}{\sqrt{n}} \right)   
\end{equation}
which, upon rearrangement,  means that $\tilbz\in \scS(\bV,\epsilon-O(\frac{\log n}{\sqrt{n}}))$. Since $\scS(\bV,\epsilon')\subset \scS(\bV,\epsilon  )$ if $\epsilon'\le\epsilon$, the vector $\tilbz\in\scS(\bV,\epsilon)$. This implies that $ (R_1, R_2)\in\scR_{\mathrm{out}} (n,\epsilon)$ from the definition of $\tilbz$. %  from the definition of $\scS(\bV,\epsilon)$. 
\end{proof}

% Now,  suppose    to the contrary, that there exists a rate pair $(R_1, R_2)$ such that $\bR\notin\scR_{\mathrm{out}}(n,\epsilon)$ but  $(R_1, R_2)$  is $(n,\epsilon)$-achievable. Then,  by~\eqref{eqn:con_sw},   $\tilbz:=\sqrt{n} ( \bR-\bH + \frac{\log n}{n}\bone ) \in\bbR^3\setminus \scS(\bV,\epsilon)$. By the definition of $\scS(\bV,\epsilon)$ in~\eqref{eqn:SVset}, $\tilbz$ is such that $\rvP( \bZ \le \tilbz) <1-\epsilon$. Fix  $\gamma:=\frac{\log n}{2n}$. 

\subsubsection{Comments on the proof and Universal Decoding} \label{sec:universal}
In place of the universal decoding rule in~\eqref{eqn:decoding_rule_sw}, one could use a non-universal one  by comparing the  normalized entropy density vector    (instead of the empirical entropy vector) evaluated at $(\hatx_1^n,\hatx_2^n)$ with the rate vector, i.e., 
\begin{equation}
-\frac{1}{n} \begin{bmatrix}  \log p_{X_1^n|X_2^n}(\hatx_1^n|\hatx_2^n) \\ \log p_{X_2^n|X_1^n}(\hatx_2^n|\hatx_1^n)  \\ \log p_{X_1^n,X_2^n}(\hatx_1^n,\hatx_2^n) \end{bmatrix}  \le\bR-\delta_n\bone .\label{eqn:decoding_rule_sw2}
\end{equation}
In this case, Taylor expansion as in the proof of the vector rate redundancy theorem [cf.~\eqref{eqn:taylor_rrt}]  would not be required because the above criterion can be written  a  normalized sum of i.i.d.\ random vectors. The multidimensional Berry-Ess\'{e}en theorem can  thus be applied  directly. Under the decoding strategy in~\eqref{eqn:decoding_rule_sw2}, close examination of the proofs shows that there is symmetry between the error probability  bounds in the direct  and converse parts as in \cite[Lemmas 7.2.1-2]{Han10}. In~\cite{Sar05b}, the authors also suggested a universal strategy for finite blocklength SW coding. They suggested the use of feedback to estimate the source statistics, whereas we use the empirical entropy here, cf.~\eqref{eqn:decoding_rule_sw}.

\subsection{Proof  of the Global Dispersion   for the DM-MAC  (Theorem~\ref{thm:mac})} \label{sec:prf_mac}

We now present the  proof of Theorem~\ref{thm:mac} on the $(n,\epsilon)$-capacity region for the DM-MAC.  We present the proof of the inner bound in Section~\ref{sec:mac_ach} and the proof that the cardinality of $Q$ can be restricted to $9$ in Section~\ref{sec:card}.  In Section~\ref{sec:universal_mac}, we comment on how the proof and the statement of the result can be modified if the input and output alphabets of the MAC are not discrete but are arbitrary.

\subsubsection{Achievability}   \label{sec:mac_ach} 
\textcolor{white}{..} 
\begin{proof}
Fix a finite  alphabet $\calQ$ and a tuple of input distributions $(p_Q, p_{X_1|Q}, p_{X_2|Q})$. Fix a  pair of $(n,\epsilon)$-achievable  rates  $(R_1, R_2) \in \scR(n,\epsilon; p_Q,p_{X_1|Q}, p_{X_2|Q}  )$. See definitions in Section~\ref{sec:defmac}.\\ 
{\em Codebook Generation}: Randomly generate a sequence $q^n\sim\prod_{k=1}^n p_Q(q_k)$. For $j=1,2$, randomly and conditionally independently generate  codewords $x_j^n(m_j) \sim \prod_{k=1}^n p_{X_j|Q} (x_{jk}|q_k)$ where $m_j \in [2^{nR_j}]$.  The codebook consisting of $q^n$, $x_1^n(m_1), m_1 \in [2^{nR_1}]$, and  $x_2^n(m_2), m_2 \in [2^{nR_2}]$  is revealed to all parties.  \\
{\em Encoding}: For $j = 1,2, $ given $m_j \in [2^{nR_j}]$, encoder $j$ sends   codeword $x_j^n(m_j) \in \calX_j^n$.  \\
{\em Decoding}: The decoder, upon receipt of the output of the DM-MAC $y^n \in\calY^n$ finds the unique message pair $(\hatm_1, \hatm_2) \in[2^{nR_1}]\times [2^{nR_2}]$ such that the empirical mutual information vector 
\begin{equation}
\hat{\bI}(q^n,x_1^n(\hatm_1),  x_2^n(\hatm_2), y^n) :=  \begin{bmatrix}\hatI(x_1^n(\hatm_1)\wedge y^n |x_2^n(\hatm_2), q^n) \\ \hatI(x_2^n(\hatm_2)\wedge y^n |x_1^n(\hatm_1), q^n)  \\\hatI(x_1^n(\hatm_1), x_2^n(\hatm_2) \wedge y^n | q^n)  \end{bmatrix}  \ge\bR+\delta_n\bone , \label{eqn:decoding_rule_mac}
\end{equation}
where $\delta_n:=  (|\calQ| |\calX_1||\calX_2||\calY|+ \frac{1}{2} )\frac{\log (n+1)}{n}.$ If there is no such message pair  or there is not a unique message pair, declare a decoding error. We remind the reader that $\hatI(x_1^n(\hatm_1)\wedge y^n |x_2^n(\hatm_2), q^n)$ is the conditional mutual information $I(\tilX_1; \tilY|\tilX_2, \tilQ)$ where the dummy random variable $(\tilQ, \tilX_1,\tilX_2,\tilY)$ has distribution, an $n$-type, $P_{q^n,x_1^n(\hatm_1),  x_2^n(\hatm_2), y^n}$.  Let $\scT(\bR, \delta_n) := \{\bz\in\bbR^d: \bz \ge \bR+\delta_n\bone\}$ be the {\em typical empirical mutual information} set. Then the criterion in \eqref{eqn:decoding_rule_mac} is can be written compactly as $\hat{\bI}(q^n,x_1^n(\hatm_1), \hatx_2^n(\hatm_2), y^n) \in \scT(\bR,\delta_n)$. Note that, unlike typicality set decoding~\cite{elgamal} or maximum-likelihood decoding~\cite{gallagerIT}, the decoding rule in  \eqref{eqn:decoding_rule_mac} is {\em universal}, i.e., the decoder does not need to be given knowledge of the channel statistics $W$.  \\
 {\em Analysis of error probability}: By the uniformity of the messages $M_1$ and  $M_2$ and the random code construction, we can assume that $(M_1, M_2)=(1,1)$. The average ensemble error probability is upper bounded by the sum of the probabilities of the following four events:
 \begin{align}
 \calE_1 &:= \{ \hat{\bI}(Q^n,X_1^n(1),  X_2^n(1), Y^n)  \notin\scT(\bR,\delta_n) \} \\
  \calE_2&:= \{ \exists \, \tilm_1\ne 1: \hat{\bI}(Q^n,X_1^n(\tilm_1),  X_2^n(1), Y^n)  \in\scT(\bR,\delta_n) \} \\
    \calE_3&:= \{ \exists\, \tilm_2\ne 1: \hat{\bI}(Q^n,X_1^n(1),  X_2^n(\tilm_2), Y^n)  \in\scT(\bR,\delta_n) \} \\
        \calE_4&:= \{ \exists\, \tilm_1\ne 1, \tilm_2\ne 1: \hat{\bI}(Q^n,X_1^n(\tilm_1),  X_2^n(\tilm_2), Y^n)  \in\scT(\bR,\delta_n) \} 
 \end{align}
%We now bound the probabilities of these events. 
We   use the definition of    $\scR(n,\epsilon; p_Q,p_{X_1|Q}, p_{X_2|Q}  )$ in~\eqref{eqn:ach_mac} to express  $\rvP(\calE_1)$ as follows:
\begin{align}
\rvP(\calE_1)  &= 1- \rvP \left( \hat{\bI}(Q^n,X_1^n(1),  X_2^n(1), Y^n)  \in \scT(\bR,\delta_n) \right) \\
&= 1- \rvP \left( \hat{\bI}(Q^n,X_1^n(1),  X_2^n(1), Y^n) \ge \bR+\delta_n \bone\right) \label{eqn:useTmac} \\
& = 1- \rvP \left( \hat{\bI}(Q^n,X_1^n(1),  X_2^n(1), Y^n) \ge  \bI(p_Q, p_{X_1|Q}, p_{X_2|Q}, W) +\frac{\bz}{\sqrt{n}} -a_n \bone +\delta_n\bone \right), \label{eqn:def_rates_mac}
\end{align}
where  \eqref{eqn:useTmac} follows from the definition  of $\scT(\bR,\delta_n)$. In \eqref{eqn:def_rates_mac}, we used the definition of $\scS(\bV,\epsilon)$ to assert that  $\bz\in\bbR^3$    is a vector satisfying $\rvP(\bZ\ge\bz)\ge 1-\epsilon$  for $\bZ\sim\calN(\bzero,\bV)$. Also, the sequence $a_n =  \frac{\nu\log n}{n}$ where $\nu$ is given in \eqref{eqn:kappa_mac}.

Now we use the vector rate redundancy theorem   with the following identifications:  random variable $X\leftarrow (Q, X_1, X_2, Y)$, smooth function $\bg(p_Q p_{X_1|Q}p_{X_2|Q}W) \leftarrow  \bI(p_Q, p_{X_1|Q}, p_{X_2|Q}, W)$, evaluation vector $\bz \leftarrow\bz$ and sequence $b_n \leftarrow a_n-\delta_n$. 
%If the coefficient of $a_n$ is larger than that of $\delta_n$, say $\nu=|\calQ||\calX_1| |\calX_2| |\calY|+ 1$ as in \eqref{eqn:ach_mac}, $a_n-\delta_n$ is a positive sequence of order $\Theta(\frac{\log n}{n})$, satisfying~\eqref{eqn:seq_an}. Also, the third moment  $\xi_{\mathrm{MAC}} := \rvE[ \|\bi(Q, X_1, X_2, Y) - \bI (  p_Q, p_{X_1|Q}, p_{X_2|Q}, W) \|_2^3]$ is uniformly bounded by  \eqref{sec:MAC_bound}  in Appendix~\ref{app:finite_third}.  
As such, the probability in~\eqref{eqn:def_rates_mac} satisfies
\begin{align} 
&\rvP \left( \hat{\bI}(Q^n,X_1^n(1),  X_2^n(1), Y^n) \ge  \bI(p_Q, p_{X_1|Q}, p_{X_2|Q}, W) +\frac{\bz}{\sqrt{n}} -a_n \bone +\delta_n\bone \right) \nn\\*
&\qquad \ge\rvP(\bZ\ge\bz)+ O\left(\frac{\log n}{\sqrt{n}}\right)  \label{eqn:apply_vrrt0} \\*
&\qquad\ge  1-\epsilon +O\left(\frac{\log n}{\sqrt{n}}\right), \label{eqn:apply_vrrt}
\end{align}
where  in  the first inequality, we   used the fact that the $\bV$ in the vector rate redundancy theorem coincides with the information  dispersion matrix    $\bV(  p_Q, p_{X_1|Q}, p_{X_2|Q}, W )$. This can easily be verified by direct differentiation of (conditional) mutual information quantities with respect to the joint distribution $p_{Q, X_1, X_2, Y}:=p_Qp_{X_1|Q} p_{X_2|Q}W$.     Combining~\eqref{eqn:def_rates_mac} and~\eqref{eqn:apply_vrrt} yields
\begin{equation}
\rvP(\calE_1) \le \epsilon-O\left(\frac{\log n}{\sqrt{n}}\right). \label{eqn:e1_mac}
\end{equation}
To bound the probabilities of $\calE_2,\calE_3$ and $\calE_4$, we use the following lemma whose proof is relegated to Appendix~\ref{prf:mi}. This result is a types-based  analogue of the {\em (conditional) joint typicality lemma} used extensively  for channel coding problems in~\cite{elgamal}. 
\begin{lemma}[Atypicality of Empirical Mutual Information] \label{lem:emil}
Fix a joint distribution $p_{U,X,Y} = p_U p_{X|U} p_{Y|U}$, i.e., $X-U-Y$ form a Markov chain in that order. Let $(U^n, X^n, Y^n) \sim\prod_{k=1}^n p_{U,X,Y} (u_k, x_k, y_k)$ so $X^n - U^n - Y^n$. Then for any $t>0$ and any $n\in\bbN$, the empirical mutual information $\hatI( X^n\wedge Y^n|U^n)$ satisfies
\begin{equation}
\rvP( \hatI( X^n\wedge Y^n|U^n) \ge t)\le (n+1)^{ |\calX| |\calY| |\calU| } 2^{-nt} .
\end{equation}
\end{lemma}
Now we use this lemma to bound $\rvP(\calE_2)$. By the union bound and the symmetry in the generation of the codewords,
\begin{align}
\rvP(\calE_2) & \le \sum_{\tilm_2\ne 1} \rvP(  \hat{\bI}(Q^n,X_1^n(\tilm_2),  X_2^n(1), Y^n)\in\scT(\bR,\delta_n)) \label{eqn:union_bd_mac} \\
&=  (\ceil{2^{nR_1}}-1) \rvP(  \hat{\bI}(Q^n,X_1^n(2),  X_2^n(1), Y^n)\in\scT(\bR,\delta_n)) \\
&\le  2^{nR_1} \rvP( \hatI(X_1^n(2)\wedge  Y^n    |X_2^n(1), Q^n) \ge R_1+\delta_n )\label{eqn:incl} \\
&\le 2^{nR_1} \rvP( \hatI(X_1^n(2)\wedge (X_2^n(1),  Y^n)    |Q^n) \ge R_1+\delta_n ) \label{eqn:cond_mi} \\
&\le(n+1)^{|\calQ| |\calX_1| |\calX_2| |\calY|  }   2^{nR_1}   2^{-n(R_1+\delta_n)} \label{eqn:apply_aemi}
\end{align}
where~\eqref{eqn:incl} follows from the inclusion $\{\hat{\bI}(Q^n,X_1^n(2),  X_2^n(1), Y^n)  \in\scT(\bR,\delta_n)\}\subset \{\hatI(X_1^n(2)\wedge  Y^n    |X_2^n(1), Q^n) \ge R_1+\delta_n\}$ and $\ceil{t}-1\le t$,    \eqref{eqn:cond_mi} follows from the fact that $I(\tilX_1 ; \tilY| \tilX_2,\tilQ) \le I(\tilX_1 ; \tilX_2,\tilY| \tilQ) $ for any four  random variables $\tilQ,\tilX_1 , \tilX_2,\tilY$.  For \eqref{eqn:apply_aemi}, we applied the atypicality of empirical mutual information lemma with the  following identifications: $t\leftarrow R_1+\delta_n$, $U\leftarrow Q$, $X \leftarrow X_1$ and $Y\leftarrow (X_2, Y)$. Note that for $\tilm_1\ne 1$, $X_1^n(\tilm_2)$ is conditionally independent of $(X_2^n(1) ,Y^n)$ given $Q^n$ so the lemma   applies.  Using the definition of $\delta_n$, we have 
\begin{equation}
\rvP(\calE_2)\le \frac{1}{\sqrt{n+1}}.\label{eqn:e2_mac}
\end{equation}
Similarly, $\rvP(\calE_3)\le \frac{1}{\sqrt{n+1}}$ and $\rvP(\calE_3)\le\frac{1}{\sqrt{n+1}}$. Uniting~\eqref{eqn:e1_mac} and \eqref{eqn:e2_mac} reveals that the average probability of error of the random code ensemble is bounded above as $\rvP(\calE) \le\sum_{i=1}^4 \rvP(\calE_i)\le \epsilon$. Therefore, there must exist a code whose average probability of error for the DM-MAC $W$ is bounded above by $\epsilon$ as desired. \end{proof}

\subsubsection{Cardinality Bounds} \label{sec:card}
\textcolor{white}{..} 
\begin{proof}
We now argue that $|\calQ|$ can be restricted to be no greater than $9$. The following $9$ functionals  are continuous in $p_{X_1, X_2|Q} :=  p_{X_1|Q}  p_{X_2|Q}$: Three mutual information quantities $I(X_1;Y|X_2,Q)$, $I(X_2;Y|X_1,Q)$ and $I(X_1, X_2;Y|Q)$, three variances  on the diagonals of $\bV(  p_Q, p_{X_1|Q}, p_{X_2|Q}, W  )$ and three covariances in the strict upper triangular part  of  $\bV(  p_Q, p_{X_1|Q}, p_{X_2|Q}, W  )$.    By the support lemma~\cite[Lemma 3.4]{Csi97} (or Eggleston's theorem), there exists a discrete random variable $Q'$, whose support has cardinality $|\calQ'|\le 9$, that preserves these $9$ continuous functionals in $p_{X_1, X_2|Q}$. Thus, the inner bound is preserved if the auxiliary time-sharing random variable $\calQ$ is restricted to have cardinality $9$.
 \end{proof}

\subsubsection{Extension to Arbitrary Alphabets} \label{sec:universal_mac}
In place of the universal decoding rule in~\eqref{eqn:decoding_rule_mac}, one could use a non-universal one  by comparing the  normalized information density vector    (instead of the empirical mutual information vector)    with the rate vector, i.e., 
\begin{equation}
 \frac{1}{n} \begin{bmatrix}    i(x_1^n(\hatm_1);y^n|x_2^n(\hatm_2),q^n) \\i(x_2^n(\hatm_2);y^n|x_1^n(\hatm_1),q^n) \\ i(x_1^n(\hatm_1), x_2^n(\hatm_2);y^n|q^n) \end{bmatrix}  \ge\bR+\delta_n\bone ,\label{eqn:decoding_rule_mac2}
\end{equation}
where $ i(x_1^n(\hatm_1);y^n|x_2^n(\hatm_2),q^n) := \log [ W^n(y^n|x_1^n(\hatm_1), x_2^n(\hatm_2)) / p_{Y^n|X_2^n,Q^n}(y^n| x_2^n(\hatm_2), q^n)]$ and similarly for the other two information densities. For this non-universal decoding strategy, Taylor expansion as in the proof of the vector rate redundancy theorem [cf.~\eqref{eqn:taylor_rrt}]  would not be required because the above criterion can be written as a normalized sum of i.i.d.\ random vectors. One can verify that a simpler version of the vector rate redundancy theorem can be proved for the decoding rule in \eqref{eqn:decoding_rule_mac2} if the channel and input distributions are such that the third moment is bounded. In addition, we   need to generalize the atypicality of empirical mutual information lemma for the steps in \eqref{eqn:union_bd_mac}--\eqref{eqn:e2_mac} to hold. This can be done using standard Chernoff bounding techniques.  Indeed, if $X- U-Y$ form a Markov chain and $(U^n, X^n, Y^n)\sim \prod_{k=1}^n p_{U,X,Y}(u_k, x_k, y_k)$, then 
\begin{equation}
\rvP\left( \frac{1}{n}\log\frac{p_{Y^n|X^n }(Y^n|X^n )}{p_{Y^n|U^n}(Y^n|U^n)}\ge t\right)\le 2^{-nt},
\end{equation}
for every $t\ge 0$. This is the analogue of Lemma~\ref{lem:emil}. Finally, note that we have used i.i.d.\ codebooks for simplicity. For the    AWGN-MAC, a codebook containing codewords of \emph{exact} power may result in a smaller dispersion. See~\cite{PPV10, Hayashi09} for the single-user case.

\subsection{Proof  of the Global Dispersion   for the DM-ABC  (Theorem~\ref{thm:bc})} \label{sec:prf_abc}
We now present the  proof of Theorem~\ref{thm:bc} on the $(n,\epsilon)$-capacity region for the DM-ABC. Conceptually, it is simple --- it uses the superposition coding technique~\cite{cover} and the vector rate redundancy theorem.  

\subsubsection{Achievability} 
\textcolor{white}{..} 
\begin{proof} Fix  an input alphabet $\calU$ and also an input distribution $p_{U,X} \in\scP(\calU\times\calX)$. This input distribution induces the distributions $p_U$ and $p_{X|U}$.  Also fix a pair of achievable rates $(R_1, R_2)$ belonging  to the region $\scR(n,\epsilon; p_{U,X}  ) $ (Definition~\ref{def:reg_bc}). \\
{\em Codebook Generation} Randomly and independently generate $2^{nR_2}$ cloud centers $u^n( m_2) \sim \prod_{k=1}^n p_U(u_k), m_2 \in [2^{nR_2}]$. For every $m_2$, randomly and conditionally independently generate $2^{nR_1}$  satellite codewords $x^n(m_1, m_2) \sim \prod_{k=1}^n p_{X|U} (x_k | u_k(m_2)), m_1 \in [2^{nR_1}]$. The codebooks consisting of the $u^n$ and $x^n$ codewords are revealed to the encoder and the two decoders. \\
{\em Encoding}: Given $(m_1, m_2) \in [2^{nR_1}]\times [2^{nR_2}]$, the encoder transmits $x^n(m_1, m_2)$. \\
{\em Decoding}: Decoder 2 only has to decode the common message $m_2$. When it receives $y_2^n \in\calY_2^n$, it finds the unique $\check{m}_2 \in [2^{nR_2}]$ such that 
\begin{equation}
\hatI ( u^n(\check{m}_2 ) \wedge y_2^n) ) \ge R_2 +\delta_n, \label{eqn:dec2}
\end{equation}
where  the sequence $\delta_n :=(|\calU||\calX|  \max \{ |\calY_1| , |\calY_2| \} +\frac{1}{2})\frac{\log (n+1)}{n}$.  If there is no such message or there is not a unique one, declare a decoding error. Decoder 1 has to decode both the common message $m_2$ and its own message $m_1$. When it receives $y_1^n\in\calY_1^n$, it finds the unique pair $( \hatm_1,\hatm_2) \in [2^{nR_1}]  \times [2^{nR_2}] $ such that 
\begin{equation}
\hatbJ ( u^n (\hatm_2) , x^n (\hatm_1, \hatm_2) , y_1^n) := \begin{bmatrix}  \hatI ( x^n (\hatm_1, \hatm_2) \wedge  y_1^n| u^n (\hatm_2)  ) \\ \hatI ( x^n (\hatm_1, \hatm_2) \wedge  y_1^n) \end{bmatrix} \ge \begin{bmatrix}  R_1 \\R_1+R_2 \end{bmatrix}+\delta_n\bone. \label{eqn:dec1}
\end{equation}
If there is no such message pair  or there is not a unique one, again declare a decoding error. For convenience in stating the error events, we use the notation $\scT(R_1, R_2 ,\delta_n) := \{ \bz \in \bbR^2: z_1 \ge R_1+\delta_n, z_2\ge R_1+R_2 +\delta_n\}$.  We remind the reader that the notation $ \hatI ( x^n (\hatm_1, \hatm_2) \wedge  y_1^n| u^n (\hatm_2)  )$ denotes the conditional mutual information $I( \tilX;\tilY|\tilU)$ where $(\tilU, \tilX,\tilY)$ is a dummy random variable with distribution, an $n$-type, $P_{u^n (\hatm_2) ,  x^n (\hatm_1, \hatm_2) ,  y_1^n}$. \\
{\em Analysis of Error Probability}: By symmetry and the random codebook generation, we can assume that $(M_1,M_2)=(1,1)$. The error event at decoder 2, namely $\calE_2:= \{\check{M}_2 \ne M_2\}$,  can be decomposed into the following 2 events:
\begin{align}
\calE_{2,1} &:= \{ \hatI ( U^n(1 ) \wedge Y_2^n)    \le R_2 +\delta_n\} \label{eqn:E21}\\
\calE_{2,2} &:= \{\exists \, \tilde{m}_2\ne 1 :\hatI ( U^n(\tilde{m}_2 ) \wedge Y_2^n  )  \ge  R_2 +\delta_n\} 
\end{align}
Decoder 1's error event, namely $\calE_1:= \{ \hatM_1 \ne M_1\}\cup \{ \hatM_2 \ne M_2\}$, can be decomposed into the following $3$ events:
\begin{align}
\calE_{1,1} &:= \{ \hatbJ ( U^n (1) , X^n (1, 1) , Y_1^n)  \notin\scT(R_1, R_2,\delta_n)\} \label{eqn:E12} \\
\calE_{1,2} &:= \{\exists \, \tilde{m}_1\ne 1 :  \hatbJ ( U^n (1)  , X^n (\tilde{m}_1  , 1   )  , Y_1^n)   \in\scT(R_1, R_2,\delta_n) \} \\
\calE_{1,3} &:= \{\exists \,\tilde{m}_1\ne 1 ,  \tilde{m}_2\ne 1 :\hatbJ ( U^n (\tilde{m}_2) , X^n (\tilde{m}_1, \tilde{m}_2  ) , Y_1^n)    \in\scT(R_1, R_2,\delta_n) \} 
\end{align}
The vector $\hatbJ  (u^n, x^n , y_1^n )$ is defined in~\eqref{eqn:dec1}.  Clearly the average error probability  for the ABC defined in~\eqref{eqn:bc1}  can be bounded above as 
\begin{equation}
\Pen\le \rvP(\calE_{2,1} \cup \calE_{1,1} ) + \rvP(\calE_{2,2})+ \rvP( \calE_{1,2}) + \rvP(\calE_{1,3} ). \label{eqn:sum-probs}
\end{equation}
Note that in contrast to the DM-MAC, we bound the probability of the union $\calE_{2,1} \cup \calE_{1,1}$ instead of bounding the probabilities of the  constituent events separately. This is an important distinction. By doing so, we can use the  vector rate redundancy theorem on an empirical mutual  information vector of length-$3$. See \eqref{eqn:info_bc} below.    We   bound the first term in \eqref{eqn:sum-probs}, which can be written as 
\begin{equation}
\rvP(\calE_{2,1} \cup \calE_{1,1} ) = 1-\rvP(  \hatbI ( U^n (1) , X^n (1,1) , Y_1^n, Y_2^n)   \ge \bR +\delta_n\bone) ,\label{eqn:union_pr}
\end{equation}
where the length-$3$ empirical mutual information vector  is defined as
\begin{equation} 
\hatbI ( U^n   , X^n   , Y_1^n, Y_2^n)  := \begin{bmatrix} \hatI ( X^n\wedge  Y_1^n| U^n  ) \\ \hatI ( U^n  \wedge Y_2^n)    \\ \hatI ( X^n\wedge  Y_1^n   ) \end{bmatrix}  . \label{eqn:info_bc}
\end{equation}
Using the fact that $(R_1, R_2) \in \scR(n,\epsilon; p_{U,X}  )$,  we can rewrite \eqref{eqn:union_pr} as
\begin{equation}
\rvP( (\calE_{2,1} \cup \calE_{1,1} )^c) = \rvP \left(  \hatbI ( U^n (1) , X^n (1,1) , Y_1^n, Y_2^n)   \ge \bI( p_{U,X} , W)+\frac{\bz}{\sqrt{n}}-a_n\bone +\delta_n\bone \right), 
\end{equation}
where from the definition of $\scS(\bV, \epsilon)$ in \eqref{eqn:SVset}, $\bz \in \bbR^3$ is a vector satisfying $\rvP(\bZ\ge \bz)\ge 1-\epsilon$ and $\bZ \sim \calN(\bzero, \bV)$. The sequence $a_n  = \frac{\nu\log n}{n}$ with $\nu$ given in Definition~\ref{def:reg_bc}. Now we again invoke the vector rate redundancy theorem (Theorem~\ref{thm:vrrt}) with the following identifications: random variable $X \leftarrow (U, X, Y_1, Y_2)$,  smooth function $\bg ( p_{U, X} W)\leftarrow \bI ( p_{U,X} , W)$, evaluation vector   $\bz\leftarrow \bz$ and  sequence $b_n\leftarrow a_n-\delta_n$. 
%Then if the coefficient of $a_n$ is larger than that of $\delta_n$, say $\nu=|\calU ||\calX| \max\{ |\calY_1|, |\calY_2| \} +1$ as in \eqref{eqn:ach_bc},  $a_n-\delta_n$ is a positive sequence of order $\Theta(\frac{\log n}{n})$, satisfying~\eqref{eqn:seq_an}. Furthermore, the third moment $\xi_{\mathrm{ABC}} := \rvE[\|\bi( U,X,  Y_1, Y_2  ) - \bI (   p_{U,X}, W ) \|_2^3 ]$  is uniformly bounded as shown in \eqref{eqn:BC_bound} in Appendix~\ref{app:finite_third}.
 Hence, going through the same argument as for the MAC (see \eqref{eqn:apply_vrrt0}--\eqref{eqn:apply_vrrt}),
\begin{equation}
\rvP( (\calE_{2,1} \cup \calE_{1,1} )^c) \ge 1-\epsilon+O\left( \frac{\log n}{\sqrt{n}}\right). \label{eqn:apply_vrrt_bc}
\end{equation}
The rest of the error events can be bounded using the atypicality of empirical mutual information lemma (Lemma~\ref{lem:emil}). Since the calculations are similar, we focus solely on $\calE_{1,2}$. For this event, we have 
\begin{align}
\rvP(\calE_{1,2} )  &\le \sum_{\tilde{m}_1 \ne 1} \rvP( \hatbJ ( U^n (1 )  , X^n (  \tilde{m}_1  , 1) , Y_1^n)   \in\scT(R_1, R_2,\delta_n) ) \\
&\le (\ceil{2^{nR_2}}-1) \rvP( \hatbJ ( U^n (1 )  , X^n (2, 1 ), Y_1^n)   \in\scT(R_1, R_2,\delta_n) ) \\
&\le 2^{nR_2} \rvP( \hatI (   X^n (2, 1 ) \wedge Y_1^n | U^n(1) )   \ge R_1+\delta_n) \label{eqn:bc_apply_lem1} \\
&\le (n+1)^{|\calU| |\calX| |\calY_1| }  2^{nR_2}  2^{-n(R_2+\delta_n)} .\label{eqn:bc_apply_lem2}
\end{align}
 The reasoning for each of these steps is similar to that for the DM-MAC. See steps \eqref{eqn:union_bd_mac} to \eqref{eqn:apply_aemi}. The crucial realization to get from \eqref{eqn:bc_apply_lem1}  to \eqref{eqn:bc_apply_lem2} via the use of the atypicality of empirical mutual information lemma is that for $\tilde{m}_1\ne 1$, the satellite codeword $X^n (  \tilde{m}_1  , 1)$ is conditionally independent of $Y_1^n$ given  the cloud center $U^n(1)$.  By the choice of $\delta_n$ introduced at the decoding step, we have 
\begin{equation}
 \rvP(\calE_{1,2} ) \le \frac{1}{\sqrt{  n+1}}.
\end{equation} 
Similarly, $ \rvP(\calE_{2,2} ) \le \frac{1}{\sqrt{n+1}}$ and $ \rvP(\calE_{2,3} ) \le \frac{1}{\sqrt{n+1}}$. This, combined with~\eqref{eqn:sum-probs} and~\eqref{eqn:apply_vrrt_bc}, shows that the average error probability for the DM-ABC, defined in \eqref{eqn:bc1}, is no greater than $\epsilon$. Hence, there exists a deterministic code whose average error probability is no greater than $\epsilon$ as desired. \end{proof}
\subsubsection{Cardinality Bounds}
\textcolor{white}{..} 
\begin{proof}
The bound on $|\calU|$ can be argued in the same way as we did for the DM-MAC in Section~\ref{sec:card}. We need $|\calX|-1$ elements to preserve $p_{X }(x ), x \in\{0,\ldots, |\calX|-2\}$ and $7$ additional elements to preserve the {\em two} mutual information quantities  $I(U;Y_2)$ and $I(X;Y_1|U)$, {\em two} variances along the diagonals of  $\bV(p_{U,X}, W)$, i.e.,  $\var(\log [ W_1(Y_1|X)/p_{Y_1|U}(Y_1|U)])$ and $\var(\log [ p_{Y_2|U}(Y_2|U)/p_{Y_2}(Y_2)])$  and  three covariances in the off-diagonal positions in $\bV(p_{U,X}, W)$.  Note that $I(X;Y_1)$ and $\var( \log [ W_1(Y_1|X)/ p_X(X)])$ are automatically preserved given that we have preserved $p_X(x)$ and they do not depend on $U$. Hence,    $|\calU|\le |\calX|+ 6$. % Also see \cite[Appendix C]{elgamal}.% We omit the details. 
\end{proof}

\appendices
\numberwithin{equation}{section}
\section{Proofs of the  Dispersions for  Slepian-Wolf} \label{app:slice}

Theorems~\ref{thm:local_disp} and~\ref{thm:sum_rate_disp} are both consequences of the following general-purpose Lemma. Their proofs will follow the proof of the Lemma. For any $\bv\in\bbR^k$ and  subset  $\calT\subset [k]$, $\bv_{\calT}$ denotes the subvector with elements indexed by $\calT$.
\begin{lemma}\label{lemma:gen_disp}
Fix an integer $k$, a non-negative column  vector $\bc\in\bbR^k$ (i.e.,  $\bc\ge \bzero$), a matrix $\bA\in\bbR^{2\times k}$, and constants $\bar{R}_1$ and $\bar{R}_2$. Denote the rows of $\bA$ by $\ba_1$ and $\ba_2$. Let 
\begin{equation}
\begin{array}{ll}
f^*:=\text{minimize}_{\bs} & \bc^T\bs\\
\qquad\,\,\,\text{subject to} & (\bar{R}_1+\ba_1\bs,\bar{R}_2+\ba_2\bs)\in\scR_{\mathrm{SW}}^*  
\end{array}
\end{equation}
where $\scR_{\mathrm{SW}}^*$  is the  (asymptotic) SW rate region  given in~\eqref{eqn:sw_reg}. Let $\scD \subset\bbR^k$ be the set of asymptotically achievable vectors $\bs$ with $\bc^T\bs=f^*$. For every $\bs\in\scD $, define
\begin{equation}\label{eqn:Nset_def}
\calN_{\bs}:=\big\{j\in\{1,2,3\}:\bar{R}_j+\ba_j\bs=H_j\big\}
\end{equation}
where $\bar{R}_3:=\bar{R}_1+\bar{R}_2$, $\ba_3:=\ba_1+\ba_2$, and $\bH= [H_1, H_2, H_3]^T$ as defined in \eqref{eqn:entropy_vec}. Let  $\bZ:=(Z_1,Z_2, Z_3)\sim\calN(\bzero,\bV)$. Let $\bu_{\bs}\in\bbR^k$ be a solution to
\begin{equation}\label{eqn:Udef}
\begin{array}{ll}
\text{minimize}_{\bu  } & \bc^T\bu\\
\text{subject to} & \rvP(\bZ_{\calN_{\bs}}\le \bA_{\calN_{\bs}} \bu)=1-\epsilon.
\end{array}
\end{equation}
Let $\bs^*_n\in\bbR^k$ be a vector minimizing $\bc^T \bs^*_n$ subject to $(R_{1,n},R_{2,n})$ being $(n,\epsilon)$-achievable for some $(R_{1,n},R_{2,n})$ satisfying
\begin{align}
R_{1,n}&\le \bar{R}_1+\ba_1\bs^*_n\\
R_{2,n}&\le \bar{R}_2+\ba_2\bs^*_n.
\end{align}
% $(\bar{R_1}+\ba_1\bs^*_n,\bar{R_2}+\ba_2\bs^*_n)$ being $(n,\epsilon)$-achievable. 
For any $\bs\in\scD $,
\begin{equation}\label{eqn:Sstar}
\bc^T\bs^*_n=f^*+\frac{\bc^T \bu_{\bs}}{\sqrt{n}}+O\left(\frac{\log n}{n}\right).
\end{equation}
\end{lemma}

\begin{IEEEproof}
Choose an arbitrary $\bs\in\scD$. We will show \eqref{eqn:Sstar} for this $\bs$. Let $\bs_n$ minimize $\bc^T \bs_n$ subject to 
\begin{equation}\label{eqn:Sn_def}
\begin{bmatrix} \bar{R}_1+\ba_1 \bs_n-H(X_1|X_2)\\ \bar{R}_2+\ba_2 \bs_n-H(X_2|X_1)\\ \bar{R}_3+\ba_3\bs_n-H(X_1,X_2)\end{bmatrix}\in \frac{1}{\sqrt{n}}\scS(\bV,\epsilon).
\end{equation}
By Theorem~\ref{thm:orr},
\begin{equation}\label{eqn:lognn_diff}
\left|\bc^T\bs^*_n-\bc^T\bs_n \right|\in O\left(\frac{\log n}{n}\right).
\end{equation}
First we find a lower bound on $\bc^T\bs_n$. From the definition of $\scS(\bV,\epsilon)$, \eqref{eqn:Sn_def} is equivalent to
\begin{equation}
\rvP\left(\begin{bmatrix} Z_1\\Z_2\\Z_3\end{bmatrix}   \le \sqrt{n}\begin{bmatrix} \bar{R}_1+\ba_1 \bs_n-H(X_1|X_2)\\ \bar{R}_2+\ba_2 \bs_n-H(X_2|X_1)\\ \bar{R}_3+\ba_3\bs_n-H(X_1,X_2)\end{bmatrix} \right)\ge 1-\epsilon \label{eqn:condF}
\end{equation}
The condition \eqref{eqn:condF} can be rewritten
\begin{equation}
\rvP(\calA_{1,n} \cap\calA_{2,n}\cap\calA_{3,n})\ge 1-\epsilon,\label{eqn:intersect3}
\end{equation}
where we have defined events
\begin{equation}
\calA_{j,n} :=\big\{Z_j \le \sqrt{n}(\bar{R}_j+\ba_j\bs_n-H_j)\big\}
\end{equation}
for $j=1,2,3$. Contiuing from \eqref{eqn:intersect3}, 
\begin{align}
1-\epsilon
&\le \rvP(\calA_{1,n}\cap\calA_{2,n}\cap\calA_{3,n})\\
&\le\rvP\left(\bigcap_{j\in\calN_\bs}\calA_{j,n}\right)\label{eqn:A_intersect}\\
&=\rvP\big(\bZ_{\calN_\bs}\le \sqrt{n}(\bar{\bR}_{\calN_\bs}+\bA_{\calN_\bs} \bs_{n}-\bH_{\calN_\bs})\big)\\
&=\rvP\big(\bZ_{\calN_\bs}\le \sqrt{n}\bA_{\calN_\bs}(\bs_n-\bs)\big)\label{eqn:prob_rewrite}
\end{align}
where \eqref{eqn:prob_rewrite} holds by the definition of $\calN_\bs$. Let $\underline{\bs}_n$ be a solution to
\begin{equation}\label{eqn:barS_def}
\begin{array}{ll}
\text{minimize}_{\bs'}& \bc^T\bs' \\
\text{subject to}& \rvP\big(\bZ_{\calN_\bs}\le \sqrt{n}\bA_{\calN_\bs}(\bs'-\bs)\big)\ge 1-\epsilon
\end{array}
\end{equation}
Note that we may assume equality in the constraint in \eqref{eqn:barS_def} because the probability is nondecreasing in each element of $\bs'$, and $\bc\ge\bzero$. Hence
\begin{align}
\bc^T\bs_n&\ge\bc^T\underline{\bs}_n\label{eqn:lower_bd_1}\\
&=\bc^T\left(\bs+\frac{\bu_\bs}{\sqrt{n}}\right)\label{eqn:lower_bd_2}\\
&=f^*+\frac{\bc^T\bu_{\bs}}{\sqrt{n}}\label{eqn:lower_bd_3}
\end{align}
where \eqref{eqn:lower_bd_1} follows from \eqref{eqn:prob_rewrite} and the definition of $\bs_n$, \eqref{eqn:lower_bd_2} follows from the definition of $\bu_\bs$ in \eqref{eqn:Udef}, and \eqref{eqn:lower_bd_3} holds because $\bc^T\bs=f^*$ for all $\bs\in\scD$.

Now we upper bound $\bc^T\bs_n$. For $j\in\{1,2,3\}$, let
\begin{equation}
\delta_j:=\ba_j\bs+\bar{R}_j-H_j.
\end{equation}
Since $(\bar{R}_1+\ba_1\bs,\bar{R}_2+\ba_2\bs)\in\scR_{\mathrm{SW}}^*$, $\delta_j\ge 0$ for all $j$. Moreover, $\delta_j>0$ for $j\in\calN_\bs^c$. With hindsight, we define the following exponentially decaying sequences:
\begin{align}
\tau_{j,n}&:=\frac{1}{2}\exp\left( - \frac{n}{2 [\bV]_{j,j}} (\delta_j/2)^2\right),\quad\text{for }j\in\calN_\bs^c\\
\tau_n&:=\sum_{j\in\calN_\bs^c}\tau_{j,n}.
\end{align}
Now let
\begin{equation}\label{eqn:sn_bar_def}
\bar{\bs}_n:=\bs+\frac{\bar{\bu}}{\sqrt{n}}
\end{equation}
where $\bar{\bu}$ is a solution to
\begin{equation}\label{eqn:U_bar_def}
\begin{array}{ll}
\text{minimize}_{\bu} & \bc^T\bu\\
\text{subject to} & \rvP(\bZ_{\calN_\bs}\le \bA_{\calN_\bs} \bu)=1-\epsilon+\tau_n.
\end{array}
\end{equation}
 Note that by continuity and differentiability of the Gaussian cumulative density function, $\bc^T\bar{\bu}$ and $\bc^T\bu_{\bs}$ differ by an exponentially decaying sequence that we denote $\tau'_n$. We claim that $\bar{\bs}_n$ satisfies condition \eqref{eqn:Sn_def}. For all $j\in\calN_\bs^c$ and sufficiently large $n$,
\begin{equation}\label{eqn:bar_bs_j}
\ba_j\bar{\bs}_n+\bar{R}_j-H_j\ge\delta_j/2.
\end{equation}
Define the events
\begin{equation}
\bar{\calA}_{j,n} :=\big\{Z_j \le \sqrt{n}(\bar{R}_j+\ba_j\bar{\bs}_n-H_j)\big\}
\end{equation}
for $j=1,2,3$. We claim that $\rvP(\bar{\calA}_{j,n}^c)$ is exponentially decaying for $j\in\calN_\bs^c$. Indeed,
\begin{align}
\rvP(\bar{\calA}_{j,n}^c)
&=\rmQ\left(\sqrt{\frac{n}{[\bV]_{j,j}}}(\bar{R}_j+\ba_j\bar{\bs}_{n}-H_j)\right)\\
&\le \frac{1}{2}\exp\left( - \frac{n}{2 [\bV]_{j,j}} (\bar{R}_j+\ba_j\bar{\bs}_n-H_j)^2\right)\label{eqn:chernoff}\\
&\le \frac{1}{2}\exp\left( - \frac{n}{2 [\bV]_{j,j}} (\delta_j/2)^2\right)\label{eqn:chernoff2}\\
&=\tau_{j,n}\label{eqn:chernoff3}
\end{align}
where the inequality in \eqref{eqn:chernoff} is due to the Chernoff bound for the $\rmQ$-function, i.e., $\rmQ(t)\le \frac{1}{2}\exp(- \frac{t^2}{2})$ for all $t\ge 0$, and \eqref{eqn:chernoff2} holds by \eqref{eqn:bar_bs_j} for sufficiently large $n$. Now we have
\begin{align}
\rvP(\bar{\calA}_{1,n}\cap\bar{\calA}_{2,n}\cap\bar{\calA}_{3,n})
&\ge \rvP\left(\bigcap_{j\in\calN_\bs}\bar{\calA}_{j,n}\right)-\sum_{j\in\calN_\bs^c}\rvP(\bar{\calA}_{j,n}^c)\\
&\ge \rvP\left(\bigcap_{j\in\calN_\bs}\bar{\calA}_{j,n}\right)-\tau_n\label{eqn:Abar_intersect2}\\
&=\rvP\big(\bZ_{\calN_\bs}\le \sqrt{n}\bA_{\calN_\bs}(\bar{\bs}_n-\bs)\big)-\tau_n\label{eqn:Abar_intersect3}\\
&=\rvP(\bZ_{\calN_\bs}\le\bA_{\calN_\bs}\bar{\bu})-\tau_n\label{eqn:Abar_intersect4}\\
&=1-\epsilon\label{eqn:Abar_intersect5}
\end{align}
where \eqref{eqn:Abar_intersect2} follows from \eqref{eqn:chernoff3} and the definition of $\tau_n$, \eqref{eqn:Abar_intersect3} follows by the same reasoning as \eqref{eqn:A_intersect}--\eqref{eqn:prob_rewrite}, \eqref{eqn:Abar_intersect4} follows from \eqref{eqn:sn_bar_def}, and \eqref{eqn:Abar_intersect5} follows from \eqref{eqn:U_bar_def}. Therefore $\bar{\bs}_n$ satisfies \eqref{eqn:intersect3} and equivalently \eqref{eqn:Sn_def}, so for sufficiently large $n$
\begin{equation}\label{eqn:bar_bound}
\bc^T\bs_n\le \bc^T\bar{\bs}_n=f^*+\frac{\bc^T \bar{\bu}}{\sqrt{n}}
\le f^*+\frac{\bc^T\bu_{\bs}}{\sqrt{n}}+\frac{\tau'_n}{\sqrt{n}}
\end{equation}
and recall $\tau'_n$ is exponentially decaying. Combining \eqref{eqn:lower_bd_3} and \eqref{eqn:bar_bound} with \eqref{eqn:lognn_diff} yields \eqref{eqn:Sstar}.
\end{IEEEproof}

\begin{IEEEproof}[Proof of Theorem~\ref{thm:local_disp}]
Fix $(R_1^*,R_2^*)$ and $\theta$. We particularize Lemma~\ref{lemma:gen_disp} by setting $k=1$, scalar $c=1$, matrix $\bA=[\cos\theta,\sin\theta]^T$, $\bar{R}_1=R^*_1$, and $\bar{R}_2=R^*_2$. It is clear that $f^*=0$ and $\scD =\{0\}$ (i.e., the only solution is $s=0$). The set $\calN_{s}$ defined in \eqref{eqn:Nset_def} will depend on which of the five cases $(R_1^*,R_2^*)$ falls into. Consider the first case: i.e. $R^*_1=H(X_1|X_2)$ and $R^*_2>H(X_2)$. Then $\calN_{s}=\{1\}$. By \eqref{eqn:Udef}, the scalar $u_s$ satisfies
\begin{equation}
\rvP(Z_1\le (\cos\theta) u_s)=1-\epsilon.
\end{equation}
This may also be written as
\begin{equation}
1-\rmQ\left(\frac{(\cos\theta) u_s}{\sqrt{[\bV]_{1,1}}}\right)=1-\epsilon.
\end{equation}
Hence
\begin{equation}
u_s=\frac{\sqrt{[\bV]_{1,1}}}{\cos\theta}\rmQ^{-1}(\epsilon).
\end{equation}
We now apply Lemma~\ref{lemma:gen_disp} to conclude
\begin{equation}
s^*_n=\sqrt{ \frac{[\bV]_{1,1}}{n}}\cdot\frac{1}{\cos\theta}\cdot\rmQ^{-1}(\epsilon)+O\left(\frac{\log n}{n}\right).
\end{equation}
That the dispersion is given by \eqref{eqn:subplota} is an immediate consequence given the assumption that $-\frac{\pi}{2}<\theta <\frac{\pi}{2}$ so $\cos\theta>0$. The local dispersions in \eqref{eqn:hori}--\eqref{eqn:subplotb} follow similarly.

We now consider a corner point; in particular, take $R^*_1=H(X_1|X_2)$ and $R^*_2=H(X_2)$. Then $\calN_s=\{1,3\}$, meaning $u_s$ satisfies
\begin{equation}
\rvP\big(Z_1\le (\cos\theta)u_s,\ Z_3\le(\cos\theta+\sin\theta)u_s\big)=1-\epsilon.
\end{equation}
This may be written
\begin{equation}\label{eqn:rho13_rewrite}
\Psi\left(\rho_{1,3};-\frac{(\cos\theta) u_s}{\sqrt{[\bV]_{1,1}}},-\frac{(\cos\theta+\sin\theta)u_s}{\sqrt{[\bV]_{3,3}}}\right)=1-\epsilon.
\end{equation}
The dispersion in \eqref{eqn:corner_pt} follows from applying Lemma~\ref{lemma:gen_disp}. The local  dispersion for the other corner point in \eqref{eqn:subplotc} follows similarly.
\end{IEEEproof}

\begin{IEEEproof}[Proof of Theorem~\ref{thm:sum_rate_disp}]
Fix $(\alpha,\beta)$. We particularize Lemma~\ref{lemma:gen_disp} be setting $k=2$, $\bc=[\alpha,\beta]^T$, $\bA=\bI_2$, and $\bar{R}_1=\bar{R}_2=0$. We have that $f^*=R^*_{\text{sum}}(\alpha,\beta)$ as given in \eqref{eqn:opt_sum_rate}. Consider first the case that $\alpha\ge\beta$. In this case, one asymptotic optimum is the corner point $\bs=(H(X_1|X_2),H(X_2))^T$. (This will not be the unique optimum if $\beta=0$ or $\alpha=\beta$, but Lemma~\ref{lemma:gen_disp} still applies.) Hence $\calN_{\bs}=\{1,3\}$, so $\bu_{\bs} =(u_1,u_2)$ is the solution to (rewriting the probability as in \eqref{eqn:rho13_rewrite})
\begin{equation}
\begin{array}{ll}
\text{minimize}_{ u_1, u_2}& \alpha u_1+\beta u_2\\
\text{subject to}& \displaystyle\Psi\left(\rho_{1,3};-\frac{u_1}{\sqrt{[\bV]_{1,1}}},-\frac{ u_1+u_2}{\sqrt{[\bV]_{3,3}}}\right)=1-\epsilon.
\end{array}
\end{equation}
Applying Lemma~\ref{lemma:gen_disp} gives \eqref{eqn:G_sol}--\eqref{eqn:G_constraint1}. For the case that $\beta\ge\alpha$ an identical argument leads to \eqref{eqn:G_constraint2}.

In the special cases that $\beta=0$, $\alpha=0$, or $\alpha=\beta$, there will be non-unique asymptotic optima; i.e. $\scD $ contains more than one element. For these cases, alternate choices for ${\bs}$ yield single-element $\calN_{\bs}$ sets. Application of Lemma~\ref{lemma:gen_disp} with this choice leads to the simpler expressions \eqref{eqn:Gsimp1}--\eqref{eqn:Gsimp3}. Still, Lemma~\ref{lemma:gen_disp} asserts that the resulting dispersions are the same as those given in \eqref{eqn:G_sol}--\eqref{eqn:G_constraint2}.
\end{IEEEproof}

\section{Proof of Corollary~\ref{cor:be}} \label{prf:be}

\begin{proof}
We use Theorem~\ref{thm:bentkus} to prove Corollary~\ref{cor:be}. Let $\bV=\bL\bL^T$ be the Cholesky decomposition of the matrix $\bV$, defined in \eqref{eqn:vdef}. The lower-triangular matrix $\bL \in \bbR^{d\times d}$ is the left Cholesky factor of $\bV$. Define the change of coordinates $\tilbU_k :=\bL \bU_k \in\bbR^d$ for all $k=1,\ldots, n$. Then, $\cov( \tilbU_k ) = \rvE [ (\bL \bU_k)(\bL \bU_k)^T ] = \bL \rvE [ \bU_k\bU_k^T ] \bL^T   =\bV$ because $\rvE [ \bU_k\bU_k^T ]=\bI$ by assumption. Substituting this  into \eqref{eqn:bent} yields 
\begin{equation}
\sup_{\scC \in   \frakC_d} \left| \rvP \left( \frac{1}{\sqrt{n}}\sum_{k=1}^n \tilbU_k \in \bL\scC  \right)-\rvP(\bL\bZ \in\bL \scC) \right|\le\frac{400 d^{1/4}\xi}{\sqrt{n}} . \label{eqn:bentL}
\end{equation}
Clearly, the family of convex, Borel   subsets  in $\bbR^d$, namely $\frakC_d$, remains closed under matrix multiplication, i.e., $\frakC_d = \bL\frakC_d$. Thus, \eqref{eqn:bentL} can be rewritten as
\begin{equation}
\sup_{\tilde{\scC} \in  \frakC_d} \left| \rvP\left( \frac{1}{\sqrt{n}}\sum_{k=1}^n \tilbU_k \in \tilde{\scC} \right)-\rvP(\tilbZ \in \tilde{\scC}) \right|\le\frac{400 d^{1/4}\xi}{\sqrt{n}} \label{eqn:bentL2} , 
\end{equation}
where $\tilde{\scC} = \bL\scC$ and $\tilbZ \sim\calN(\bzero,  \bV)$. Now, recall that $\xi = \rvE [ \| \bU_1 \|_2^3]$. We upper bound this quantity as follows: Replacing $\bU_1$ by $\bL^{-1} \tilbU_1$ yields
\begin{align}
 \xi &= \rvE \left[  \|\bL^{-1} \tilbU_1 \|_2^3\right] \\
 &= \rvE \left[ (\tilbU_1^T  \bL^{-T}\bL^{-1} \tilbU_1 )^{3/2}\right] \\
  &= \rvE \left[ (\tilbU_1^T  \bV^{-1} \tilbU_1 )^{3/2}\right] \\
    &\le  \lambda_{\max}(\bV^{-1})^{3/2} \rvE  \left[  (\tilbU_1^T  \tilbU_1 )^{3 /2}\right] \label{eqn:cf} \\
        &= \frac{1}{\lambda_{\min}(\bV)^{3/2}}  \rvE \left[  \|\tilbU_1\|^{3}_2 \right] ,  \label{eqn:final_be}
\end{align}
where \eqref{eqn:cf} is because $\by^T\bA\by\le \lambda_{\max}(\bA) \|\by\|_2^2$ for all  $\by$ and all $\bA\succ 0$.  The proof  is completed upon  the substitution of the upper bound in~\eqref{eqn:final_be} into~\eqref{eqn:bentL2} and the identification of the third moment  of $\tilbU_1$ namely, $\tilde{\xi} :=\rvE [  \|\tilbU_1\|^{3}_2 ]$.
\end{proof}
\section{Proof of Lemma~\ref{lem:simple_pr} }  \label{app:simple_pr}

\begin{proof}
Define the events $\calF := \{ \bG\ge \bv+  \phi   \bone  \}$ and $\calG := \{\bDelta > -  \phi  \bone\}$. Then,   $\calF\cap \calG\subset\{\bG+\bDelta\ge\bv\}$. As such 
\begin{align}
\rvP(\bG+\bDelta\ge\bv)&\ge \rvP(\calF\cap\calG)   \\
&=\rvP(\calF\setminus (\calF\cap\calG^c)) \\
&= \rvP(\calF)-\rvP( \calF\cap \calG^c)  \label{eqn:simple_pr4} \\
&\ge  \rvP(\calF)-\rvP(\calG^c)  .  \label{eqn:simple_prob2}
\end{align}
In addition, we have
\begin{equation}
\rvP(\calG^c)  = \rvP(\bDelta \le  -  \phi   \bone) \le \rvP(\|\bDelta\|_{\infty} \ge \phi ) .\label{eqn:simple_prob3}
\end{equation}
The combination of \eqref{eqn:simple_prob2} and \eqref{eqn:simple_prob3} yields \eqref{eqn:simple_prob1} as desired. \end{proof}
%\begin{align}
%\rvP(\bG+\bDelta\ge\bv)&\ge \rvP( \{ \bG\ge \bv+  t   \bone  \}\cap\{\bDelta\ge -  t   \bone\})\\
%&\ge\rvP(\bG\ge \bv+  t   \bone ) - \rvP( \|\bDelta\|_{\infty}\ge   t  ) 
%\end{align}

\section{Finiteness of Third Moments} \label{app:finite_third} 
In this appendix, we prove that the third moments are finite. For notation, see Sections~\ref{sec:defs_sw},  \ref{sec:defmac} and~\ref{sec:defbc}.
\begin{lemma} \label{lem:finite_third} 
For the SW, MAC and ABC problems, let the {\em third moments} be  defined as 
\begin{align}
\xi_{\mathrm{SW}} &:= \rvE \left[ \| \bh(X_1, X_2) -\bH(p_{X_1, X_2})\|_2^3 \right] \\
\xi_{\mathrm{MAC}}&:= \rvE\left[ \|\bi(Q, X_1, X_2, Y) - \bI (  p_Q, p_{X_1|Q}, p_{X_2|Q}, W) \|_2^3\right]\\
\xi_{\mathrm{ABC}}&:= \rvE\left[\|\bi( U,X,  Y_1, Y_2  ) - \bI (   p_{U,X}, W ) \|_2^3 \right] .
\end{align}
Then, all three quantities are uniformly bounded in terms of the cardinalities of the alphabets.
% More precisely, there exists functions $g_{\mathrm{SW}}$
%\begin{align}
%\xi_{\mathrm{SW}}& \le 5\sqrt{3}  \cdot  (|\calX_1|+|\calX_2|+|\calX_1||\calX_2|)  \label{eqn:SW_bound} \\
%\xi_{\mathrm{MAC}}& \le 15\sqrt{3} \cdot |\calY| \label{sec:MAC_bound}   \\
%\xi_{\mathrm{ABC}}& \le  5\sqrt{3} \cdot  (2|\calY_1| + |\calY_2| )     . \label{eqn:BC_bound}
%\end{align}
\end{lemma}

\begin{proof}
We will only prove the second assertion for $\xi_{\mathrm{MAC}}$. The other two assertions for the SW and ABC  follow {\em mutatis mutandis} and  essentially leverage on the fact that the ranges of the random variables are finite. The proof is based on \cite[Lemma~46]{PPV10}. 

For brevity, let $A_1$, $A_2$ and $A_3$  be the components of the random vector $\bi(Q, X_1, X_2, Y)$ defined in~\eqref{eqn:info_dens}. So for  example,  $A_1 := \log [W(Y|X_1, X_2)/p_{Y|X_2,Q} (Y|X_2,Q)]-I(X_1;Y|X_2, Q)$.   Because $a\mapsto a^{3/2}$ is convex,
\begin{align}
\xi_{\mathrm{MAC}}  &= \rvE\left[ (A_1^2+A_2^2+A_3^2)^{3/2}\right]   \\*
&\le\frac{1}{3}\sum_{t=1}^3\rvE \left[ \left(3A_t^2 \right)^{3/2}\right] \\*
&=\sqrt{3} \,  \sum_{t=1}^3  \rvE\left[ |A_t|^3\right] \label{eqn:cube_convex}
\end{align}
Subsequently, we   simplify notation by dropping the subscripts on the distributions, e.g., $p(y|x_2,q) :=p_{Y|X_2,Q}(y|x_2,q)$ [see \eqref{eqn:output1}]. Also we define the $\ell_q$-norm $\|A\|_q =  \rvE\left[ |A|^q \right]^{1/q}$ for any random variable $A$ and any $q\ge 1$.  We focus on the first term in the sum in~\eqref{eqn:cube_convex},  namely $\rvE [ |A_1|^3 ]=\|A_1\|_3^3$. The $\ell_3$-norm can be bounded as
\begin{align}
\|A_1\|_3 &= \left\| \log \frac{W(Y|X_1, X_2) }{p(Y|X_2, Q)}-I(X_1;Y|X_2,Q)\right\|_3\\
&\le \left\| \log \frac{W(Y|X_1, X_2) }{p(Y|X_2, Q)}\right\|_3+ I(X_1;Y|X_2,Q)\\
&\le \left\|   \log\frac{1}{W(Y|X_1, X_2) } \right\|_3 + \left\|\log\frac{1}{p(Y|X_2, Q)}\right\|_3+ \log |\calY|\\
&\le 2(4.1 |\calY| )^{1/3}+\log |\calY|, \label{eqn:apply_e_bd}
\end{align}
where \eqref{eqn:apply_e_bd} follows from the fact that $x\log^3\frac{1}{x}\le (3e^{-1}\log e)^3\le 4.1$ for all $x>0$. 
%\begin{align}
%\rvE \left[ A_1^3\right]&=  \rvE  \left[\left( \log \frac{W(Y|X_1, X_2) }{p(Y|X_2, Q)} - I(X_1;Y|X_2, Q) \right)^3 \right] \\
%&\le  \rvE \left[\left(\log \frac{W(Y|X_1, X_2) }{p(Y|X_2, Q)}\right)^3 \right] \label{eqn:cube_ineq}  \\
%&\le  \rvE \left[\left( \log \frac{1}{p(Y|X_2, Q)}  \right)^3 \right]  \label{eqn:Wless1}\\
%& = \sum_{q,x_2 } p(q) p(x_2|q)\sum_y p(y|x_2, q) \left( \log \frac{1}{p(y|x_2, q) }  \right)^3 ,  \label{eqn:card_bound2} 
%\end{align}
%where \eqref{eqn:cube_ineq} follows from the fact that $t\mapsto t^3$ is monotonically increasing and mutual information is non-negative. Inequality~\eqref{eqn:Wless1} follows because $W(y|x_1, x_2)\le 1$ for all $(x_1, x_2,y)\in\calX_1\times\calX_2\times\calY$.   Now by simple calculus,  the function $u\mapsto u (-\log u)^3$ is bounded above by $(\frac{3}{\ln 2})^3 \exp_2(- \frac{3}{\ln 2})<5$ for all $u \in [0,1]$. Hence,  \eqref{eqn:card_bound2} reduces to
%\begin{equation}
%\rvE \left[ A_1^3\right]\le 5|\calY|.
%\end{equation}
All the other terms can be bounded similarly.   This completes the proof.%This proves \eqref{eqn:cardX1X2}.
\end{proof}
\section{Proof of Lemma~\ref{lem:emil}}  \label{prf:mi}

\begin{proof}  For convenience, we introduce dummy random variables $(\tilU,\tilX, \tilY)$ distributed according to $P_{U^n, X^n, Y^n}$, the type of $(U^n ,X^n, Y^n)$. This means that $p_{\tilU,\tilX, \tilY}  = P_{U^n, X^n, Y^n}$.  Then, note that 
\begin{equation}
I(\tilX;\tilY|\tilU) = I(\tilX;\tilY|\tilU)-\rvE_{p_{\tilU,\tilX, \tilY}} \left[ \log \frac{p_{ X,Y|U}(X, Y|U)}{p_{ X|U}(X|U) p_{ Y|U}(Y|U)}\right] \label{eqn:subtract_emp}
\end{equation}
since $X-U-Y$ form a Markov chain in that order so  $p_{ X,Y|U}(x,y|u)/ (p_{ X|U}(x|u) p_{ Y|U}(y|u))=1$ for all $(x,y,u)\in\calX\times\calY\times\calU$.   Let $p_{\tilX, \tilY|\tilU} := p_{\tilU,\tilX, \tilY} / p_{\tilU} $ be the conditional type and  let $p_{\tilX |\tilU} $  and $p_{\tilY|\tilU} $ be the $\calX$- and $\calY$-marginals of $p_{\tilX, \tilY|\tilU}$ respectively. Now,  by expressing the mutual information $I(\tilX;\tilY|\tilU)$ as an expectation, we readily see that~\eqref{eqn:subtract_emp} simplifies as 
\begin{align}
I(\tilX;\tilY|\tilU) & = D(p_{\tilX, \tilY|\tilU}  || p_{X,Y|U} |  p_{\tilU})    - D(p_{\tilX|\tilU} || p_{X|U} |  p_{\tilU})- D( p_{ \tilY|\tilU}  || p_{Y|U} |   p_{\tilU} )  . 
  \label{eqn:cond_re} 
  \end{align} 
  Because conditional relative entropies in \eqref{eqn:cond_re}  are non-negative, 
  \begin{align}
I(\tilX;\tilY|\tilU) \le D(p_{\tilX, \tilY|\tilU}  || p_{X,Y|U} |  p_{\tilU}).\label{eqn:div_nonng}
\end{align}
To simplify notation, let $W:= p_{X,Y|U}$.  Fix $t>0$.  Now consider
\begin{align}
\rvP( I(\tilX;\tilY|\tilU)\ge t) &\le \rvP( D(p_{\tilX, \tilY|\tilU} || W | p_{\tilU})\ge t )\label{eqn:use_div_nonneg} \\
 & = \sum_{Q\in \scP_n(\calU)}\sum_{u^n\in \calT_{Q}}p_U^n(u^n) \sum_{\substack{V \in \scV_n(\calX\times\calY;Q) :\\ D(V||W | Q)\ge t}}    W^n( \calT_V(u^n)  |u^n)\label{eqn:type_Q} \\
  &\le\sum_{Q\in \scP_n(\calU)} \sum_{u^n\in \calT_{Q}}p_U^n(u^n) \sum_{\substack{V \in \scV_n(\calX\times\calY;Q) :\\ D(V||W | Q)\ge t}} 2^{-n D(V||W | Q ) } \label{eqn:prob_tc}\\
   &\le \sum_{Q\in \scP_n(\calU)}\sum_{u^n\in \calT_{Q}}p_U^n(u^n) (n+1)^{ |\calU| |\calX| |\calY|}  2^{-nt } \label{eqn:type_count} \\
   &= (n+1)^{ |\calU| |\calX| |\calY|}  2^{-nt } 
\end{align}
where in~\eqref{eqn:use_div_nonneg} we used the bound in~\eqref{eqn:div_nonng}. For~\eqref{eqn:type_Q}, we noted that the type of $u^n$  in the innermost sum is $P_{u^n}=Q$. In \eqref{eqn:prob_tc}, we used~\cite[Lemma 1.2.6]{Csi97} to upper bound the $W^n(\fndot|u^n)$-probability of a $V$-shell. In~\eqref{eqn:type_count}, we applied the Type Counting Lemma for conditional types~\cite[Eq.~(2.5.1)]{Csi97} which asserts that $|\scV_n(\calX\times\calY;Q)|\le (n+1)^{ |\calU| |\calX| |\calY|}$. This completes the proof. \end{proof}

\subsection*{Acknowledgements}
The authors acknowledge discussions with Y.~Polyanskiy,  S.~C.~Draper, L.~Zheng,  Y.~Kochman,  D.~Wang and J.~Sun. The first author thanks Prof.~A.~S.~Willsky and the members of the Stochastic Systems Group  for their  hospitality while this work was done during his visit to MIT from Oct 2011 to Jan 2012.%This work was performed during the first author's visit to the Stochastic Systems Group (SSG) at MIT from Oct 2011 to Jan 2012. He would like to thank Prof.~A.~S.~Willsky and the members of SSG for their  hospitality during the visit.

V.~Y.~F.~Tan is supported  by a fellowship from A*STAR, Singapore. O.~Kosut is supported by Shell Global Solutions International B.V.  

\bibliographystyle{IEEETran}
\bibliography{isitbib}

\begin{IEEEbiographynophoto}{{\bf Vincent~Y.~F.~Tan}}
 (S'07--M'11) received the B.A.\ and M.Eng.\ degrees in electrical and  information engineering from Sidney Sussex College, Cambridge University. He received the Ph.D.\ degree in electrical engineering and computer science (EECS) from the Massachusetts Institute of Technology (MIT). He was then a postdoctoral researcher in the Department of Electrical and Computer Engineering  (ECE)  at the University of Wisconsin-Madison and following that, a scientist  at the Institute for Infocomm Research (I$^2$R), Singapore. He is currently an assistant    professor in the Department of ECE   at the National University of Singapore.   His research interests include   information theory, detection and estimation, and learning and inference of graphical models.

Dr.\ Tan is a recipient of the 2005 Charles Lamb Prize, a Cambridge University Engineering Department prize awarded annually to the top candidate in Electrical and Information Engineering. He also received the 2011 MIT EECS Jin-Au Kong outstanding doctoral thesis prize. He is a member of the IEEE Machine Learning for Signal Processing (MLSP) Technical Committee.  
\end{IEEEbiographynophoto}

\begin{IEEEbiographynophoto}{{\bf Oliver Kosut}}
 (S'06--M'10) received B.S.\ degrees in electrical engineering and mathematics from the Massachusetts Institute of Technology, Cambridge, MA in 2004 and a Ph.D.\ degree in electrical and computer engineering from Cornell, Ithaca, NY in 2010.

He was a visiting student at University of California at Berkeley in 2008--9. He was a Postdoctoral Research Associate in the Laboratory for Information and Decision Systems at MIT, Cambridge, MA from 2010 to 2012. Since 2012 he has been an Assistant Professor at Arizona State University, Tempe, AZ. His research interests include network information theory, security, and power systems.\end{IEEEbiographynophoto}

\end{document}